\documentclass[journal,12pt,onecolumn,draftclsnofoot,]{IEEEtran}
\ifCLASSINFOpdf
\else
\fi

\usepackage{microtype}
\usepackage{graphicx}
\usepackage{subfigure}
\usepackage{booktabs,pifont} % for professional tables

\usepackage{empheq}

\DeclarePairedDelimiter\abs{\lvert}{\rvert}
\usepackage{pgfplots}
\pgfplotsset{compat=1.12}
\usetikzlibrary{shapes,arrows}
\usetikzlibrary{positioning}
\usepackage{tikz}
\usepackage{forest}
\usetikzlibrary{trees, shadows,positioning,chains,fit,shapes,calc}
\usepackage{amsmath,bm,times}
\usetikzlibrary{calc}
\usepackage{dsfont}
\usepackage{cuted}
\usepackage{caption}

\usepackage{mathtools}
\DeclarePairedDelimiter{\ceil}{\lceil}{\rceil}
\DeclarePairedDelimiter{\floor}{\lfloor}{\rfloor}
\usepackage{amsthm}
\usepackage{comment}
\usepackage{enumitem}
\usepackage{arydshln}
\usepackage{cite}
\usepackage{multirow}
\usepackage{calc}
\usepackage{blkarray}
\usepackage{url}
\theoremstyle{definition}
\newtheorem{theorem}{Theorem}

\newtheorem{lemma}{Lemma}
\newtheorem{claim}{Claim}

\newtheorem{corollary}{Corollary}
\newtheorem{example}{Example}
\newtheorem{remark}{Remark}
\newtheorem{definition}{Definition}
\usepackage{graphicx}
\usepackage{dblfloatfix}
\usetikzlibrary{decorations.pathreplacing}
\usepackage{blindtext, graphicx, amsfonts,
	amssymb,multirow,epstopdf}
\def\BibTeX{{\rm B\kern-.05em{\sc i\kern-.025em b}\kern-.08em
    T\kern-.1667em\lower.7ex\hbox{E}\kern-.125emX}}

\makeatletter
\renewcommand*\env@matrix[1][*\c@MaxMatrixCols c]{%
  \hskip -\arraycolsep
  \let\@ifnextchar\new@ifnextchar
  \array{#1}}
\makeatother
\usepackage[linesnumbered,ruled]{algorithm2e}

\newcommand{\cmark}{\ding{51}}%
\newcommand{\xmark}{\ding{55}}

\newcommand\undermat[2]{% http://tex.stackexchange.com/a/102468/5764
  \makebox[0pt][l]{$\smash{\underbrace{\phantom{%
    \begin{matrix}#2\end{matrix}}}_{\text{$#1$}}}$}#2}
\newcommand{\calB}{\mathcal{B}}
\newcommand{\calC}{\mathcal{C}}

\newcommand{\calD}{\mathcal{D}}

\newcommand{\calG}{\mathcal{G}}
\newcommand{\calH}{\mathcal{H}}
\newcommand{\calM}{\mathcal{M}}

\newcommand{\calU}{\mathcal{U}}
\newcommand{\calV}{\mathcal{V}}

\newcommand{\calN}{\mathcal{N}}
\newcommand{\calO}{\mathcal{O}}

\newcommand{\bfy}{\mathbf{y}}

\newcommand{\bfA}{\mathbf{A}}
\newcommand{\bfa}{\mathbf{a}}

\newcommand{\bfx}{\mathbf{x}}

\newcommand{\bfB}{\mathbf{B}}

\newcommand{\bfR}{\mathbf{R}}

\newcommand{\bfe}{\mathbf{e}}

\begin{document}

\title{Distributed Matrix Computations with Low-weight Encodings}

\definecolor{mygr}{rgb}{0.6,0.4,0.0}
\definecolor{my1color}{rgb}{0.6,0.4,0.0}
\definecolor{mycolor1}{rgb}{0.00000,0.44700,0.74100}%
\definecolor{mycolor2}{rgb}{0.85000,0.32500,0.09800}%
\definecolor{mycolor3}{rgb}{0.45000,0.62500,0.19800}%
\tikzset{
 basic/.style  = {draw, text width=2cm, drop shadow, font=\sffamily, rectangle},
  root/.style   = {basic, rounded corners=2pt, thin, align=center, fill=white},
  level-2/.style = {basic, rounded corners=6pt, thin,align=center, fill=white, text width=3cm},
  level-3/.style = {basic, thin, align=center, fill=white, text width=1.8cm},
  block/.style    = {draw, thick, rectangle, minimum height = 2em, minimum width = 2em},
sum/.style      = {draw, circle, node distance = 1cm},
sum1/.style      = {draw, circle, minimum size = 1.1 cm},
input/.style    = {coordinate},
output/.style   = {coordinate},
}

\author{\normalsize \IEEEauthorblockN{Anindya Bijoy Das\IEEEauthorrefmark{1}, Aditya Ramamoorthy\IEEEauthorrefmark{2}, David J. Love\IEEEauthorrefmark{1}, Christopher G. Brinton\IEEEauthorrefmark{1}} \\
\IEEEauthorblockA{\IEEEauthorrefmark{1}School of Electrical and Computer Engineering, Purdue University, West Lafayette, IN, USA\\
\IEEEauthorrefmark{2}Department of Electrical and Computer Engineering, Iowa State University, Ames, IA, USA\\
\texttt{\normalfont das207@purdue.edu, adityar@iastate.edu, djlove@purdue.edu, cgb@purdue.edu}
}
\thanks{This work was supported in part by the National Science Foundation (NSF) under grants CNS-2146171, CCF-1910840, and CCF-2115200, as well as by the Office of Naval Research (ONR) under grants N000142212305 and N000142112472.}
}

\IEEEtitleabstractindextext{%

\begin{abstract}

Straggler nodes are well-known bottlenecks of distributed matrix computations which induce reductions in computation/communication speeds. A common strategy for mitigating such stragglers is to incorporate Reed-Solomon based MDS (maximum distance separable) codes into the framework; this can achieve resilience against an optimal number of stragglers. However, these codes assign dense linear combinations of submatrices to the worker nodes. 
When the input matrices are sparse, these approaches increase the number of non-zero entries in the encoded matrices, which in turn adversely affects the worker computation time. In this work, we develop a distributed matrix computation approach where the assigned encoded submatrices are random linear combinations of a small number of submatrices. In addition to being well suited for sparse input matrices, our approach continues to have the optimal straggler resilience in a certain range of problem parameters. Moreover, compared to recent sparse matrix computation approaches, the search for a ``good'' set of random coefficients to promote numerical stability in our method is much more computationally efficient. We show that our approach can efficiently utilize partial computations done by slower worker nodes in a heterogeneous system which can enhance the overall computation speed. 
Numerical experiments conducted through Amazon Web Services (AWS) demonstrate up to $30\%$ reduction in per worker node computation time and $100\times$ faster encoding compared to the available methods. 
%Numerical experiments conducted through Amazon Web Services (AWS) support our claims on straggler mitigation with a significant gain in worker node computation speed and a considerable reduction in communication delay.\\

%up to 30% reduction in computational complexity per worker node and
%can be 100×
\end{abstract}

}

\maketitle
\IEEEdisplaynontitleabstractindextext
\IEEEpeerreviewmaketitle
%\blfootnote{This work was supported in part by the National Science Foundation (NSF) under Grant CCF-1910840 and CCF-2115200.}
\section{Introduction}
\label{sec:intro}
%Contemporary computing platforms are hard-pressed to support the growing demands for data processing by end users, from AI/ML model training to more general optimization problems. While advances in hardware can serve as part of the solution, the increasing complexity of data tasks (e.g., sizes of deep neural network models) and volumes of data to be processed will continue impeding scalability. 
Matrix operations are the fundamental building blocks of data-intensive algorithms (e.g., machine learning modeling) executed on contemporary computing platforms.
The ever-increasing volumes of data generated by end users translates to high dimensional matrices for storage and processing, underscoring the potential benefits of distributed computation. The idea behind these schemes is to break down the whole matrix computation into smaller tasks and distribute them across multiple worker nodes. In these systems, it is well known that the overall job execution time can be dominated by slower (or failed) worker nodes, which are referred to as stragglers. 

Recently, a number of coding theory techniques \cite{lee2018speeding, das2019random, dutta2016short, yu2017polynomial, das2020coded, yu2020straggler, tandon2017gradient, dasunifiedtreatment,  9513242} have been proposed to mitigate the effect of stragglers for distributed matrix multiplications. 
For example, consider the computation of $\bfA^T \bfx$ (where $\bfA \in \mathbb{R}^{t \times r}$ and $\bfx \in \mathbb{R}^{t}$) across three nodes.
A popular approach \cite{lee2018speeding} would partition $\bfA$ as $\bfA = [\bfA_0 ~|~ \bfA_1]$, and assign nodes $W_0, W_1$ and $W_2$ the job of computing $\bfA^T_0 \bfx$, $\bfA^T_1 \bfx$ and $\left(\bfA_0+\bfA_1\right)^T \bfx$, respectively. 
While each worker must carry half of the overall computational load, we can recover $\bfA^T \bfx$ as soon as {\it any} two out of three workers return their results. In other words, the system is resilient to one straggler. The recovery threshold, i.e., the minimum number of worker nodes ($\tau$) that need to finish their respective jobs such that the result $\bfA^T \bfx$ can be recovered from any subset of $\tau$ worker nodes, has emerged as an important optimization metric. 

Similar to the matrix-vector case, coding theory techniques have been developed for distributed matrix-matrix multiplication, i.e., to compute $\bfA^T \bfB$, with the same goal: to minimize the recovery threshold \cite{yu2017polynomial, dutta2019optimal}.
Under the assumption of a homogeneous system where each worker can store $1/k_A$ and $1/k_B$ fraction of matrices $\bfA$ and $\bfB$, respectively, and each node is assigned a $1/k_{AB}$ fraction ($k_{AB} = k_A k_B$) of the overall load of computing $\bfA^T \bfB$, the achievable recovery threshold is lower bounded by $k_A k_B$ \cite{yu2017polynomial}. 
%In the case of distributed matrix-matrix multiplication ($\bfA^T \bfB$), we have the additional constraint that each of the worker nodes can also store $1/k_B$ fraction of another large matrix $\bfB \in \mathbb{R}^{t \times w}$ and per worker node load is $1/k_{AB}$ fraction ($k_{AB} = k_A k_B$) of the whole job, and then, the achievable recovery threshold is lower bounded by $k_A k_B$ \cite{yu2017polynomial}. 
Note that, with the assumption $k_A = k_B$, the approach in \cite{dutta2019optimal} achieves a recovery threshold of $2 k_A - 1$ whereas the method in \cite{yu2017polynomial} provides a threshold $k_A^2$. 
However, for given matrices $\bfA \in \mathbb{R}^{t \times r}$ and $\bfB \in \mathbb{R}^{t \times w}$, 
%the overall computational load to compute $\bfA^T \bfB$ is $\mathcal{O} (rwt)$. 
the per worker node computational complexity of the approach in \cite{dutta2019optimal} is  $\mathcal{O} \left( \frac{rwt}{k_A} \right)$, which is around $k_A$ times higher than the corresponding computational complexity of every node for the approach in \cite{yu2017polynomial}, which is $\mathcal{O} \left( \frac{rwt}{k_A^2} \right)$. 

%However, compared to \cite{yu2017polynomial}, the per worker node computational load of the approach in \cite{dutta2019optimal} is substantially higher \cite{9084368}.

%Note that the approach in \cite{dutta2019optimal} assumed $k_A = k_B$, and achieves a recovery threshold of $2 k_A - 1$ whereas the method in \cite{yu2017polynomial} provides a threshold $k_A^2$. However, compared to \cite{yu2017polynomial}, the per worker node computational load of the approach in \cite{dutta2019optimal} is $k_A$ orders higher \cite{9084368}.

%It should be noted that the approach in \cite{dutta2019optimal} achieves a recovery threshold of $2 k_A - 1$ with the assumption that every worker node can store $1/k_A$ fraction of each of matrices $\bfA$ and $\bfB$ (in other words, $k_A = k_B$). However, as discussed in \cite{9084368}, the computational load of the approach in \cite{dutta2019optimal} is $k_A$ orders higher than that of the approaches in \cite{yu2017polynomial, das2019random}.

Several works based on maximum distance separable (MDS) codes \cite{yu2017polynomial, 8849468, 8919859, das2019random} have met this optimal recovery threshold. However, they have other limitations in practical distributed computing systems. First, real-world data matrices are often sparsely populated (e.g., see examples in \cite{sparsematrices}), leading to structures that can be exploited for computational efficiency gains. However, MDS code-based techniques are not built to preserve data sparsity which can significantly increase the overall job execution time. In addition, most of the available methods are not focused on the case of heterogeneous systems, 
where different worker nodes are rated with different storage capacities and speeds (e.g., if the computation is being distributed across wireless edge devices).

Motivated by these limitations, in this work, we develop a novel approach for distributed matrix-vector and matrix-matrix multiplication which explicitly accounts for sparsity in the input matrices. Our proposed approach assigns coded submatrices as random linear combinations of a very small number of uncoded submatrices to preserve the inherent sparsity up to a certain level. 
Moreover, unlike the straggler optimal schemes in \cite{das2020coded} and \cite{dasunifiedtreatment}, our approach involves a much less computationally burdensome process to find a ``good'' set of random coefficients for numerical stability of the system. 
In addition, our approach addresses the case where the worker nodes are heterogeneous in nature, having different computation and communication speeds. 

The paper is organized as follows. In Section \ref{sec:background}, we discuss the problem formulation, related literature background and summarize our contributions. Then, in Section \ref{sec:prop_approach_matvec} and \ref{sec:matmat}, we present the details of our distributed matrix-vector and matrix-matrix multiplication schemes, respectively, with results on straggler resilience, extension to heterogeneous systems, and utilization of partial computations. Next, Section \ref{sec:property} discusses different properties of our proposed schemes in terms of worker computation delay, communication delay, numerical stability and the required time to find a ``good'' set of coefficients. In Section \ref{sec:numexp}, we present numerical experiments comparing the performance of our proposed method with other recent approaches. Finally, Section \ref{sec:conclusion} concludes the paper with a discussion of possible future
directions.

\section{Problem Formulation, Background and Summary of Contributions}
\label{sec:background}

\subsection{Problem Formulation}
We consider a distributed system comprised of a central node and a set of worker nodes aiming to compute $\mathbf{A}^T \mathbf{x}$ for matrix-vector multiplication or $\bfA^T \bfB$ for matrix-matrix multiplication, for given matrices $\mathbf{A} \in \mathbb{R}^{t \times r}$, $\bfB \in \mathbb{R}^{t \times w}$ and vector $\mathbf{x} \in \mathbb{R}^{t}$. In the homogeneous setting, we assume a system of $n$ worker nodes rated with the same computation and communication speeds, for local data processing and transmitting/receiving processed data, respectively. In particular, each worker can store the equivalent of $\gamma_A = \frac{1}{k_A}$ fraction of $\bfA$ and the whole vector $\bfx$ (or $\gamma_B = \frac{1}{k_B}$ fraction of $\bfB$) for matrix-vector multiplication (or matrix-matrix multiplication). In the heterogeneous setting, by contrast,  workers are rated with different storage capacity and speeds. Stragglers arise in practice from speed variations or failures experienced by the nodes at particular times \cite{das2019random}. 
%Moreover, we consider the case where the input matrices $\bfA$ and $\bfB$ are sparse. 

In our approach in the homogeneous setting, we first partition matrices $\bfA$ and $\bfB$ into $k_A$ and $k_B$ disjoint block-columns, respectively, as $\bfA = \begin{bmatrix}
    \bfA_0 & \bfA_1 & \dots & \bfA_{{k}_A - 1}
    \end{bmatrix} $ and $ \bfB = \begin{bmatrix}
    \bfB_0 & \bfB_1 & \dots & \bfB_{k_B - 1}
    \end{bmatrix} ,
$ such that $\bfA_i \in \mathbb{R}^{t \times r/{k}_A}$ and $\bfB_j \in \mathbb{R}^{t \times w/{k}_B}$, for $0 \leq i \leq {k}_A - 1$ and $0 \leq j \leq {k}_B - 1$. 
Next, we will assign a random linear combination of some block-columns of $\bfA$ and the vector $\bfx$ (or another random linear combination of some block-columns of $\bfB$) to each worker node for matrix-vector multiplication (for matrix-matrix multiplication). 
% Next, we will assign a random linear combination of some block-columns of $\bfA$ and the vector $\bfx$ to each worker node for matrix-vector multiplication. For the matrix-matrix case, we will assign a random linear combination of some block-columns of $\bfA$ and another random linear combination of some block-columns of $\bfB$. 
As discussed in Sec. \ref{sec:intro}, assigning dense linear combinations can destroy the inherent sparsity of the corresponding matrices. Instead, we aim to assign linear combinations of a lesser number of submatrices. To quantify this, we define the ``weight'' of the encoded submatrices as follows.

\begin{definition}
\label{def:weight}
We define the weights of the encoding process $\omega_A$ and $\omega_B$ for matrices A and B, respectively, as the number of submatrices that are linearly combined to obtain each encoded submatrix. We assume uniform weights across the worker nodes, i.e., the combination received by each worker is formed from the same number of submatrices.
\end{definition}

Thus, in this work, we consider the problem of minimizing recovery threshold for both matrix-vector and matrix-matrix multiplication in the homogeneous system while maintaining low $\omega_A$ and $\omega_B$ for the assigned submatrices. We will extend the resulting approach for the heterogeneous setting as well by assigning tasks proportional to worker capabilities.
Note that we outline the notations used in this manuscript in Table \ref{tab:not} in App. \ref{app:not}.

\subsection{Background and Literature Review}
Several coded computation schemes have been proposed for matrix multiplication \cite{lee2018speeding, xhemrishi2022distributed, das2019random, dutta2016short, yu2017polynomial, das2020coded, hollanti2022secure, yu2020straggler, tandon2017gradient, aliasgari2020private, 8849468, 8919859, dasunifiedtreatment,  9513242, tandon2018secure, mallick2018rateless, 8849395} in recent years. We give a comparative summary between these schemes in terms of properties they support in Table \ref{tab:compare}; for a more detailed overview, we refer the reader to \cite{9084368}. Here we begin with an illustration of the polynomial code approach proposed in \cite{yu2017polynomial}.

\begin{table*}[t]
\caption{{\small Comparison among existing works on distributed matrix-computations (the MatDot Code-based approach \cite{dutta2019optimal} involves a higher worker node computational complexity)}} 

\label{tab:compare}
\begin{center}
\begin{small}
\begin{sc}
\begin{tabular}{c c c c c c}
\hline
\toprule
\multirow{2}{1 cm}{Codes} & Mat-mat& Optimal & Numerical  & Sparsely & Hetero.\\
  & Mult? & Thresh.? & Stability? & Coded? & System?\\
 \midrule
Repetition Codes & \cmark & \xmark & \cmark & \cmark& \cmark\\ \hline
Rateless Codes \cite{mallick2018rateless} & \xmark & \xmark & \cmark  & \xmark & \cmark\\ \hline 
Geometric Conv. Codes \cite{8849395}  & \xmark &\cmark & \cmark   & \xmark& \cmark\\ \hline
Prod. Codes \cite{lee2017high}, Fact. Codes \cite{9513242}   & \cmark &\xmark & \cmark   & \xmark& \xmark\\ \hline
Polynomial Codes \cite{yu2017polynomial}  & \cmark & \cmark & \xmark  & \xmark& \cmark\\ \hline
MatDot Codes \cite{dutta2019optimal}  & \cmark & \cmark & \xmark  & \xmark& \cmark\\ \hline
Bivariate  Poly. Code \cite{9519610}  & \cmark & \cmark & \xmark &  \xmark& \cmark\\ \hline
%Dynamic Hetero.-Aware Code \cite{keshtkarjahromi2018dynamic} & \xmark &\xmark & \cmark  & \cmark & \xmark\\ \hline
OrthoPoly \cite{8849468}, RKRP code\cite{8919859}, & \multirow{2}{0.3 cm}{\cmark}& \multirow{2}{0.3 cm}{\cmark} & \multirow{2}{0.3 cm}{\cmark}  & \multirow{2}{0.3 cm}{\xmark}& \multirow{2}{0.3 cm}{\cmark}\\
Conv. Code \cite{das2019random}, Cir. Rot. Mat. \cite{ramamoorthy2019numerically} & &  & & \\ \hline
%C$^3$LES \cite{c3les} & \xmark & \xmark & \cmark & \cmark & \cmark\\ \hline
Sparse Private Approach \cite{xhemrishi2022distributed} & \xmark& \xmark & \cmark & \cmark & \cmark\\ \hline
$\beta$-level Coding \cite{das2020coded} & \cmark& \xmark & \cmark  & \cmark& \xmark\\ \hline
SCS Optimal Scheme \cite{das2020coded} & \cmark& \cmark & \cmark  & \cmark& \xmark\\ \hline
Class-based Scheme \cite{dasunifiedtreatment} & \cmark& \cmark & \cmark  & \cmark& \xmark\\ \hline
%Fully Secure approach \cite{tandon2018secure}, & \multirow{2}{0.3 cm}{\cmark}& \multirow{2}{0.3 cm}{\cmark} & \multirow{2}{0.3 cm}{\xmark} & \multirow{2}{0.3 cm}{\cmark} & \multirow{2}{0.3 cm}{\xmark}\\ 
%SGPD Codes \cite{aliasgari2020private} and PSDMM \cite{hollanti2022secure}& &&&&\\ \hline
{\bf Proposed Scheme} & \cmark & \cmark & \cmark  & \cmark& \cmark\\
\bottomrule
\end{tabular}
\end{sc}
\end{small}
\end{center}

\end{table*}%

Consider a homogeneous system with $n = 5$ worker nodes where each worker can store $\gamma_A = \gamma_B = 1/2$ fraction of both matrices $\bfA$ and $\bfB$. We partition matrices $\bfA$ and $\bfB$ into $k_A = k_B = 2$ block-columns each, as $\bfA_0$, $\bfA_1$ and $\bfB_0, \bfB_1$, respectively. Next we define two matrix polynomials as $\bfA(z) = \bfA_0 + \bfA_1 z$ and $\bfB(z) = \bfB_0 + \bfB_1 z^2$, so that $
    \bfA^T (z) \bfB(z) = \bfA_0^T \bfB_0 + \bfA_1^T \bfB_0 z + \bfA_0^T \bfB_1 z^2 + \bfA_1^T \bfB_1 z^3$. The central node evaluates $\bfA(z)$ and $\bfB(z)$ at $n = 5$ distinct real numbers and transmits the corresponding matrices to worker node $W_i$, for $i = 0, 1, 2, \dots, n - 1$. Now each node computes its respective assigned matrix-matrix block-product and returns the result to the central node. Since $\bfA^T (z) \bfB(z)$ is a degree-$3$ polynomial, once the central node receives results from the fastest $\tau = 4$ worker nodes, it can decode all the coefficients in $\bfA^T (z) \bfB(z)$, hence, $\bfA^T \bfB$. Thus, the recovery threshold is $\tau = 4$ and the system is resilient to $s = 1$ straggler. 

Unlike the schemes in \cite{wang2018coded, mallick2018rateless} which are sub-optimal in terms of straggler resilience, the polynomial code approach is among the first to provide the optimal recovery threshold.  However, recent works on matrix computations have identified metrics beyond recovery threshold that also need to be considered. Here we discuss the importance of factoring them into our methodology.

%computational fluid dynamics,
{\bf Sparsity of matrices $\bfA$ and $\bfB$}: Sparsity is quite prevalent in real-world datasets with applications in optimization, deep learning, power systems,  electromagnetism etc. (see \cite{sparsematrices} for such examples). In other words, there are many practical problems where the corresponding matrices to be operated on are sparse, which can be exploited to significantly reduce matrix computation time \cite{wang2018coded}. Consider two column vectors of length $m$, denoted by $\bfa$ and $\bfy$, where $\bfa$ has around $\psi m$ non-zero entries ($0 < \psi << 1$). It takes approximately $2 \psi m$ floating point operations (FLOPs) to compute $\bfa^T \bfy$, whereas it could take around $2m$ FLOPs if $\bfa$ was dense. 

%In a similar manner, the dense linear encoding in the approaches in \cite{yu2017polynomial, 8849468, 8919859, das2019random} significantly increases the number of non-zero entries in the encoded matrices which are then assigned to the worker nodes. 
As shown in Table \ref{tab:compare}, many existing coding approaches do not preserve sparsity. For example, in the polynomial code approach \cite{yu2017polynomial} or its variants \cite{8849468}, the encoded submatrices of $\bfA$ and $\bfB$ are obtained by linearly combining $k_A$ and $k_B$ submatrices. Thus, the number of non-zero entries can increase by up to $k_A$ and $k_B$ times, respectively, compared with the original matrices $\bfA$ and $\bfB$, which would lead to a significantly higher computation time. This underscores the importance of developing schemes that minimize the number of uncoded submatrices that are combined. 

Note that there are several methods available in the literature \cite{wang2018coded, das2020coded, dasunifiedtreatment, xhemrishi2022distributed, 9965842} which demonstrate some advantages in sparse matrix computations. However, the approach in \cite{xhemrishi2022distributed} is not developed for matrix-matrix multiplication, and the approach in \cite{9965842} has different assumptions than ours: the central node is also responsible for some computations. In addition, the approach in \cite{wang2018coded} and $\beta$-level coding scheme in \cite{das2020coded} do not meet the exact optimal recovery threshold, which require more nodes to finish their respective tasks. While the approach in \cite{dasunifiedtreatment} and the SCS-optimal scheme in \cite{das2020coded} meet the exact optimal recovery threshold, they partition the matrices into large number of block-columns, hence they need a large amount of time to find a ``good'' set of coefficients to obtain the linear combinations (more details are given in Sec. \ref{sec:trialtime} and Sec. \ref{sec:numexp}). Besides, some of the assigned submatrices in those approaches are densely coded, thus, an improved coding approach could further optimize the computational speed over those methods.

{\bf Numerical stability of the system}: Another important issue is the numerical stability of the system. Since the encoding and decoding methods in coded computation operate over the real field, the decoding of the unknowns from a system  of equations can be quite inaccurate if the corresponding system matrix is ill-conditioned. There can be a blow-up of round-off errors in the decoded result owing to the high condition numbers of the corresponding decoding matrices. For example, the polynomial code approach in \cite{yu2017polynomial} incorporates Vandermonde matrices into the encoding process which are known to be ill-conditioned. Literature aiming to address this issue \cite{8849468, 8919859, das2019random} has emphasized that the worst case condition number $(\kappa_{worst})$ over all different choices of stragglers should be treated as an important metric for minimization. 

%The approach based on orthogonal polynomials \cite{8849468} shows that  the $\kappa_{worst}$ value in that scheme can be at most $O(n^{2s})$, while the work based on rotation and circulant permutation matrices \cite{ramamoorthy2019numerically} demonstrates that $\kappa_{worst}$ can be upper bounded by $O(n^{s + 5.5})$. The method in \cite{8919859} take random linear combinations of the submatrices to obtained the encoded submatrices (also suggested in Remark 8 of \cite{yu2017polynomial}) and shows significant improvement over the polynomial code approach. The work in \cite{das2019random} proposes a scheme based on convolutional codes and provides a computable upper bound on $\kappa_{worst}$.

However, many of the numerically stable methods are based on random codes \cite{8919859, das2019random, das2020coded, dasunifiedtreatment} which require significant time to find a ``good'' set of random coefficients that make the system numerically stable. The idea is to first generate a set of random coefficients and find the $\kappa_{worst}$ over all choices of stragglers. Next, repeat this step several times (say, $20$) and retain the set of random coefficients which gives minimum $\kappa_{worst}$. The latency incurred by this process increases with the number of worker nodes and can delay the encoding process.

{\bf Heterogeneous system and partial computations done by the stragglers}: Another important issue arises in distributed computing with heterogeneous worker nodes (different memory, speed and bandwidth) where algorithms based on homogeneous assumptions may lead to sub-optimal performance \cite{reisizadeh2019coded}. 
%In this regard, algorithms need to be developed which take account the assumption of worker node heterogeneity. 
In this work, we address this issue by assigning each worker a processing load according to its memory and speed. 
%Thus, we assign a higher (or lower) computational load to a worker node with higher (or lower) memory and computation speed. 
Some of the approaches \cite{yu2017polynomial,8849468, 8919859} that were originally developed for a homogeneous system could similarly be extended to the heterogeneous setting.

%Moreover, the performance in heterogeneous computing systems often depends on how well the scheme is able to utilize (rather than discard) partial computations done by slower workers \cite{9252114}. Specifically, efficient utilization of the partial computations done by slower nodes may enhance the overall speed. 
%To address this, we consider a scheme in which we assign each of the worker nodes to {\it multiple} smaller tasks, with the size of each task dictated by the limitations of the least powerful worker. Assume the total number of resulting tasks is $\Delta$ ($\Delta$ corresponds to $k_A$ and $k_A k_B$ for matrix-vector and matrix-matrix multiplication, respectively, in the homogeneous case). 
%With this division, define $Q$ to be the maximum number of block products that must be returned to the central node in order for it to decode the intended result ($\bfA^T \bfx$ or $\bfA^T \bfB$), i.e., from {\it any} $Q$ block products returned across all the worker nodes (respecting the computation order in each node). We associate the scheme with the value $Q / \Delta$, where $Q/\Delta$ is always lower bounded by $1$, and a small $Q/\Delta$ value indicates that the system can utilize the partial computations of the slower workers efficiently \cite{das2020coded, dasunifiedtreatment}.

Moreover, the performance of distributed computing often depends on how well the associated scheme is able to utilize (rather than discard) partial computations done by slower workers \cite{9252114}. Specifically, efficient utilization of the partial computations done by slower nodes may enhance the overall speed. 
To address this, in our scheme in the heterogeneous setting, we assign {\it multiple} smaller tasks to some workers, with the size of each task dictated by the limitations of the least powerful worker. The workers then compute their respective tasks and return the results sequentially. 
With this division, define $Q$ to be the minimum number of block products that must be returned to the central node for the guarantee of decoding the intended result ($\bfA^T \bfx$ or $\bfA^T \bfB$) successfully even in the worst case, i.e., from {\it any} $Q$ block products returned across all the worker nodes (respecting the computation order in each node) \cite{das2020coded}. We associate the scheme with the metric $Q / \Delta$, where $\Delta$ is the number of submatrix products to be recovered in the intended product ($\Delta = {k}_A$ or $k_A k_B$ for matrix-vector or matrix-matrix multiplication, respectively). $Q/\Delta$ is always lower bounded by $1$, and a system with a small $Q/\Delta$ can utilize the partial computations of the slower workers efficiently \cite{das2020coded, dasunifiedtreatment}.

%Moreover, the performance in heterogeneous computing systems often depends on how the scheme utilizes the partial computations done by the slower workers \cite{9252114}. Specifically, a slower node may not be a useless node, rather an efficient utilization of the partial computations done by that slower node may enhance the overall speed. To address this, we assume that the nodes can be assigned multiple smaller tasks which are computed sequentially by the worker nodes. For such a system, let us assume that the central node can decode the intended result ($\bfA^T \bfx$ or $\bfA^T \bfB$) from {\it any} $Q$ block products returned by all the worker nodes (respecting the computation order in each node), we say that the scheme has the corresponding $\frac{Q}{\Delta}$ value. Here $\Delta$ is the number of submatrix products to be recovered in the intended product which is  ${k}_A$ for matrix-vector multiplication or ${k}_A {k}_B$ for matrix-matrix multiplication in the homogeneous setting. Thus, $Q/\Delta$ is always lower bounded by $1$ and a small $Q/\Delta$ value of a system indicates that the system can utilize the partial computations of the slower workers quite efficiently \cite{das2020coded, dasunifiedtreatment}.

\subsection{Summary of Contributions}
The contributions of this work can be summarized as follows.
\begin{itemize}
    \item We develop novel straggler-resilient approaches for distributed matrix-vector and matrix-matrix multiplication. 
    For a system with $n$ homogeneous worker nodes, each of which can store $1/k_A$ and $1/k_B$ fractions of matrices $\bfA$ and $\bfB$, respectively, our developed approach for distributed matrix-matrix multiplication can be resilient to {\it any} $s = n - k_A k_B$ stragglers (where $s \leq \textrm{min}(k_A, k_B)$). Thus, our approach is straggler optimal, since it meets the lower bounds on straggler resilience as given in \cite{yu2017polynomial}. 
    In addition, with the assumption that each node stores the whole vector $\bfx$, our approach for the matrix-vector case is also straggler-optimal.
    %In the matrix-vector case (to compute $\bfA^T \bfx$), for a system with $n$ homogeneous worker nodes, each of which can store $\gamma_A = 1/k_A$ fraction of matrix $\bfA$ and the vector $\bfx$, we develop a cyclic assignment based scheme which can be resilient to up to $n - k_A$ stragglers. 
    %In the matrix-matrix case (to compute $\bfA^T \bfB$), the additional constraint is that each worker can also store $\gamma_B = 1/k_B$ fraction of matrix $\bfB$, and our developed approach can be resilient to up to $n - k_A k_B$ stragglers. Thus, our approaches are straggler optimal, because they meet the lower bounds on straggler resilience as given in \cite{yu2017polynomial}.
    
    \item While our approaches are applicable to any types of matrices, it is specifically suited to sparse ``input'' matrices. A very limited number of uncoded submatrices are linearly combined in our encoding, so that the inherent sparsity of the input matrices can be preserved up to certain level. For example, in a system with $n = 12, k_A = 10$ and $s = 2$ for distributed matrix-vector multiplication, the traditional dense codes \cite{yu2017polynomial, 8849468, 8919859, das2019random} require linear combinations of $k_A = 10$ submatrices, whereas our scheme combines only $s+1 = 3$ uncoded submatrices (see Example \ref{ex:matvec}). Thus, our approach will be significantly impactful in the scenario where $s+1 < k_A$ (as mentioned in Remark \ref{rem:skA} in Section \ref{sec:prop_approach_matvec}).
    
    % \item We also develop a secure matrix computation approach to protect from information leakage in the worker nodes which are honest but curious. While the traditional secure matrix computation approaches \cite{tandon2018secure, aliasgari2020private, hollanti2022secure} are not well-suited to sparse matrices, we provide a trade-off between privacy leakage of the system and sparsity of the encoded submatrices. Thus our approach can help to enhance the computation speed than \cite{tandon2018secure, aliasgari2020private, hollanti2022secure} while maintaining a certain level of privacy.

    \item We also show that our proposed approaches are numerically stable. It has been verified by comparing other approaches in terms of the worst case condition number over different choices of stragglers. Moreover, our scheme involves a significantly less computationally expensive step compared to \cite{das2020coded, dasunifiedtreatment} to find a ``good'' set of random coefficients for numerically stability, in other words, the encoding time can be significantly reduced in our schemes.

    \item In addition, we extend our algorithms to the heterogeneous case where the worker nodes are heterogeneous in nature having different computation and communication speeds and storage capacities. Furthermore, our approaches meet the lower bound of $Q/\Delta$ value which indicates that these can efficiently utilize the partial computations done by the slower worker nodes in a heterogeneous system, thus can enhance the overall job execution speed significantly.
    
    \item Finally, we conduct exhaustive numerical experiments on Amazon Web Services (AWS) cluster using large-sized sparse matrices, and present comparisons that demonstrate the advantages of our schemes in terms of worker node computation time, communication time, numerical stability and coefficient determination time. The results show that our proposed approach can have up to $30\%$ reduction in computational complexity per worker node and can be $100\times$ faster for determining a ``good'' set of coefficients for numerical stability.
\end{itemize}

\section{Proposed Approach for Matrix-vector Multiplication}
\label{sec:prop_approach_matvec}

In this section, we detail our approach for straggler resilient distributed matrix-vector multiplication in case of both homogeneous and heterogeneous worker nodes.

\subsection{Homogeneous System}
\label{sec:homogeneous}
First we discuss our coded matrix-vector multiplication approach in the homogeneous system with resilience to up to $s = n - k_A$ stragglers. The overall procedure is given in Alg. \ref{Alg:New_matvec}. We partition matrix $\bfA$ into $k_A$ disjoint block columns as $\bfA_0, \bfA_1, \bfA_2, \dots, \bfA_{k_A - 1}$, and assign a random linear combination of $\omega_A$ (weight) submatrices of $\bfA$ to every worker node. Formally, we set $\omega_A = \textrm{min} \,(s+1, k_A)$, and assign a linear combination of $\bfA_i, \bfA_{i + 1}, \bfA_{i + 2}, \dots, \bfA_{i+\omega_A - 1} \,$ $\left(\textrm{indices modulo} \, k_A \right)$ to worker node $W_i$, for $i = 0, 1, 2, \dots, n-1$, where the linear coefficients are chosen i.i.d. (independent and identically distributed) from a continuous distribution. Note that every worker node $W_i$ has access to  the vector $\bfx$. Once the fastest $\tau = k_A$ worker nodes return their computation results, the central node decodes $\bfA^T \bfx$. The following theorem establishes the resiliency of Alg. \ref{Alg:New_matvec} to straggler nodes.

\begin{algorithm}[t]
	\caption{Proposed scheme for distributed matrix-vector multiplication}
	\label{Alg:New_matvec}
   \SetKwInOut{Input}{Input}
   \SetKwInOut{Output}{Output}

   \Input{Matrix $\bfA$, vector $\bfx$, $n$-number of worker nodes, $s$-number of stragglers, storage fraction $\gamma_A = \frac{1}{k_A}$; $s \leq n - k_A$.}
   Partition $\bfA$ into $k_A$ block-columns as $\bfA = \begin{bmatrix}
    \bfA_0 & \bfA_1 & \dots & \bfA_{k_A - 1}
    \end{bmatrix}$\;
   Create a $n \times k_A$ random matrix $\bfR$ with entries  $r_{i,j}$, $0\leq i \leq n -1 $ and $0\leq j \leq k_A -1$\;
   Set weight $\, \omega_A = min(s+1, k_A)$\;
   \For{$i\gets 0$ \KwTo $n-1$}{
   Define $T = \left\lbrace i, i+1, \dots, i + \omega_A - 1 \right\rbrace$ (modulo $k_A$)\;
   The central node creates a random linear combination of $\bfA_{q}$'s where $q \in T$, thus $\tilde{\bfA}_i = \sum\limits_{q \in T} r_{i,q} \bfA_q$. Then, it assigns encoded submatrix $\tilde{\bfA}_i$ to worker node $W_i$\;
   Worker node $W_i$ computes $\tilde{\bfA}_i^T \bfx$\;
   }
   \Output{The central node recovers $\bfA^T \bfx$ from the results of the fastest $k_A$ worker nodes.}

\end{algorithm}

\begin{theorem}
\label{thm:matvec}
Assume that a system has $n$ worker nodes each of which can store $1/k_A$ fraction of matrix $\bfA$ and the whole vector $\bfx$ for the distributed matrix-vector multiplication $\mathbf{A}^T \mathbf{x}$. If we assign the jobs according to Alg. \ref{Alg:New_matvec}, we achieve resilience to any $s = n - k_A$ stragglers.
\end{theorem}

\begin{proof}
{\it Case 1}: First consider the case when $s<k_A$. Since we have partitioned the matrix $\bfA$ into $k_A$ disjoint block-columns, to recover $\bfA^T \bfx$, we need to decode all $k_A$ vector unknowns, $\bfA^T_0 \bfx, \bfA^T_1 \bfx, \bfA^T_2 \bfx, \dots, \bfA^T_{k_A - 1} \bfx$. We denote the set of these $k_A$ unknowns as $\calB$. Now we choose an arbitrary set of $k_A$ worker nodes. Each of these worker nodes corresponds to an equation in terms of $\omega_A$ of those $k_A$ unknowns. We denote the set of $k_A$ equations as $\calC$, thus,  $|\calB| = |\calC| = k_A$. 

Now we consider a bipartite graph $\calG = \calC \cup \calB$, where any vertex (equation) in $\calC$ is connected to some vertices (unknowns) in $\calB$ which participate in the corresponding equation. 
%Thus, each vertex in $\calC$ has a neighborhood of cardinality $\omega_A$ in $\calB$. 
%such an example with $k_A = 5$ and $\omega_A = 3$ is shown in Fig. \ref{fig:hall_bipartite}. 
Our goal is to show the existence of a perfect matching among the vertices of $\calC$ and $\calB$. We argue this according to Hall's marriage theorem \cite{marshall1986combinatorial} for which we need to show that for any $\bar{\calC} \subseteq \calC$, the cardinality of the neighbourhood of $\bar{\calC}$, denoted as $\calN (\bar{\calC}) \subseteq \calB$, is at least as large as $|\bar{\calC}|$. Thus, for $|\bar{\calC}| = m \leq k_A$, we need to show that $|\calN (\bar{\calC})| \geq m$.

{\it Case 1a}: First we consider the case that $m \leq 2s$, thus, we assume that $m$ is equal to either $2p$ or $2p - 1$, where $1 \leq p \leq s$. Here we recall that the number of unknowns participating in any equation is $\omega_A =$ min$(s+1, k_A)$, and the participating unknowns are shifted in a cyclic manner among the equations. If we choose any $\delta$ worker nodes out of the first $k_A$ worker nodes $\left(W_0, W_1, W_2, \dots, W_{k_A - 1} \right)$, according to the proof of cyclic scheme in Appendix C in \cite{das2020coded}, the minimum number of total participating unknowns is $\textrm{min} (\omega_A + \delta - 1, k_A)$. In other words, the first equation (among those $\delta$ equations) consists of $\omega_A$ unknowns, and then any additional equation includes at least one additional unknown until the number of total participating unknowns is $k_A$.

Now, according to Alg. \ref{Alg:New_matvec}, the same unknowns participate in two different equations corresponding to two different worker nodes, $W_j$ and $W_{k_A + j}$, where $j = 0, 1, \dots, s-1$. Thus for any $|\bar{\calC}| = m = 2p, 2p - 1 \leq 2s$, we have
$|\calN (\bar{\calC})|  \geq \textrm{min} \left( \omega_A + \ceil{m/2} - 1, k_A \right) = \textrm{min} \left( \omega_A + p - 1, k_A \right)$. Now, since $\omega_A =$ min$(s+1, k_A)$, we can say that $|\calN (\bar{\calC})| \geq \textrm{min} \left( s + p, k_A \right) \geq m$.

{\it Case 1b}: Now we consider the remaining case where $m = 2s + q$, $1 \leq q \leq k_A - 2s$. We need to find the minimum number of unknowns which participate in any set of $m$ equations. As we have discussed before, the same unknowns participate in two different equations corresponding to two different worker nodes, $W_j$ and $W_{k_A + j}$, where $j = 0, 1, \dots, s-1$. Thus, the additional $q$ equations will correspond to at least $q$ additional unknowns. Therefore, $|\calN (\bar{\calC})| \geq \textrm{min} \left( \omega_A + \ceil{2s/2} + q - 1, k_A \right)$. Now, since $\omega_A = \textrm{min} (s+1, k_A)$, we have, $|\calN (\bar{\calC})| \geq \textrm{min} \left( 2s + q, k_A \right) = m$.  
% \begin{align*}
% |\calN (\bar{\calC})| &\geq \textrm{min} \left( \omega_A + \ceil{2s/2} + q - 1, k_A \right) \\ & = \textrm{min} \left( \omega_A + s + q - 1, k_A \right) = \textrm{min} \left( 2s + q, k_A \right) \geq m.  
% \end{align*} 

Thus, in this case when $s < k_A$, for any $\calC$, we have shown that $|\calN (\bar{\calC})| \geq |\bar{\calC}|$. So, there exists a perfect matching among the vertices of $\calC$ and $\calB$ according to Hall's marriage theorem. Now we consider the largest matching where the vertex $c_i \in \calC$ is matched to the vertex $b_j \in \calB$, which indicates that $b_j$ participates in the equation corresponding to $c_i$. Let us consider a $k_A \times k_A$ system matrix where row $i$ corresponds to the equation associated to $c_i$ where $b_j$ participates. Let us replace row $i$ of the system matrix by $\bfe_j$ where $\bfe_j$ is a unit row-vector of length $k_A$ with the $j$-th entry being $1$, and $0$ otherwise. Thus we have a $k_A \times k_A$ matrix where each row has only one non-zero entry which is $1$. 
Since we have a perfect matching, this $k_A \times k_A$ matrix will have only one non-zero entry in every column. This is a permutation of the identity matrix, and, thus, is full rank. Since the matrix is full rank for a choice of definite values, according to Schwartz-Zippel lemma \cite{schwartz1980fast}, the matrix continues to be full rank for random choices of non-zero entries. Thus, the central node can recover all $k_A$ unknowns from any set of $k_A$ worker nodes.

{\it Case 2}: Next consider the case when $s\geq k_A$. In this case, $\omega_A = \textrm{min} \left(s+1, k_A\right) = k_A$, thus all the worker nodes are assigned linear combinations of $k_A$ submatrices. Since the entries are chosen randomly from a continuous distribution, we can say that any $k_A \times k_A$ submatrix of the $n \times k_A$ system matrix is full rank. Thus we can recover all $k_A$ unknowns from the returned results of any $k_A$ out of $n$ worker nodes.
\end{proof}

\begin{remark}
\label{rem:skA}
While our proposed method is applicable for any values of $s$ and $k_A$, it is particularly impactful if $s < k_A - 1$ (thus, $\omega_A < k_A$), i.e., the percentage of the number of stragglers is less than $50\%$ of the total number of nodes (in the common practical cases, the corresponding percentage can be even less than $10\%$ \cite{das2019random, gupta2020serverless}). In these cases, the encoded submatrices have less weights than the dense coded approaches \cite{yu2017polynomial, 8849468, 8919859} which could reduce the expected communication delay and the average worker node computation delay for sparse ``input'' matrices.
\end{remark}

\begin{example}
\label{ex:matvec}
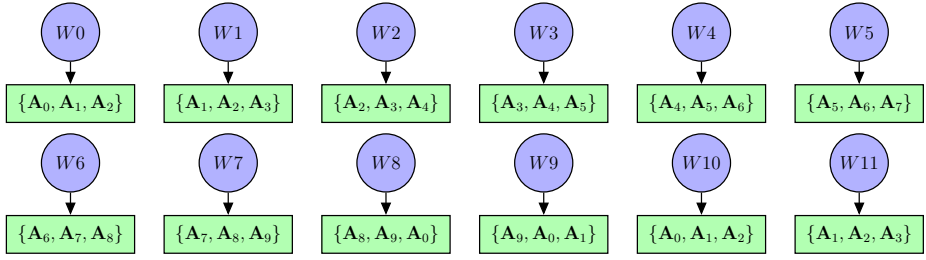
\begin{figure}[t]
\centering
\definecolor{mycolor6}{rgb}{0.92941,0.69412,0.12549}%

%\captionsetup{justification=centering}
%\resizebox{0.99\linewidth}{!}{
\resizebox{0.75\linewidth}{!}{
\begin{tikzpicture}[auto, thick, node distance=2cm, >=triangle 45]

\draw

    node [sum, minimum size = 1.3cm, fill=blue!30] (blk1) {$W0$}
    node [sum, minimum size = 1.3cm,fill=blue!30,right = 2.2 cm of blk1] (blk2) {$W1$}
    node [sum, minimum size = 1.3cm,fill=blue!30,right = 2.2 cm of blk2] (blk3) {$W2$}
    node [sum, minimum size = 1.3cm,fill=blue!30,right = 2.2 cm of blk3] (blk4) {$W3$}
    node [sum, minimum size = 1.3cm,fill=blue!30,right = 2.2 cm of blk4] (blk5) {$W4$}
    node [sum, minimum size = 1.3cm,fill=blue!30,right = 2.2 cm of blk5] (blk6) {$W5$}
    node [sum, minimum size = 1.3cm,fill=blue!30,below = 1.6 cm of blk1] (blk7) {$W6$}
    node [sum, minimum size = 1.3cm,fill=blue!30,right = 2.2 cm of blk7] (blk8) {$W7$}
    node [sum, minimum size = 1.3cm,fill=blue!30,right = 2.2 cm of blk8] (blk9) {$W8$}
    node [sum, minimum size = 1.3cm,fill=blue!30,right = 2.2 cm of blk9] (blk_10) {$W9$}
    node [sum, minimum size = 1.3cm,fill=blue!30,right = 2.2 cm of blk_10] (blk_11) {$W10$}
    node [sum, minimum size = 1.3cm,fill=blue!30,right = 2.2 cm of blk_11] (blk_12) {$W11$}

    node [block, fill=green!30, minimum width = 6.8em, below = 0.5 cm of blk1] (blk11) {$\left\lbrace\bfA_0, \bfA_1, \bfA_2 \right\rbrace$}
    node [block, fill=green!30, minimum width = 6.8em, below = 0.5 cm of blk2] (blk21) {$\left\lbrace\bfA_1, \bfA_2, \bfA_3 \right\rbrace$}
    node [block, fill=green!30, minimum width = 6.8em, below = 0.5 cm of blk3] (blk31) {$\left\lbrace\bfA_2, \bfA_3, \bfA_4 \right\rbrace$}
    node [block, fill=green!30, minimum width = 6.8em, below = 0.5 cm of blk4] (blk41) {$\left\lbrace\bfA_3, \bfA_4, \bfA_5 \right\rbrace$}
    node [block, fill=green!30, minimum width = 6.8em, below = 0.5 cm of blk5] (blk51) {$\left\lbrace\bfA_4, \bfA_5, \bfA_6 \right\rbrace$}
    node [block, fill=green!30, minimum width = 6.8em, below = 0.5 cm of blk6] (blk61) {$\left\lbrace\bfA_5, \bfA_6, \bfA_7 \right\rbrace$}
    node [block, fill=green!30, minimum width = 6.8em, below = 0.5 cm of blk7] (blk71) {$\left\lbrace\bfA_6, \bfA_7, \bfA_8 \right\rbrace$}
    node [block, fill=green!30, minimum width = 6.8em, below = 0.5 cm of blk8] (blk81) {$\left\lbrace\bfA_7, \bfA_8, \bfA_9 \right\rbrace$}
    node [block, fill=green!30, minimum width = 6.8em, below = 0.5 cm of blk9] (blk91) {$\left\lbrace\bfA_8, \bfA_9, \bfA_0 \right\rbrace$}
    node [block, fill=green!30, minimum width = 6.8em, below = 0.5 cm of blk_10] (blk_101) {$\left\lbrace\bfA_9, \bfA_0, \bfA_1 \right\rbrace$}
    node [block, fill=green!30, minimum width = 6.8em, below = 0.5 cm of blk_11] (blk_111) {$\left\lbrace\bfA_0, \bfA_1, \bfA_2 \right\rbrace$}
    node [block, fill=green!30, minimum width = 6.8em, below = 0.5 cm of blk_12] (blk_121) {$\left\lbrace\bfA_1, \bfA_2, \bfA_3 \right\rbrace$}
    
;
\draw[->](blk1) -- node{} (blk11);
\draw[->](blk2) -- node{} (blk21);
\draw[->](blk3) -- node{} (blk31);
\draw[->](blk4) -- node{} (blk41);
\draw[->](blk5) -- node{} (blk51);
\draw[->](blk6) -- node{} (blk61);
\draw[->](blk7) -- node{} (blk71);
\draw[->](blk8) -- node{} (blk81);
\draw[->](blk9) -- node{} (blk91);
\draw[->](blk_10) -- node{} (blk_101);
\draw[->](blk_11) -- node{} (blk_111);
\draw[->](blk_12) -- node{} (blk_121);

\end{tikzpicture}
}
\caption{\small Submatrix allocation for $n = 12$ workers and $s = 2$ stragglers, with $\gamma_A = \frac{1}{10}$ according to Alg. \ref{Alg:New_matvec}. Here, the weight of every submatrix is $\omega_A = \textrm{min} (s + 1, k_A) = 3$. Any $\{\bfA_i, \bfA_j, \bfA_k\}$ indicates a random linear combination of the corresponding submatrices where the coefficients are chosen i.i.d. at random from a continuous distribution.}
\label{matvec12}
\vspace{-0.6 cm}
\end{figure}

Consider a system with $n = 12$ worker nodes each of which can store $1/10$ fraction of matrix $\bfA$. We partition matrix $\bfA$ into $k_A = 10$ disjoint block-columns, $\bfA_0, \bfA_1, \dots, \bfA_9$. According to Alg. \ref{Alg:New_matvec}, we set the weight $\omega_A = \textrm{min}(s + 1, k_A) = \textrm{min} (n - k_A + 1, k_A) = 3$, and assign random linear combinations of submatrices $\bfA_i, \bfA_{i + 1}, \bfA_{i + 2} \, \left(\textrm{indices modulo} \, 10 \right)$ to worker node $W_i$, for $i = 0, 1, 2, \dots, 11$, as shown in Fig. \ref{matvec12}. 
%Since there are $12$ equations (in other words, $12$ worker nodes) in terms of $10$ unknowns, we have a $12 \times 10$ system matrix. 
%According to the proof of Theorem \ref{thm:matvec}, any $10 \times 10$ submatrix of that system matrix is full-rank. 
Thus, according to Theorem \ref{thm:matvec}, the system has a recovery threshold $\tau = k_A = 10$, and it is resilient to any $s = 2$ stragglers.
\end{example}

% \begin{align*}
% \begin{bmatrix}
%     * & * & * & 0 & 0 & 0 & 0 & 0  & 0 & 0  \\
%     0 & * & * & * & 0 & 0 & 0 & 0  & 0 & 0  \\
%     0 & 0 & * & * & * & 0 & 0 & 0  & 0 & 0  \\
%     0 & 0 & 0 & * & * & * & 0 & 0  & 0 & 0  \\
%     0 & 0 & 0 & 0 & * & * & * & 0  & 0 & 0  \\
%     0 & 0 & 0 & 0 & 0 & * & * & *  & 0 & 0  \\
%     0 & 0 & 0 & 0 & 0 & 0 & * & *  & * & 0  \\
%     0 & 0 & 0 & 0 & 0 & 0 & 0 & *  & * & *  \\
%     * & 0 & 0 & 0 & 0 & 0 & 0 & 0  & * & *  \\
%     * & * & 0 & 0 & 0 & 0 & 0 & 0  & 0 & *  \\
%     * & * & * & 0 & 0 & 0 & 0 & 0  & 0 & 0  \\
%     0 & * & * & * & 0 & 0 & 0 & 0  & 0 & 0  \\
% \end{bmatrix} 
% \end{align*}

\subsection{Extension to Heterogeneous System}
\label{sec:hetero_mv}
In this section, we extend our approach in Alg. \ref{Alg:New_matvec} to a heterogeneous system of $\bar{n}$ worker nodes where the nodes may have different computation speeds and communication speeds. We assume that true knowledge about the storage and speeds of the worker nodes are available prior to the assignment of the jobs. We also assume that we have $\lambda$ different types of nodes in the system, with worker node type $0, 1, \dots, \lambda - 1$. First, without loss of generality (w.l.o.g.), we sort the worker nodes in a non-ascending order in terms of the worker node types. Next, suppose that $\alpha$ is the number of the assigned columns and $\beta$ is the number of processed columns per unit time in the ``weakest'' type node. In this scenario, we assume that a worker node $W_i$ of type $j_i$ receives $c_{j_i} \alpha$ coded columns of data matrix $\bfA$ and has a computation speed $c_{j_i} \beta$, where $c_{j_i} \geq 1$ is an integer. Thus, a higher $c_{j_i}$ indicates a ``stronger'' type node $W_i$ which has a $c_{j_i}$ times higher memory and can process at a $c_{j_i}$ times higher computation speed than the ``weakest'' type node. Since we sort the nodes in a non-ascending order in terms of the worker node types, we have $j_0 \geq j_1 \geq j_2 \geq \dots \geq j_{\bar{n}-1} = 0$, hence $c_{j_0} \geq c_{j_1} \geq c_{j_2} \geq \dots \geq c_{j_{\bar{n}-1}} = 1$. Note that $\lambda = 1$ and all $c_{j_i} = 1$ lead us to the homogeneous system in Sec. \ref{sec:homogeneous} where $0 \leq i \leq n -1$ and $j_i = 0$. 

Assume that the ``weakest'' type worker node requires $\mu$ units of time to process $\alpha$ columns of $\bfA$. Thus, any node $W_i$ of type $j_i$ can process $c_{j_i} \alpha$ columns in time $\mu$. In this scenario, from the computation and storage perspective, worker node $W_i$ (of type $j_i$) can be considered as a combination of $c_{j_i} \geq 1$ worker nodes of the ``weakest'' type. Thus, $\bar{n}$  worker nodes in the heterogeneous system can be thought as homogeneous system of $n = \sum_{i = 0}^{\bar{n} - 1} c_{j_i}$ worker nodes of the ``weakest'' type. In other words, the worker node $W_k$ in the heterogeneous system ($0 \leq k \leq \bar{n} - 1$) can be thought as a combination of worker nodes $\bar{W}_m, \bar{W}_{m+1}, \dots, \bar{W}_{m + c_{kj} - 1}$ in a homogeneous setting, where $m = \sum_{i = 0}^{k-1} c_{j_i}$ and $W_k$ is of type $j_i$ worker node. Now, for any worker node index $\bar{k}_A$ ( such that $0 \leq \bar{k}_A \leq \bar{n} - 1$), we define $k_A = \sum_{i = 0}^{\bar{k}_A - 1} c_{j_i}$ and $ s = \sum_{i = \bar{k}_A}^{\bar{n} - 1} c_{j_i}$, so, $n = \sum\limits_{i = 0}^{\bar{n} - 1} c_{j_i} = k_A + s$.
Thus, a heterogeneous system of $\bar{n}$ worker nodes can be thought as a homogeneous system of $n = k_A + s$ nodes, for any $\bar{k}_A$ ($0 \leq \bar{k}_A \leq \bar{n} - 1$). We state the following corollary of Theorem \ref{thm:matvec} for heterogeneous system.

% \begin{theorem}
% \label{thm:matvec}
% (a) A heterogeneous system of $\bar{n}$ worker nodes of different types can be considered as a homogeneous system of $n = \sum_{i = 0}^{\bar{n} - 1} c_{j_i} $ worker nodes of the ``weakest'' type. (b) Consider a $\bar{k}_A$ ($0 \leq \bar{k}_A \leq \bar{n} - 1$), for which we have  $k_A = \sum_{i = 0}^{\bar{k}_A - 1} c_{j_i}$ and $s = \sum_{i = \bar{k}_A}^{n - 1} c_{j_i}$, where $W_i$ of of worker node type $j$. Now, if the jobs are assigned according to Alg. \ref{Alg:New_matvec} in the modified homogeneous system of $n = k_A + s$ ``weakest'' worker nodes, the system can be resilient to $s$ such nodes. (c) The scheme in Alg. \ref{Alg:New_matvec} for the modified homogeneous system of $n$ worker nodes provides a $Q/\Delta$ value as $1$.
% \end{theorem}

\begin{corollary}
 \label{cor:matvec}
Consider a heterogeneous system of $\bar{n}$ nodes of different types and assume any $\bar{k}_A$ (where $0 \leq \bar{k}_A \leq \bar{n} - 1$). Now, if the jobs are assigned to the modified homogeneous system of $n = k_A + s$ ``weakest'' type worker nodes according to Alg. \ref{Alg:New_matvec}, the system (a) will be resilient to $s$ such nodes and (b) will provide a $Q/\Delta$ value as $1$.
\end{corollary}

\begin{proof}
We first modify the heterogeneous system of $\bar{n}$ nodes to a homogeneous system of $n = k_A + s$ nodes which are of the ``weakest'' type. Now, we assign the jobs to this homogeneous system according to Alg. 1 and each node is assigned $1/k_A$-th fraction of the whole job. Thus, according to Theorem 1, we have resilience to {\it any} $s$ nodes, which concludes the proof of (a).

Next, according to Alg. 1, we partition matrix $\bfA$ into $k_A$ block-columns, thus, we have total $\Delta = k_A$ submatrix products to be decoded to recover $\bfA^T \bfx$. Now according to part (a) of this proof, in the modified homogeneous system of $n$ worker nodes, we assign $c_{j_i}$ tasks to worker node $W_i$ which is of type $j$. Thus, we assign, in total, $n = \sum_{i = 0}^{\bar{n} - 1} c_{j_i}$ tasks, and the scheme is resilient to any $s= \sum_{i = \bar{k}_A}^{\bar{n} - 1} c_{j_i}$ tasks. Thus, all the unknowns can be recovered from {\it any} $n - s = \sum_{i = 0}^{\bar{k}_A - 1} c_{j_i} = k_A$ submatrix products, hence $Q = k_A$. Thus, we have $Q/\Delta = 1$. 
\end{proof}

% \begin{proof}
% We first modify the heterogeneous system of $\bar{n}$ nodes to a homogeneous system of $n = k_A + s$ nodes which are of the ``weakest'' type. Now, we assign the jobs to this homogeneous system according to Alg. \ref{Alg:New_matvec} and each node is assigned $1/k_A$-th fraction of the whole job. Thus, according to Theorem \ref{thm:matvec}, we have resilience to {\it any} $s$ nodes, which concludes the proof of (a).

% Next, according to Alg. \ref{Alg:New_matvec}, we partition matrix $\bfA$ into $k_A$ block-columns, thus, we have total $\Delta = k_A$ submatrix products to be decoded to recover $\bfA^T \bfx$. Now according to part (a) of this proof, in the modified homogeneous system of $n$ worker nodes, we assign $c_{j_i}$ tasks to worker node $W_i$ which is of type $j$. Thus, we assign, in total, $n = \sum_{i = 0}^{\bar{n} - 1} c_{j_i}$ tasks, and the scheme is resilient to any $s= \sum_{i = \bar{k}_A}^{\bar{n} - 1} c_{j_i}$ tasks. Thus, all the unknowns can be recovered from {\it any} $n - s = \sum_{i = 0}^{\bar{k}_A - 1} c_{j_i} = k_A$ submatrix products, hence $Q = k_A$. Thus, we have $Q/\Delta = 1$. 
% \end{proof}

As we have discussed in Sec. \ref{sec:homogeneous}, under the assumption that each worker node has been assigned $1/k_A$ fraction of the whole job, our proposed approach for a homogeneous system can be resilient to any $s$ stragglers out of $n = k_A + s$ worker nodes. Now, the heterogeneous system of $\bar{n}$ worker nodes is resilient to $s = \sum_{i = \bar{k}_A}^{n - 1} c_{j_i} $ block-column processing, where each worker node has been assigned $1/k_A$ fraction of the whole job. Varying the indexing of the worker nodes depending on different node types or changing the value of $\bar{k}_A$, one can decrease (or increase) the per worker load of the job, which can increase (or decrease) the value of $\bar{s}$. The number of actual stragglers that the system is resilient to can vary depending on the worker node types.

\begin{example}
\label{exmpl:hetero}
\begin{figure}[t]
\centering
\resizebox{0.8\linewidth}{!}{

\definecolor{mycolor6}{rgb}{0.92941,0.69412,0.12549}%
\definecolor{mycolor7}{rgb}{0.74902,0.00000,0.74902}%
\definecolor{mycolor8}{rgb}{0.60000,0.20000,0.00000}%

\begin{tikzpicture}[auto, thick, node distance=2cm, >=triangle 45]

\draw
    
    node [sum, minimum size = 1.15cm, fill=blue!30] (bl1) {$W2$}
    node [block,minimum height = 0.85cm, minimum width = 1.5cm,fill=blue!30] at (-2.9,0.3) (bl4) {$W1$}
    node [block, minimum height = 0.85cm, minimum width = 1.5cm,fill=blue!30,left = 1.4 cm of bl4] (bl5) {$W0$}
    node [sum, minimum size = 1.15cm,fill=blue!30,right = 1.8 cm of bl1] (bl2) {$W3$}
    node [sum, minimum size = 1.15cm,fill=blue!30,right = 1.8 cm of bl2] (bl3) {$W4$}
    node [sum, minimum size = 1.15cm,fill=blue!30,right = 1.8 cm of bl3] (bl6) {$W5$}
    node [sum, minimum size = 1.15cm,fill=blue!30,right = 1.8 cm of bl6] (bl7) {$W6$}

    node [block, fill=green!30, minimum width = 4.8em, below = 0.5 cm of bl1] (bl11) {\normalfont $\{  {\bfA}_4, {\bfA}_5, {\bfA}_6 \}$}
    node [block, fill=green!30, minimum width = 4.8em, below = 0.5 cm of bl2] (bl21) {\normalfont $\{ {\bfA}_5, {\bfA}_6, {\bfA}_0 \}$}
    node [block, fill=green!30, minimum width = 4.8em, below = 0.5 cm of bl3] (bl31) {\normalfont $\{ {\bfA}_6, {\bfA}_0, {\bfA}_1 \}$}
    node [block, fill=green!30, minimum width = 4.8em, below = 0.5 cm of bl4] (bl41) {\normalfont $\{ {\bfA}_2, {\bfA}_3, {\bfA}_4 \}$}
    node [block, fill=green!30, minimum width = 4.8em, below = 0.005 cm of bl41] (bl42) {\normalfont $\{{\bfA}_3, {\bfA}_4, {\bfA}_5 \}$}
    node [block, fill=green!30, minimum width = 4.8em, below = 0.5 cm of bl5] (bl51) {\normalfont $\{ {\bfA}_0, {\bfA}_1, {\bfA}_2 \}$}
    node [block, fill=green!30, minimum width = 4.8em, below = 0.005 cm of bl51] (bl52) {\normalfont $\{ {\bfA}_1, {\bfA}_2, {\bfA}_3 \}$}
    node [block, fill=green!30, minimum width = 4.8em, below = 0.5 cm of bl6] (bl61) {\normalfont $\{ {\bfA}_0, {\bfA}_1, {\bfA}_2 \}$}
    node [block, fill=green!30, minimum width = 4.8em, below = 0.5 cm of bl7] (bl71) {\normalfont $\{ {\bfA}_1, {\bfA}_2, {\bfA}_3 \}$}

    % node [block, fill=green!30, minimum width = 4.8em, below = 0.5 cm of bl1] (bl11) {\small $\{\large  \bar{\bfA}_4, \bar{\bfA}_5, \bar{\bfA}_6 \}$}
    % node [block, fill=green!30, minimum width = 4.8em, below = 0.5 cm of bl2] (bl21) {\small$\{ \bar{\bfA}_5, \bar{\bfA}_6, \bar{\bfA}_0 \}$}
    % node [block, fill=green!30, minimum width = 4.8em, below = 0.5 cm of bl3] (bl31) {\small$\{ \bar{\bfA}_6, \bar{\bfA}_0, \bar{\bfA}_1 \}$}
    % node [block, fill=green!30, minimum width = 4.8em, below = 0.5 cm of bl4] (bl41) {\small$\{ \bar{\bfA}_2, \bar{\bfA}_3, \bar{\bfA}_4 \}$}
    % node [block, fill=green!30, minimum width = 4.8em, below = 0.005 cm of bl41] (bl42) {\small$\{ \bar{\bfA}_3, \bar{\bfA}_4, \bar{\bfA}_5 \}$}
    % node [block, fill=green!30, minimum width = 4.8em, below = 0.5 cm of bl5] (bl51) {\small$\{ \bar{\bfA}_0, \bar{\bfA}_1, \bar{\bfA}_2 \}$}
    % node [block, fill=green!30, minimum width = 4.8em, below = 0.005 cm of bl51] (bl52) {\small$\{ \bar{\bfA}_1, \bar{\bfA}_2, \bar{\bfA}_3 \}$}
    % node [block, fill=green!30, minimum width = 4.8em, below = 0.5 cm of bl6] (bl61) {\small$\{ \bar{\bfA}_0, \bar{\bfA}_1, \bar{\bfA}_2 \}$}
    % node [block, fill=green!30, minimum width = 4.8em, below = 0.5 cm of bl7] (bl71) {\small$\{ \bar{\bfA}_1, \bar{\bfA}_2, \bar{\bfA}_3 \}$}

;
\draw[->](bl1) -- node{} (bl11);
\draw[->](bl2) -- node{} (bl21);
\draw[->](bl3) -- node{} (bl31);
\draw[->](bl4) -- node{} (bl41);
\draw[->](bl5) -- node{} (bl51);
\draw[->](bl6) -- node{} (bl61);
\draw[->](bl7) -- node{} (bl71);

\end{tikzpicture}
}

\caption{\small A heterogeneous system where $\bar{n} = 7$ and $\bar{k}_A = 5$, thus $n = 9$ and $k_A = 7$. Each of $W_0$ and $W_1$ is assigned twice the load of each of $W_2, W_3, \dots, W_6$. This system is resilient to any s = 2 block-column processing, in other words, it is resilient to any two type 0 nodes (for example, $W_3$ and $W_6$) or any one type 1 node (for example, $W_1$).
}
\label{hetero_ex}
\vspace{-0.7 cm}
\end{figure}
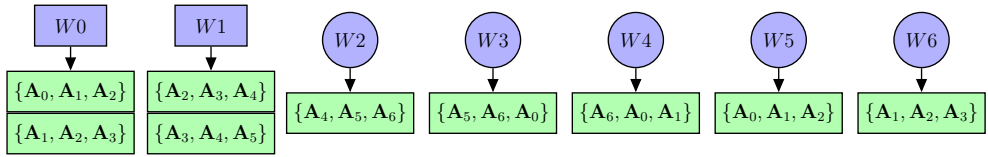
Consider the example in Fig. \ref{hetero_ex} consisting of $\bar{n} = 7$ worker nodes. 
%Let us assume, $W_0$ and $W_1$ are of type $1$ worker nodes each of which receives twice as many coded columns of $\bfA$ than each of $W_2, W_3, W_4, W_5$ and $W_6$, which are of type $0$ nodes.
Let us assume, $c_{j_i} = 2$ when $i = 0, 1$ and $c_{j_i} = 1$ when $2 \leq i \leq 6$, thus, $n = \sum_{i = 0}^{\bar{n} - 1} c_{j_i} = 9$. Now assume that $\bar{k}_A = 5$, thus $k_A = \sum_{i = 0}^{\bar{k}_A - 1} c_{j_i} = 7$, and $s = \sum_{i = \bar{k}_A}^{\bar{n} - 1} c_{j_i} = 2$. So, each weakest device is assigned $1/7$-th fraction of the whole job . This scheme is resilient to any $s = 2$ block-column processing, in other words, it is resilient to any {\it two} type $0$ nodes or any {\it one} type $1$ node.

Moreover, for the heterogeneous setting, our proposed approach provides a $Q/\Delta$ value to be $1$, which indicates that the central node can recover all $\Delta = 7$ unknowns from {\it any} $Q = 7$ block-products. In this example, $W_0$ and $W_1$ are assigned multiple jobs and our proposed approach can efficiently utilize their partial computations if any of them is slower than their rated speed. For instance, assume that $W_6$ is a failure, and $W_0$ is slower and able to compute one (out of two) of the submatrix products while each of $W_1, W_2, \dots, W_5$ completes its respective assigned job. Then, the central node can recover all the submatrix products from the successful worker nodes with the help of the partial computations done by $W_0$. 
\end{example}

\begin{remark}
It is well-known in the cloud computation that the low-cost machines (which are the ``weaker'' ones) are the most probable ones to straggle \cite{tandon2017gradient}. In that case, the number of stragglers that our proposed heterogeneous matrix computation scheme can be resilient to will be higher.
\end{remark}

\section{Proposed Approach for Matrix-matrix Multiplication}
\label{sec:matmat}
% In this section, we detail our approach for straggler resilient distributed matrix-matrix multiplication in case of both homogeneous and heterogeneous worker nodes.

\subsection{Homogeneous System}
\label{sec:homogeneous_matmat}
First, we discuss our distributed matrix-matrix multiplication approach for the homogeneous system with resilience to $s = n - k_A k_B$ stragglers where $s \leq \textrm{max} (k_A, k_B)$. Without loss of generality (w.l.o.g.), we can assume that $k_A \geq k_B$, thus, $s \leq k_A$. Each node stores the equivalent of $1/k_A$ and $1/k_B$ fractions (of block-columns) of matrices $\bfA$ and $\bfB$, respectively. Thus, if a node multiplies its respective assignments from $\bfA$ and $\bfB$, it completes $1/k_{AB}$ fraction of overall job of computing $\bfA^T \bfB$, where $k_{AB} = k_A k_B$. The overall procedure is given in Alg. \ref{Alg:New_matmat}. 

In our approach, we partition matrices $\bfA$ and $\bfB$ into $k_A$ and $k_B$ disjoint block columns, respectively, as $\bfA_0, \bfA_1, \bfA_2, \dots, \bfA_{k_A - 1}$ and $\bfB_0, \bfB_1, \bfB_2, \dots, \bfB_{k_B - 1}$, respectively. Next, we set $\omega_A$ and $\omega_B$ in such a way so that $\omega_A \omega_B > s$; and assign a random linear combination of $\omega_A$ (weight) submatrices of $\bfA$ and another random linear combination of $\omega_B$ (weight) submatrices of $\bfB$ to every worker node where $1 < \omega_A < k_A $ and $1 < \omega_B < k_B $. It should be noted that for a given storage fraction $\gamma_A = 1/k_A$ (or $\gamma_B = 1/k_B$) of matrix $\bfA$ (or $\bfB$) for each of the worker nodes, the case $\omega_A = 1$ (or $\omega_B = 1$) leads to an approach which provides suboptimal performance in terms of number of stragglers that the system is resilient to \cite{das2020coded}. 

\begin{algorithm}[t]
	\caption{Proposed scheme for distributed matrix-matrix multiplication}
	\label{Alg:New_matmat}
   \SetKwInOut{Input}{Input}
   \SetKwInOut{Output}{Output}
   \Input{Matrices $\bfA$ and $\bfB$, $n$-number of worker nodes, $s$-number of stragglers, storage fraction $\gamma_A = \frac{1}{k_A}$ and $\gamma_A = \frac{1}{k_B}$; $s \leq n - k_A k_B \leq \textrm{max} (k_A, k_B)$.}
   Partition $\bfA$ and $\bfB$ into $k_A$ and $k_B$ block-columns, respectively\;
   Create a $n \times k_A$ random matrix $\bfR_A$ with entries  $r^A_{i,j}$, $0\leq i \leq n -1 $ and $0\leq j \leq k_A -1$\;
   Create a $n \times k_B$ random matrix $\bfR_B$ with entries  $r^B_{i,j}$, $0\leq i \leq n -1 $ and $0\leq j \leq k_B -1$\;
   Set weights $\omega_A$ and $\omega_B$ (where $\omega_A \geq \omega_B$) in such a way that $\omega_A \omega_B > s$, $1 < \omega_A < k_A $ and $1 < \omega_B < k_B $\;
   \For{$i\gets 0$ \KwTo $n-1$}{
   Define $T = \left\lbrace i, i+1, \dots, i + \omega_A - 1 \right\rbrace$ (modulo $k_A$)\;
   Create a random linear combination of $\bfA_{q}$'s where $q \in T$, thus $\tilde{\bfA}_i = \sum\limits_{q \in T} r^A_{i,q} \bfA_q$\;
   Set $j = \floor{i/k_A}$, and define $S = \left\lbrace j, j+1, \dots, j + \omega_B - 1 \right\rbrace$ (modulo $k_B$)\;
   Create a random linear combination of $\bfB_{q}$'s where $q \in S$, thus $\tilde{\bfB}_i = \sum\limits_{q \in S} r^B_{i,q} \bfB_q$\;
   The central node assigns encoded submatrices $\tilde{\bfA}_i$ and $\tilde{\bfB}_i$ to worker node $W_i$\;
   Worker node $W_i$ computes $\tilde{\bfA}_i^T \tilde{\bfB}_i$\;
   }
   \Output{The central node recovers $\bfA^T \bfB$ from the fastest $k_A k_B$ worker nodes.}

\end{algorithm}

Formally, we assign a random linear combination of $\bfA_i, \bfA_{i+1}, \dots, \bfA_{i+\omega_A - 1}$ (indices of $\bfA$ are reduced modulo $k_A$) to worker node $W_i$, $0 \leq i \leq n - 1$. Thus, we can say that the participating submatrices of $\bfA$ are shifted in a cyclic manner over all $n$ worker nodes. Next we set $j = \floor{i/k_A}$, and assign $\bfB_j, \bfB_{j+1}, \dots, \bfB_{j+\omega_B - 1}$ (indices of $\bfB$ are reduced modulo $k_B$) to worker node $W_i$. Once the fastest $\tau = k_A k_B$ worker nodes finish and return their computation results, the central node can recover all the unknowns in the form of  $\bfA_u^T \bfB_v$, where $0 \leq u \leq k_A - 1$ and $0 \leq v \leq k_B - 1$.

\begin{example}
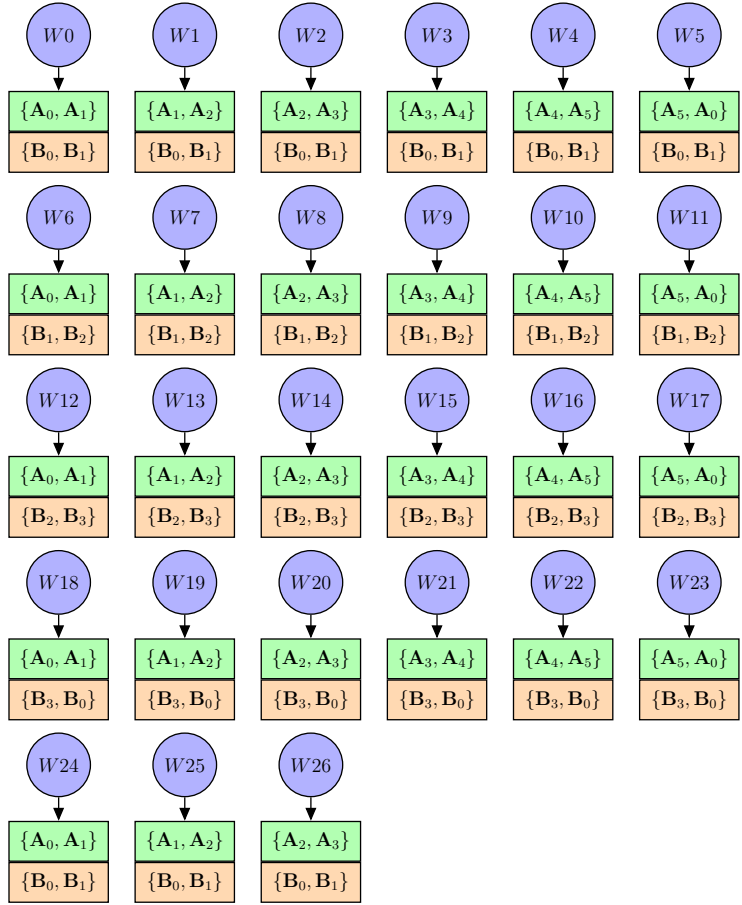
\begin{figure}[t]
\centering
%\captionsetup{justification=centering}
%\resizebox{0.99\linewidth}{!}{
\resizebox{0.6\linewidth}{!}{
\begin{tikzpicture}[auto, thick, node distance=2cm, >=triangle 45]

\draw

    node [sum, minimum size = 1.35cm, fill=blue!30] (blk1) {$W0$}
    node [sum, minimum size = 1.35cm,fill=blue!30,right = 1.3 cm of blk1] (blk2) {$W1$}
    node [sum, minimum size = 1.35cm,fill=blue!30,right = 1.3 cm of blk2] (blk3) {$W2$}
    node [sum, minimum size = 1.35cm,fill=blue!30,right = 1.3 cm of blk3] (blk4) {$W3$}
    node [sum, minimum size = 1.35cm,fill=blue!30,right = 1.3 cm of blk4] (blk5) {$W4$}
    node [sum, minimum size = 1.35cm,fill=blue!30,right = 1.3 cm of blk5] (blk6) {$W5$}
    node [sum, minimum size = 1.35cm,fill=blue!30,below = 2.5 cm of blk1] (blk7) {$W6$}
    node [sum, minimum size = 1.35cm,fill=blue!30,right = 1.3 cm of blk7] (blk8) {$W7$}
    node [sum, minimum size = 1.35cm,fill=blue!30,right = 1.3 cm of blk8] (blk9) {$W8$}
    node [sum, minimum size = 1.35cm,fill=blue!30,right = 1.3 cm of blk9] (blk10) {$W9$}
    node [sum, minimum size = 1.35cm,fill=blue!30,right = 1.3 cm of blk10] (blk11) {$W10$}
    node [sum, minimum size = 1.35cm,fill=blue!30,right = 1.3 cm of blk11] (blk12) {$W11$}
    node [sum, minimum size = 1.35cm,fill=blue!30,below = 2.5 cm of blk7] (blk13) {$W12$}
    node [sum, minimum size = 1.35cm,fill=blue!30,right = 1.3 cm of blk13] (blk14) {$W13$}
    node [sum, minimum size = 1.35cm,fill=blue!30,right = 1.3 cm of blk14] (blk15) {$W14$}
    node [sum, minimum size = 1.35cm,fill=blue!30,right = 1.3 cm of blk15] (blk16) {$W15$}
    node [sum, minimum size = 1.35cm,fill=blue!30,right = 1.3 cm of blk16] (blk17) {$W16$}
    node [sum, minimum size = 1.35cm,fill=blue!30,right = 1.3 cm of blk17] (blk18) {$W17$}
    node [sum, minimum size = 1.35cm,fill=blue!30,below = 2.5 cm of blk13] (blk19) {$W18$}
    node [sum, minimum size = 1.35cm,fill=blue!30,right = 1.3 cm of blk19] (blk20) {$W19$}
    node [sum, minimum size = 1.35cm,fill=blue!30,right = 1.3 cm of blk20] (blk21) {$W20$}
    node [sum, minimum size = 1.35cm,fill=blue!30,right = 1.3 cm of blk21] (blk22) {$W21$}
    node [sum, minimum size = 1.35cm,fill=blue!30,right = 1.3 cm of blk22] (blk23) {$W22$}
    node [sum, minimum size = 1.35cm,fill=blue!30,right = 1.3 cm of blk23] (blk24) {$W23$}
    node [sum, minimum size = 1.35cm,fill=blue!30,below = 2.5 cm of blk19] (blk25) {$W24$}
    node [sum, minimum size = 1.35cm,fill=blue!30,right = 1.3 cm of blk25] (blk26) {$W25$}
    node [sum, minimum size = 1.35cm,fill=blue!30,right = 1.3 cm of blk26] (blk27) {$W26$}
    
    node [block, fill=green!30, minimum width = 5em, below = 0.5 cm of blk1] (blk1_1) {$\left\lbrace\bfA_0, \bfA_1 \right\rbrace$}
    node [block, fill=green!30, minimum width = 5em, below = 0.5 cm of blk2] (blk2_1) {$\left\lbrace\bfA_1, \bfA_2 \right\rbrace$}
    node [block, fill=green!30, minimum width = 5em, below = 0.5 cm of blk3] (blk3_1) {$\left\lbrace\bfA_2, \bfA_3\right\rbrace$}
    node [block, fill=green!30, minimum width = 5em, below = 0.5 cm of blk4] (blk4_1) {$\left\lbrace\bfA_3, \bfA_4 \right\rbrace$}
    node [block, fill=green!30, minimum width = 5em, below = 0.5 cm of blk5] (blk5_1) {$\left\lbrace\bfA_4, \bfA_5\right\rbrace$}
    node [block, fill=green!30, minimum width = 5em, below = 0.5 cm of blk6] (blk6_1) {$\left\lbrace\bfA_5, \bfA_0 \right\rbrace$}
    node [block, fill=green!30, minimum width = 5em, below = 0.5 cm of blk7] (blk7_1) {$\left\lbrace\bfA_0, \bfA_1 \right\rbrace$}
    node [block, fill=green!30, minimum width = 5em, below = 0.5 cm of blk8] (blk8_1) {$\left\lbrace\bfA_1, \bfA_2 \right\rbrace$}
    node [block, fill=green!30, minimum width = 5em, below = 0.5 cm of blk9] (blk9_1) {$\left\lbrace\bfA_2, \bfA_3\right\rbrace$}
    node [block, fill=green!30, minimum width = 5em, below = 0.5 cm of blk10] (blk10_1) {$\left\lbrace\bfA_3, \bfA_4 \right\rbrace$}
    node [block, fill=green!30, minimum width = 5em, below = 0.5 cm of blk11] (blk11_1) {$\left\lbrace\bfA_4, \bfA_5\right\rbrace$}
    node [block, fill=green!30, minimum width = 5em, below = 0.5 cm of blk12] (blk12_1) {$\left\lbrace\bfA_5, \bfA_0 \right\rbrace$}
    node [block, fill=green!30, minimum width = 5em, below = 0.5 cm of blk13] (blk13_1) {$\left\lbrace\bfA_0, \bfA_1 \right\rbrace$}
    node [block, fill=green!30, minimum width = 5em, below = 0.5 cm of blk14] (blk14_1) {$\left\lbrace\bfA_1, \bfA_2 \right\rbrace$}
    node [block, fill=green!30, minimum width = 5em, below = 0.5 cm of blk15] (blk15_1) {$\left\lbrace\bfA_2, \bfA_3\right\rbrace$}
    node [block, fill=green!30, minimum width = 5em, below = 0.5 cm of blk16] (blk16_1) {$\left\lbrace\bfA_3, \bfA_4 \right\rbrace$}
    node [block, fill=green!30, minimum width = 5em, below = 0.5 cm of blk17] (blk17_1) {$\left\lbrace\bfA_4, \bfA_5\right\rbrace$}
    node [block, fill=green!30, minimum width = 5em, below = 0.5 cm of blk18] (blk18_1) {$\left\lbrace\bfA_5, \bfA_0 \right\rbrace$}
    node [block, fill=green!30, minimum width = 5em, below = 0.5 cm of blk19] (blk19_1) {$\left\lbrace\bfA_0, \bfA_1 \right\rbrace$}
    node [block, fill=green!30, minimum width = 5em, below = 0.5 cm of blk20] (blk20_1) {$\left\lbrace\bfA_1, \bfA_2 \right\rbrace$}
    node [block, fill=green!30, minimum width = 5em, below = 0.5 cm of blk21] (blk21_1) {$\left\lbrace\bfA_2, \bfA_3\right\rbrace$}
   node [block, fill=green!30, minimum width = 5em, below = 0.5 cm of blk22] (blk22_1) {$\left\lbrace\bfA_3, \bfA_4 \right\rbrace$}
    node [block, fill=green!30, minimum width = 5em, below = 0.5 cm of blk23] (blk23_1) {$\left\lbrace\bfA_4, \bfA_5\right\rbrace$}
    node [block, fill=green!30, minimum width = 5em, below = 0.5 cm of blk24] (blk24_1) {$\left\lbrace\bfA_5, \bfA_0 \right\rbrace$}
    node [block, fill=green!30, minimum width = 5em, below = 0.5 cm of blk25] (blk25_1) {$\left\lbrace\bfA_0, \bfA_1\right\rbrace$}
    node [block, fill=green!30, minimum width = 5em, below = 0.5 cm of blk26] (blk26_1) {$\left\lbrace\bfA_1, \bfA_2 \right\rbrace$}
    node [block, fill=green!30, minimum width = 5em, below = 0.5 cm of blk27] (blk27_1) {$\left\lbrace\bfA_2, \bfA_3 \right\rbrace$}
    node [block, fill=orange!30, minimum width = 5em, below = 0.0005 cm of blk1_1] (blk1_2) {$\left\lbrace\bfB_0, \bfB_1 \right\rbrace$}
    node [block, fill=orange!30, minimum width = 5em, below = 0.0005 cm of blk2_1] (blk2_2) {$\left\lbrace\bfB_0, \bfB_1 \right\rbrace$}
    node [block, fill=orange!30, minimum width = 5em, below = 0.0005 cm of blk3_1] (blk3_2) {$\left\lbrace\bfB_0, \bfB_1\right\rbrace$}
    node [block, fill=orange!30, minimum width = 5em, below = 0.0005 cm of blk4_1] (blk4_2) {$\left\lbrace\bfB_0, \bfB_1 \right\rbrace$}
    node [block, fill=orange!30, minimum width = 5em, below = 0.0005 cm of blk5_1] (blk5_2) {$\left\lbrace\bfB_0, \bfB_1\right\rbrace$}
    node [block, fill=orange!30, minimum width = 5em, below = 0.0005 cm of blk6_1] (blk6_2) {$\left\lbrace\bfB_0, \bfB_1\right\rbrace$}
    node [block, fill=orange!30, minimum width = 5em, below = 0.0005 cm of blk7_1] (blk7_2) {$\left\lbrace\bfB_1, \bfB_2 \right\rbrace$}
    node [block, fill=orange!30, minimum width = 5em, below = 0.0005 cm of blk8_1] (blk8_2) {$\left\lbrace\bfB_1, \bfB_2 \right\rbrace$}
    node [block, fill=orange!30, minimum width = 5em, below = 0.0005 cm of blk9_1] (blk9_2) {$\left\lbrace\bfB_1, \bfB_2\right\rbrace$}
    node [block, fill=orange!30, minimum width = 5em, below = 0.0005 cm of blk10_1] (blk10_2) {$\left\lbrace\bfB_1, \bfB_2 \right\rbrace$}
    node [block, fill=orange!30, minimum width = 5em, below = 0.0005 cm of blk11_1] (blk11_2) {$\left\lbrace\bfB_1, \bfB_2\right\rbrace$}
    node [block, fill=orange!30, minimum width = 5em, below = 0.0005 cm of blk12_1] (blk12_2) {$\left\lbrace\bfB_1, \bfB_2\right\rbrace$}
    node [block, fill=orange!30, minimum width = 5em, below = 0.0005 cm of blk13_1] (blk13_2) {$\left\lbrace\bfB_2, \bfB_3\right\rbrace$}
    node [block, fill=orange!30, minimum width = 5em, below = 0.0005 cm of blk14_1] (blk14_2) {$\left\lbrace\bfB_2, \bfB_3\right\rbrace$}
    node [block, fill=orange!30, minimum width = 5em, below = 0.0005 cm of blk15_1] (blk15_2) {$\left\lbrace\bfB_2, \bfB_3\right\rbrace$}
    node [block, fill=orange!30, minimum width = 5em, below = 0.0005 cm of blk16_1] (blk16_2) {$\left\lbrace\bfB_2, \bfB_3\right\rbrace$}
    node [block, fill=orange!30, minimum width = 5em, below = 0.0005 cm of blk17_1] (blk17_2) {$\left\lbrace\bfB_2, \bfB_3\right\rbrace$}
    node [block, fill=orange!30, minimum width = 5em, below = 0.0005 cm of blk18_1] (blk18_2) {$\left\lbrace\bfB_2, \bfB_3\right\rbrace$}
    node [block, fill=orange!30, minimum width = 5em, below = 0.0005 cm of blk19_1] (blk19_2) {$\left\lbrace\bfB_3, \bfB_0\right\rbrace$}
    node [block, fill=orange!30, minimum width = 5em, below = 0.0005 cm of blk20_1] (blk20_2) {$\left\lbrace\bfB_3, \bfB_0\right\rbrace$}
    node [block, fill=orange!30, minimum width = 5em, below = 0.0005 cm of blk21_1] (blk21_2) {$\left\lbrace\bfB_3, \bfB_0\right\rbrace$}
    node [block, fill=orange!30, minimum width = 5em, below = 0.0005 cm of blk22_1] (blk22_2) {$\left\lbrace\bfB_3, \bfB_0\right\rbrace$}
    node [block, fill=orange!30, minimum width = 5em, below = 0.0005 cm of blk23_1] (blk23_2) {$\left\lbrace\bfB_3, \bfB_0\right\rbrace$}
    node [block, fill=orange!30, minimum width = 5em, below = 0.0005 cm of blk24_1] (blk24_2) {$\left\lbrace\bfB_3, \bfB_0\right\rbrace$}
    node [block, fill=orange!30, minimum width = 5em, below = 0.0005 cm of blk25_1] (blk25_2) {$\left\lbrace\bfB_0, \bfB_1 \right\rbrace$}
    node [block, fill=orange!30, minimum width = 5em, below = 0.0005 cm of blk26_1] (blk26_2) {$\left\lbrace\bfB_0, \bfB_1 \right\rbrace$}
    node [block, fill=orange!30, minimum width = 5em, below = 0.0005 cm of blk27_1] (blk27_2) {$\left\lbrace\bfB_0, \bfB_1\right\rbrace$}

;
\draw[->](blk1) -- node{} (blk1_1);
\draw[->](blk2) -- node{} (blk2_1);
\draw[->](blk3) -- node{} (blk3_1);
\draw[->](blk4) -- node{} (blk4_1);
\draw[->](blk5) -- node{} (blk5_1);
\draw[->](blk6) -- node{} (blk6_1);
\draw[->](blk7) -- node{} (blk7_1);
\draw[->](blk8) -- node{} (blk8_1);
\draw[->](blk9) -- node{} (blk9_1);
\draw[->](blk10) -- node{} (blk10_1);
\draw[->](blk11) -- node{} (blk11_1);
\draw[->](blk12) -- node{} (blk12_1);
\draw[->](blk13) -- node{} (blk13_1);
\draw[->](blk14) -- node{} (blk14_1);
\draw[->](blk15) -- node{} (blk15_1);
\draw[->](blk16) -- node{} (blk16_1);
\draw[->](blk17) -- node{} (blk17_1);
\draw[->](blk18) -- node{} (blk18_1);
\draw[->](blk19) -- node{} (blk19_1);
\draw[->](blk20) -- node{} (blk20_1);
\draw[->](blk21) -- node{} (blk21_1);
\draw[->](blk22) -- node{} (blk22_1);
\draw[->](blk23) -- node{} (blk23_1);
\draw[->](blk24) -- node{} (blk24_1);
\draw[->](blk25) -- node{} (blk25_1);
\draw[->](blk26) -- node{} (blk26_1);
\draw[->](blk27) -- node{} (blk27_1);

\end{tikzpicture}
}
\caption{\small Submatrix allocation according to Alg. \ref{Alg:New_matmat} when $n = 27$ and $s = 3$, with $\gamma_A = \frac{1}{6}$ and $\gamma_B = \frac{1}{4}$. The weights of the submatrices are $\omega_A = \omega_B = 2$. Any assignment $\{\bfA_i, \bfA_j\}$ or $\{\bfB_i, \bfB_j\}$ indicates a random linear combination of the corresponding submatrices where the coefficients are chosen i.i.d. at random from a continuous distribution.}
\label{fig:matmat27}
    \vspace{-0.8 cm}
\end{figure}

\label{ex:matmat}
Consider the example in Fig. \ref{fig:matmat27} where $n = 27, \gamma_A = 1/6$ and $\gamma_B = 1/4$. So, we partition $\bfA$ and $\bfB$ into $k_A = 6$ and $k_B = 4$ block-columns, respectively. In each node, we assign one coded submatrix from $\bfA$ and one from $\bfB$ which are linear combinations of $\omega_A = \omega_B = 2$ uncoded submatrices with coefficients chosen i.i.d. at random from a continuous distribution. It can be verified that this scheme is resilient to $s = n - k_A k_B = 3$ stragglers. In what follows, we will use this example several times to describe different structures and properties of our scheme.
\end{example}

\subsubsection{Structure of the Job Assignment}
To describe the structure of the proposed scheme, first we partition the worker nodes into $k_A$ disjoint classes, denoted by $\calM_i$'s, where any $\calM_i$ consists of all the worker nodes $W_j$'s if $j \equiv i (\textrm{mod} \; k_A)$. In other words, $\calM_i = \left\lbrace W_i, \, W_{k_A+i}, \, W_{2 k_A+i}, \, \dots  \right\rbrace$, for $i = 0, 1, 2, \dots, k_A - 1$. 
 Since $n = k_A k_B + s \leq k_A (k_B + 1)$ (as $s \leq k_A$), we can say $|\calM_i|$ is either $k_B$ or $k_B + 1$. Moreover, according to our proposed scheme, the participating submatrices of $\bfA$ are the same over all the worker nodes in any $\calM_i$. For instance, in Example \ref{ex:matmat}, we have $\calM_0 =  \left\lbrace W_0, \, W_6, \, W_{12}, \, W_{18}, \, W_{24}   \right\rbrace$, and random linear combinations of $\bfA_0$ and $\bfA_1$ are assigned to all the corresponding worker nodes. At this point, we define a set, $\calD^A_i = \left\lbrace \bfA_i, \bfA_{i+1}, \dots, \bfA_{i + \omega_A - 1} \right\rbrace$, which consists of the participating submatrices of $\bfA$ corresponding to worker node set $\calM_i$, where the indices are reduced modulo $k_A$. Now we state the following claim which gives a lower bound for the cardinality of the union of any arbitrary number of $\calD^A_i$'s.

\begin{claim}
\label{claim:cyclicA}
Consider any $q$ sets $\calD^A_i$'s, $q \leq k_A - \omega_A + 1$, denoted w.l.o.g.,  $\bar{\calD}^A_j$, $0 \leq j \leq q - 1$ arbitrarily. Then $\abs*{\bigcup\limits_{j = 0}^{q-1} \bar{\calD}^A_j} \geq \omega_A + q - 1$.
\end{claim}
\begin{proof}
From the definition of $\calD^A_i$ above, we can say that the participating uncoded  submatrices of $\bfA$ are shifted in a cyclic fashion within $\calD^A_0, \calD^A_1, \dots, \calD^A_{k_A - 1}$. Thus, according to the proof of cyclic scheme in Appendix C in \cite{das2020coded}, the minimum number of total constituent submatrices of $\bfA$ within any $L$ of $\calD_i^A$'s is $\textrm{min}\, (\omega_A + L - 1, k_A)$. Now consider any $q$ arbitrary $\bar{\calD}^A_j$'s, $0 \leq j \leq q - 1$. Thus, $\abs*{\bigcup\limits_{j = 0}^{q-1} \bar{\calD}^A_j} \geq \, \textrm{min}\, (\omega_A + q - 1, k_A) = \omega_A + q - 1$, since $q \leq k_A - \omega_A + 1$.
\end{proof}

\begin{example}
Consider Example \ref{ex:matmat} in Fig. \ref{fig:matmat27} where $k_A = 6$ and $\omega_A = 2$. Now, choose any arbitrary $q \leq 5$ sets of $\calD^A_i$'s. For example, we choose $q = 3$ of $\calD_i^A$'s, such as  $\calD^A_0$, $\calD^A_1$ and $\calD^A_3$ (w.l.o.g. we can denote them as $\bar{\calD}^A_0, \bar{\calD}^A_1$ and $\bar{\calD}^A_2$). Then, according to Claim \ref{claim:cyclicA} we have, $\abs*{\calD^A_0 \cup \calD^A_1 \cup \calD^A_3} = \abs*{\bar{\calD}^A_0 \cup \bar{\calD}^A_1 \cup \bar{\calD}^A_2} \geq \omega_A + 2 = 4$. Thus, the total number of constituent submatrices of $\bfA$ within $\calM_0$, $\calM_1$ and $\calM_3$ is lower bounded by $4$. 
It can be verified from Fig. \ref{fig:matmat27} that 
%submatrices $\bfA_0, \bfA_1, \bfA_2, \bfA_3$ and $\bfA_4$ participate within the worker nodes in $\calM_0$, $\calM_1$ and $\calM_3$, thus 
the exact number is $5$, where $\calD^A_0 = \left\lbrace \bfA_0, \bfA_1 \right\rbrace$, $\calD^A_1 = \left\lbrace \bfA_1, \bfA_2 \right\rbrace$ and $\calD^A_3 = \left\lbrace \bfA_3, \bfA_4 \right\rbrace$.
\end{example}

Now, in our approach, according to Alg. \ref{Alg:New_matmat}, the participating submatrices of $\bfB$ are shifted in a cyclic fashion over the worker nodes of any $\calM_i$. For instance, in Example \ref{ex:matmat}, the participating submatrices, $\bfB_0, \bfB_1, \bfB_2$ and $\bfB_3$, are shifted in a cyclic fashion within the worker nodes of $\calM_0$, i.e., $W_0, W_6, W_{12}, W_{18}$ and $W_{24}$. Next, in the following claim, we find the minimum number of participating unknowns (in the form of  $\bfA_u^T \bfB_v$) within any $\delta$ worker nodes from any $\calM_q$. Now we state the following claim.

\begin{claim}
\label{claim:cyclicB}
Consider $\calM_q$, $0 \leq q \leq k_A - 1$. Denote the minimum of total number of participating unknowns (in the form of  $\bfA_i^T \bfB_j$) within any $\delta$ worker nodes from $\calM_q$ by $\rho$. Then
\begin{align*}
    \rho = \begin{cases}
    \omega_A \times \omega_B & \; \text{if} \; \delta = 1 \; ; \\
      \omega_A \times \textrm{min} \left( \omega_B + \delta - 1, k_B\right) & \; \text{if} \; \abs*{\calM_0} = k_B \; \text{and} \; \delta \geq 2 \; ; \\
    \omega_A \times \textrm{min} \left( \omega_B + \delta - 2, k_B\right) & \; \text{if} \; \abs*{\calM_0} = k_B + 1 \; \text{and} \; \delta \geq 2 \, .
    \end{cases}   
\end{align*}
\end{claim}

\begin{proof}
Consider the case when $\abs*{\calM_q} = k_B$. The participating uncoded submatrices of $\bfB$ are shifted in a cyclic fashion within the worker nodes of $\calM_q$. Thus according to the proof of cyclic scheme in Appendix C in \cite{das2020coded}, the minimum number of constituent submatrices of $\bfB$ within any $\delta$ worker nodes of $\calM_q$ is $\textrm{min}\, (\omega_B + \delta - 1, k_B)$ if $\delta \geq 1$. Now the coded submatrices of $\bfB$ are multiplied by linear combinations of the same $\omega_A$ submatrices, thus the minimum of total number of participating unknowns is $\rho = \omega_A \times \textrm{min} \left( \omega_B + \delta - 1, k_B\right) $. It subsumes the case, $\delta = 1$.

Next, when $\abs*{\calM_q} = k_B + 1$, there are two worker nodes within $\calM_q$ where the participating uncoded submatrices of $\bfB$ are the same. Thus in that case, $\rho = \omega_A \times \omega_B$ for $\delta = 1, 2$. Other than that, similar to the previous case, the participating uncoded submatrices of $\bfB$ are shifted in a cyclic fashion within the worker nodes of $\calM_q$. Thus, $\rho = \omega_A \times \textrm{min} \left( \omega_B + \delta - 2, k_B\right)$, if $\delta \geq 2$.
\end{proof}

% \begin{proof}
% Consider the case when $\abs*{\calM_q} = k_B$. The participating uncoded submatrices of $\bfB$ are shifted in a cyclic fashion within the worker nodes of $\calM_q$. Thus according to the proof of cyclic scheme in Appendix C in \cite{das2020coded}, the minimum number of constituent submatrices of $\bfB$ within any $\delta$ worker nodes of $\calM_q$ is $\textrm{min}\, (\omega_B + \delta - 1, k_B)$ if $\delta \geq 1$. Now the coded submatrices of $\bfB$ are multiplied by linear combinations of the same $\omega_A$ submatrices, thus the minimum of total number of participating unknowns is $\rho = \omega_A \times \textrm{min} \left( \omega_B + \delta - 1, k_B\right) $. It subsumes the case, $\delta = 1$.

% Next, when $\abs*{\calM_q} = k_B + 1$, there are two worker nodes within $\calM_q$ where the participating uncoded submatrices of $\bfB$ are the same. Thus in that case, $\rho = \omega_A \times \omega_B$ for $\delta = 1, 2$. Other than that, similar to the previous case, the participating uncoded submatrices of $\bfB$ are shifted in a cyclic fashion within the worker nodes of $\calM_q$. Thus, $\rho = \omega_A \times \textrm{min} \left( \omega_B + \delta - 2, k_B\right)$, if $\delta \geq 2$.
% \end{proof}

% \begin{proof}
% The proof of this lemma appears in Appendix \ref{App:proof_no_of_unknowns}.
% \end{proof}

\subsubsection{Rearrangement of ${\calM}_i$'s}
\label{subsec:rearrange}
Before moving forward to the necessary theorem and corresponding claims, we discuss a pre-processing step that rearranges the $\calM_i$'s. Choose any arbitrary $m$ worker nodes ($m \leq k_A k_B$), and assume that $\delta_i$ worker nodes have been chosen from $\calM_i$, for $0 \leq i \leq k_A - 1$, so that $\sum_{i = 0}^{k_A - 1} \delta_i  = m$. Now, we rearrange the ${\calM}_i$'s in the following process.

(i) We rearrange the $\delta_i$'s in a decreasing sequence so that $\tilde{\delta}_0 \geq \tilde{\delta}_1 \geq \tilde{\delta}_2 \geq \dots \geq \tilde{\delta}_{k_A - 1}$ and rename the corresponding $\calM_i$'s as $\tilde{\calM}_i$'s so that $\tilde{\delta}_i$ worker nodes have been  chosen from $\tilde{\calM}_i$. 

(ii) If multiple $\delta_i$'s are equal, we place the $\calM_i$'s first which have smaller cardinality. In other words, if $\delta_i = \delta_j$ for any $i,j \leq k_A - 1$, where $\abs*{\calM_i} = k_B$ and $\abs*{\calM_j} = k_B + 1$, then we place $\calM_i$ first (i.e., rename as $\tilde{\calM}_k$ for some $k$) and then we place $\calM_j$ (rename as $\tilde{\calM}_{k+1}$).

Now, we denote $\rho_0$ as the minimum of total number of participating unknowns (in the form of  $\bfA_i^T \bfB_j$) within the $\tilde{\delta}_0 $ worker nodes of $\tilde{\calM}_0$. Thus, according to Claim \ref{claim:cyclicB}, 
\begin{align}
    \label{eq:rhozero}
    \rho_0 = \begin{cases}
          \omega_A \times \omega_B & \; \text{if} \; \tilde{\delta}_0 = 1 \; ; \\
      \omega_A \times \textrm{min} \left( \omega_B + \tilde{\delta}_0 - 1, k_B\right) & \; \text{if} \; \abs*{\calM_q} = k_B \; \text{and} \; \tilde{\delta}_0 \geq 2 \; ; \\
    \omega_A \times \textrm{min} \left( \omega_B + \tilde{\delta}_0 - 2, k_B\right) & \; \text{if} \; \abs*{\calM_q} = k_B + 1 \; \text{and} \; \tilde{\delta}_0 \geq 2 \, .
    \end{cases}  
\end{align}
After that, we move to $\tilde{\calM}_1, \tilde{\calM}_2, \dots, \tilde{\calM}_{k_A - \omega_A}$, sequentially, to find the number of additional participating unknowns within the corresponding $\tilde{\delta}_i $ worker nodes of $\tilde{\calM}_i$, where $1 \leq i \leq k_A - 1$. We denote $\rho_i$ as the minimum number of such additional participating unknowns in $\tilde{\calM}_i$. 

Here, according to Claim \ref{claim:cyclicA}, $\abs*{\bigcup\limits_{j = 0}^{0} \bar{\calD}^A_j} \geq \omega_A$ and $\abs*{\bigcup\limits_{j = 0}^{1} \bar{\calD}^A_j} \geq \omega_A + 1$. Thus,  there will be at least one additional participating submatrix of $\bfA$ in $\tilde{\calM}_0 \cup \tilde{\calM}_1$ in comparison to $\tilde{\calM}_0$, and the property will continue to hold until we consider the set $\tilde{\calM}_0 \cup \tilde{\calM}_1 \cup \dots \cup \tilde{\calM}_0 \cup \tilde{\calM}_{k_A - \omega_A}$. Now, since the submatrices of $\bfB$ (which will be multiplied by the additional submatrix of $\bfA$) are shifted in a cyclic fashion within any $\tilde{\calM}_i$, so we can write
\begin{align}
\label{eq:rhobigger}
    \rho_i = \begin{cases}
     \omega_B & \; \text{if} \; \tilde{\delta}_i = 1 \; ; \\
     \textrm{min} \left( \omega_B + \tilde{\delta}_i - 1, k_B\right)  & \; \text{if} \; \abs*{\tilde{\calM}_i} = k_B \; \text{and} \; \tilde{\delta}_i \geq 2 \; \\
     \textrm{min} \left( \omega_B + \tilde{\delta}_i - 2, k_B\right)  & \; \text{if} \; \abs*{\tilde{\calM}_i} = k_B + 1 \; \text{and} \; \tilde{\delta}_i \geq 2 \, 
    \end{cases}   
\end{align} for $1 \leq i \leq k_A - \omega_A$. Note that, $\rho_i$ has a trivial lower bound, {\it zero}, when $k_A - \omega_A + 1 \leq i \leq k_A - 1$. 

Now we state the following corollary which is a special case of Lemma \ref{lem:no_of_unknowns}  stated later in this section. The lemma (hence, the corollary) provides a lower bound on the minimum number of participating unknowns (in the form of $\bfA^T_u \bfB_v$) in the equations from any arbitrary $m$ nodes. Here, we assume that $k_A \geq k_B$. Note that, if $k_B > k_A$, we can compute $\bfA^T \bfB$ as $\left(\bfB^T \bfA\right)^T$ without any additional computational cost. Thus, we can assume $k_A \geq k_B$ without loss of generality.

\begin{corollary}
\label{cor:toy_example}
(Corollary of upcoming Lemma \ref{lem:no_of_unknowns}) For any arbitrary $k_A \geq 3$ and $k_B \geq 3$ (where $k_A \geq k_B$ without loss of generality), let us assign the jobs to $n = k_A k_B + s$ (where $s \leq 3$) worker nodes according to Alg. \ref{Alg:New_matmat} using $\omega_A = \omega_B = 2$. Then the total number of participating unknowns (in the form of $\bfA_i^T \bfB_j$) within any $m$ worker nodes ($m \leq k_A k_B$) is at least $m$.
\end{corollary}

\begin{proof}
First, we choose any arbitrary $m \leq k_A k_B$ worker nodes out of all $n$ nodes. As discussed in the preprocessing step above, we choose $\tilde{\delta}_i$ worker nodes from the set $\tilde{\calM}_i$, thus, $\sum_{i=0}^{k_A - 1} \tilde{\delta}_i = m$. Moreover, the minimum number of participating unknowns in $\tilde{\calM}_0$ is denoted by $\rho_0$ as given by \eqref{eq:rhozero}. The minimum number of additional participating unknowns from $\tilde{\calM}_i$ is denoted by $\rho_i$ and given by \eqref{eq:rhobigger} when $1 \leq i \leq k_A - \omega_A$ or trivially lower bounded by {\it zero} for other values of $i$ when $k_A - \omega_A + 1 \leq i \leq k_A - 1$. Thus, the minimum number of participating unknowns in these 
$m$ worker nodes is at least $\sum_{i=0}^{k_A - 1} {\rho}_i = \sum_{i=0}^{k_A - \omega_A} {\rho}_i + \sum_{i=k_A - \omega_A + 1}^{k_A - 1} {\rho}_i \geq \sum_{i=0}^{k_A - \omega_A} {\rho}_i$.
In this corollary, $\omega_A = 2$. Thus, in order to prove the corollary, it is sufficient to prove that
\begin{align}
\label{eq:provehall}
    \sum\limits_{i=0}^{k_A - 2} \rho_i \; \geq \; \sum\limits_{i = 0}^{k_A - 1} \tilde{\delta}_i.
\end{align} 
% Now, for this case of $\omega_A = \omega_B = 2$, according to \eqref{eq:rhozero} and \eqref{eq:rhobigger}, we have
% \begin{align}
% \label{eq:rhozero2}
%     \rho_0 = \begin{cases}
%      4 & \; \text{if} \; \tilde{\delta}_0 = 1 \; ; \\
%      \textrm{min} \left( 2 \tilde{\delta}_0 + 2, 2 k_B\right) & \; \text{if} \; \abs*{\tilde{\calM}_0} = k_B \; \text{and} \; \tilde{\delta}_0 \geq 1 \; ; \\
%     \textrm{min} \left( 2\tilde{\delta}_0 , 2 k_B\right) & \; \text{if} \; \abs*{\tilde{\calM}_0} = k_B + 1 \; \text{and} \; \tilde{\delta}_0 \geq 2 \, .
%     \end{cases}   
% \end{align} and, for $1 \leq i \leq k_A - 2$,
% \begin{align}
% \label{eq:rhobigger2}
%     \rho_i = \begin{cases}
%      2 & \; \text{if} \; \tilde{\delta}_i = 1 \; ; \\
%      \textrm{min} \left( \tilde{\delta}_i + 1, k_B\right) & \; \text{if} \; \abs*{\tilde{\calM}_i} = k_B \; \text{and} \; \tilde{\delta}_i \geq 1 \; ; \\
%      \textrm{min} \left( \tilde{\delta}_i, k_B\right) & \; \text{if} \; \abs*{\tilde{\calM}_i} = k_B + 1 \; \text{and} \; \tilde{\delta}_i \geq 2 \; .
%     \end{cases}   
% \end{align} 

Now we carry out the following exhaustive case analysis and we show that for every case, \eqref{eq:provehall} is true. Here, according to the pre-processing step (i) for the rearrangement of $\calM_i$'s in Sec. \ref{subsec:rearrange}, we have $\tilde{\delta}_0 \geq \tilde{\delta}_1 \geq \tilde{\delta}_2 \geq \dots \geq \tilde{\delta}_{k_A - 1}$. Since $n = k_A k_B + s$ and $s \leq 3 \leq k_A$, any $\tilde{\delta}_i$ can have a value at most $k_B + 1$. The cases are based on different values of $\tilde{\delta}_0$: in case 1, we assume $1 \leq \tilde{\delta}_0 \leq k_B$ and in case 2, we assume that $\tilde{\delta}_0 = k_B + 1$. Moreover, since $s \leq 3$ and $m \leq k_A k_B$, there can be at most three $\tilde{\delta}_i$'s which can have value $k_B + 1$. We discuss those in Cases 2a, 2b and 2c, respectively.

{\it Case 1}: $\mathbf{1 \leq \tilde{\delta}_0 \leq k_B}$. If, $\tilde{\delta}_0 = 1$ or $1 < \tilde{\delta}_0 \leq k_B$, since $\omega_A = \omega_B = 2$, according to \eqref{eq:rhozero}, we have, $\rho_0 \geq 2 \tilde{\delta}_0 \geq \tilde{\delta}_0 + \tilde{\delta}_{k_A - 1}$, as $\tilde{\delta}_{k_A - 1} \leq \tilde{\delta}_0$. Now, according to \eqref{eq:rhobigger}, for $1 \leq i \leq k_A - 2$, we have $\rho_i \geq \tilde{\delta}_i$. Thus, $
    \sum\limits_{i=0}^{k_A - 2} \rho_i = \rho_0 + \sum\limits_{i=1}^{k_A - 2} \rho_i \geq  \tilde{\delta}_0 + \tilde{\delta}_{k_A - 1} + \sum\limits_{i=1}^{k_A - 2} \tilde{\delta}_i = \sum\limits_{i = 0}^{k_A - 1} \tilde{\delta}_i$ ; hence, \eqref{eq:provehall} is true.

{\it Case 2}: $\mathbf{\tilde{\delta}_0 = k_B + 1}$. In this case, since $\omega_A = \omega_B = 2$, we have $\rho_0 = 2 k_B$ according to \eqref{eq:rhozero}. Now, for simplicity we consider the following three sub-cases.

{\it Case 2a}: $\tilde{\delta}_0 = k_B + 1$ and $\tilde{\delta}_1 \leq k_B$. Since $\sum\limits_{i = 0}^{k_A - 1} \tilde{\delta}_i = m \leq k_A k_B$ and $\tilde{\delta}_0 = k_B + 1$, $\tilde{\delta}_{k_A - 1}$ can be at most $k_B - 1$. Thus $\rho_0 = 2 k_B \geq \tilde{\delta}_0 + \tilde{\delta}_{k_A - 1}$. Thus similar to Case 1, we can say that \eqref{eq:provehall} is true.

{\it Case 2b}: $\tilde{\delta}_0 = \tilde{\delta}_1 = k_B + 1$ and $\tilde{\delta}_2 \leq k_B$. In this case, $\rho_0 + \rho_1 = 3 k_B$. Now, $m = \sum_{i=0}^{k_A - 1} \tilde{\delta}_i \leq k_A k_B$ which indicates that $\sum_{i=2}^{k_A - 1} \tilde{\delta}_i \leq (k_A-2) k_B - 2$ and we know that $\tilde{\delta}_i$'s are arranged in a non-increasing order. Thus, it can happen that $\tilde{\delta}_{k_A - 1} \leq k_B - 2$ or $\tilde{\delta}_{k_A - 1}= \tilde{\delta}_{k_A - 2} = k_B - 1$. If $\tilde{\delta}_{k_A - 1} \leq k_B - 2$, we have
\begin{align*}
    \tilde{\delta}_0+ \tilde{\delta}_1 + \tilde{\delta}_{k_A - 1} + \sum_{i=2}^{k_A - 2} \tilde{\delta}_i = k_B + 1 + k_B + 1 + k_B - 2 + \sum\limits_{i=2}^{k_A - 2} \tilde{\delta}_i \leq \rho_0 + \rho_1 + \sum\limits_{i=2}^{k_A - 2} \rho_i ,
\end{align*} where the last inequality follows from \eqref{eq:rhobigger}. Hence, we are done since $\sum\limits_{i=0}^{k_A - 2} \rho_i \; \geq \; \sum\limits_{i = 0}^{k_A - 2} \tilde{\delta}_i$, as we know $\rho_i \geq \tilde{\delta}_i$ for $i=2, 3, \dots, k_A - 2$. 

The remaining case is, $\tilde{\delta}_{k_A - 1}= \tilde{\delta}_{k_A - 2} = k_B - 1$. Now, since the corollary aims at resilience to at most {\it three} stragglers, there can be at most {\it three} $\tilde{\calM}_i$'s which have cardinality $k_B + 1$. But, $\tilde{\delta}_0 = \tilde{\delta}_1 = k_B + 1$, and thus, there can be at most {\it one} more $\tilde{\calM}_i$ left with cardinality $k_B + 1$. So, either both of $\tilde{\calM}_{k_A - 2}$ and $\tilde{\calM}_{k_A - 1}$ will have cardinality $k_B$ or one of them will have cardinality $k_B + 1$. However, according to our rearrangement procedure (ii) in  Sec. \ref{subsec:rearrange}, if two $\tilde{\delta}_i$'s are equal, we place the $\tilde{\calM}_i$ which has smaller cardinality, first. Thus, in both cases, $\tilde{\calM}_{k_A - 2}$ must have cardinality $k_B$, hence $\rho_{k_A - 2} = k_B$ according to 
\eqref{eq:rhobigger}. Thus,
\begin{align*}
\rho_0 + \rho_1 + \rho_{k_A - 2} = 4 k_B = \tilde{\delta}_0 + \tilde{\delta}_1 + \tilde{\delta}_{k_A - 2} + \tilde{\delta}_{k_A - 1},
\end{align*}hence we are done similar to Case 1, since for $2 \leq i \leq k_A - 3$, we have $\rho_i \geq \tilde{\delta}_i$.

{\it Case 2c}: $\tilde{\delta}_0 = \tilde{\delta}_1 = \tilde{\delta}_2 = k_B + 1$. In this case, $\rho_0 + \rho_1 + \rho_2 = 4 k_B$. Note that $m \leq k_A k_B$, thus  $\tilde{\delta}_{k_A - 1}$ can be at most $k_B - 1$. Consider the scenario, when $\tilde{\delta}_{k_A - 1} \leq k_B - 3$. In this scenario,
\begin{align*}
    \tilde{\delta}_0 + \tilde{\delta}_1 + \tilde{\delta}_2 + \tilde{\delta}_{k_A - 1} \leq k_B + 1 + k_B + 1 + k_B + 1 + k_B - 3 = 4 k_B = \rho_0 + \rho_1 + \rho_2,
\end{align*} hence we are done similar to Case 1, since for $3 \leq i \leq k_A - 2$, we have $\rho_i \geq \tilde{\delta}_i$. If $\tilde{\delta}_{k_A - 1} = k_B - 2$, then $k_B - 2 \leq \tilde{\delta}_{k_A - 2} \leq k_B - 1$, in that case $\rho_{k_A - 2} \geq  1 + \tilde{\delta}_{k_A - 2}$ according to \eqref{eq:rhobigger} since $|\tilde{\calM}_{k_A - 2}| = k_B$. Thus, $\tilde{\delta}_0 + \tilde{\delta}_1 + \tilde{\delta}_2 + \tilde{\delta}_{k_A - 2} + \tilde{\delta}_{k_A - 1}  \leq 3 (k_B + 1) + \rho_{k_A - 2} - 1 + k_B - 2  = 4 k_B + \rho_{k_A - 2}$. So, $ \tilde{\delta}_0 + \tilde{\delta}_1 + \tilde{\delta}_2 + \tilde{\delta}_{k_A - 2} + \tilde{\delta}_{k_A - 1} \leq \rho_0 + \rho_1 + \rho_2 + \rho_{k_A - 2}$. Hence, we are done. Finally, if $\tilde{\delta}_{k_A - 1} = k_B - 1$, since $\sum_{i=0}^{k_A - 1} \tilde{\delta}_i \leq k_A k_B$ and $\tilde{\delta}_i$'s are arranged in a non-increasing order, we must have $\tilde{\delta}_{k_A - 2} = \tilde{\delta}_{k_A - 3}  = k_B - 1$. Thus we can write
\begin{align*}
\rho_0 + \rho_1 + \rho_2 + \rho_{k_A - 3} + \rho_{k_A - 2} = 4k_B + 2 k_B = 6 k_B = \tilde{\delta}_0 + \tilde{\delta}_1 + \tilde{\delta}_2 + \tilde{\delta}_{k_A - 3} + \tilde{\delta}_{k_A - 2} + \tilde{\delta}_{k_A - 1},
\end{align*}hence we are done, since for $3 \leq i \leq k_A - 4$, we have $\rho_i \geq \tilde{\delta}_i$.
\end{proof}

\begin{lemma}
\label{lem:no_of_unknowns}
For any arbitrary $k_A \geq 3$ and $k_B \geq 3$ (where $k_A \geq k_B$ without loss of generality), if we assign the jobs to $n = k_A k_B + s$ worker nodes (where $s \leq k_A$) according to Alg. \ref{Alg:New_matmat}, then the minimum of total number of participating unknowns within any $m$ worker nodes ($m \leq k_A k_B$) will be lower bounded by $m$.
\end{lemma}
\begin{proof}
The proof of this lemma appears in Appendix \ref{App:proof_no_of_unknowns}.
\end{proof}

\begin{theorem}
\label{thm:matmat}
Assume that a system has $n$ worker nodes each of which can store the equivalent of $1/k_A$ fraction of matrix $\bfA$ and $1/k_B$ fraction of matrix $\bfB$ (without loss of generality $k_A \geq k_B$) for distributed matrix-matrix multiplication $\mathbf{A}^T \bfB$. If we assign the jobs according to Alg. \ref{Alg:New_matmat}, we achieve resilience to {\it any} $s = n - k_A k_B$ stragglers where $s \leq k_A$.
\end{theorem}

\begin{proof}
Since we have partitioned the matrices $\bfA$ and $\bfB$ into $k_A$ and $k_B$ disjoint block-columns, respectively, to recover $\bfA^T \bfB$, we need to decode all $k_A k_B$ block-matrix unknowns, in the form of $\bfA_i^T \bfB_j$, where $0 \leq i \leq k_A - 1$ and $0 \leq j \leq k_B - 1$. We denote the set of these $k_A k_B$ unknowns as $\calU$. Now we choose an arbitrary set of $k_A k_B$ worker nodes. Each of these worker nodes corresponds to an equation in terms of $\omega_A \omega_B$ of those $k_A k_B$ unknowns. We denote the set of $k_A k_B$ equations as $\calV$, and thus, $|\calU| = |\calV| = k_A k_B$. 

Now we consider a bipartite graph $\calG = \calV \cup \calU$, where any vertex (equation) in $\calV$ is connected to $\omega_A \omega_B$ vertices (unknowns) in $\calU$ which have participated in the corresponding equation. Thus, each vertex in $\calV$ has a neighborhood of cardinality $\omega_A \omega_B$ in $\calB$. 
Similar to the proof of Theorem 1, our goal is to show that there exists a perfect matching among the vertices of $\calV$ and $\calU$. 
We argue this according to Hall's marriage theorem for which we need to show that for any $\bar{\calV} \subseteq \calV$, the cardinality of the neighbourhood of $\bar{\calV}$, denoted as $\calN (\bar{\calV}) \subseteq \calU$, is at least as large as $|\bar{\calV}|$. In this case, according to Lemma 1, for $|\bar{\calV}| = m \leq k_A k_B$, we have shown that $|\calN (\bar{\calV})| \geq m$. Thus, there exists a perfect matching among the vertices of $\calV$ and $\calU$.

Again, similar to the proof of Theorem 1, we consider the largest matching where the vertex $v_i \in \calV$ is matched to the vertex $u_j \in \calU$, which indicates that $u_j$ participates in the equation corresponding to $v_i$. Let us consider the $k_A k_B \times k_A k_B$ system matrix where row $i$ corresponds to the equation associated to $v_i$ where $u_j$ participates. 
Now following the same steps of the proof of Theorem 1, we prove that the corresponding $k_A k_B \times k_A k_B$ system matrix is full rank.
%Let us replace row $i$ of the system matrix by $\bfe_j$ where $\bfe_j$ is a unit row-vector of length $k_A k_B$ with the $j$-th entry being $1$, and $0$ otherwise. Thus we have a $k_A k_B \times k_A k_B$ matrix where each row has only one non-zero entry which is $1$. Since we have a perfect matching, we can say that this $k_A k_B\times k_A k_B$ matrix will have only one non-zero entry in every column. This is a permutation of the identity matrix, and thus, is full rank. Since the matrix is full rank for a choice of definite values, according to Schwartz-Zippel lemma \cite{schwartz1980fast}, the matrix continues to be full rank for random choices of non-zero entries. 
Thus, the central node can recover all $k_A k_B$ unknowns from any set of $k_A k_B$ worker nodes. Hence, the scheme is resilient to {\it any} $s = n - k_A k_B$ stragglers.
\end{proof}

\subsection{Extension to Heterogeneous System}
\label{sec:hetero_mm}
Similar to the matrix-vector case in Sec. \ref{sec:hetero_mv}, we extend our approach in Alg. \ref{Alg:New_matmat} to heterogeneous system where the worker nodes may have different computation and communication speeds. We have all the same assumptions as we had in the matrix-vector case in Sec. \ref{sec:hetero_mv}. We have $\lambda$ different types of devices in the system, with worker node type $0, 1, \dots, \lambda - 1$. Any worker node $W_i$ (for $0 \leq i \leq n - 1$) receives $c_{j_i} \alpha_A$ columns of matrix $\bfA$ and $c_{j_i} \alpha_B$ columns of matrix $\bfB$ where any worker node of the weakest type receives $\alpha_A$ and $\alpha_B$ columns, respectively, and $c_{j_i} \geq 1$ is a positive integer. Moreover any worker node $W_i$ of node type $j_i$ has a computation speed $c_{j_i} \beta$, where $\beta$ is the computation speed for the worker node of the weakest type.

As we have discussed for the matrix-vector case in Sec. \ref{sec:hetero_mv}, from the computation and storage perspective, $W_i$ can be considered as a collection of $c_{j_i} \geq 1$ worker nodes of the ``weakest'' type. Thus, $\bar{n}$  worker nodes in the heterogeneous system can be thought as homogeneous system of $n = \sum_{i = 0}^{\bar{n} - 1} c_{j_i}$ worker nodes of the ``weakest'' type. Now, for any worker node index $\bar{k}_{AB}$ (such that $0 \leq \bar{k}_{AB} \leq \bar{n} - 1$), we define $k_{AB} = \sum_{i = 0}^{\bar{k}_{AB} - 1} c_{j_i}$ and $s = \sum_{i = \bar{k}_{AB}}^{\bar{n} - 1} c_{j_i}$, so, $ n = \sum_{i = 0}^{\bar{n} - 1} c_{j_i} = k_{AB} + s$. Thus, a heterogeneous system of $\bar{n}$ worker nodes can be thought as a homogeneous system $n = k_{AB} + s$ nodes, for any $\bar{k}_{AB}$ ($0 \leq \bar{k}_{AB} \leq \bar{n} - 1$). Now we state the following corollary (of Theorem \ref{thm:matmat}) for heterogeneous system.

\begin{corollary}
 \label{cor:matmat}
Consider a heterogeneous system of $\bar{n}$ nodes of different types for distributed matrix-matrix multiplication and assume any $\bar{k}_{AB}$ (where $0 \leq \bar{k}_{AB} \leq \bar{n} - 1$). Now, if the jobs are assigned to the modified homogeneous system of $n = k_{AB} + s$ ``weakest'' type worker nodes according to Alg. \ref{Alg:New_matmat} where $k_{AB} = k_A k_B$, (a) the system will be resilient to $s$ such nodes and (b) will provide a $Q/\Delta$ value as $1$.
\end{corollary}

\begin{proof}
We modify the heterogeneous system of $\bar{n}$ worker nodes to a homogeneous system of $n = k_{AB} + s$ worker nodes all of which are of the ``weakest'' type. Now, we assign the jobs to the homogeneous system according to Alg. 2 and each node is assigned $1/k_{AB}$-th fraction of the overall computational load. Thus, according to Theorem 2, we will have resilience to {\it any} $s$ nodes, which concludes the proof of part (a).

Next, according to Alg. 2, we partition matrix $\bfA = [ \bfA_0 \;\;\; \bfA_1 \;\;\; \dots \;\;\; \bfA_{k_A - 1}]$, and $\bfB = [ \bfB_0 \;\;\; \bfB_1 \;\;\; \dots \;\;\; \bfB_{k_B - 1} ]$, thus, we have total $\Delta = k_A k_B$ submatrix products to be decoded. Now according to part (a) of this proof, in the modified homogeneous system of $n$ worker nodes, we have assigned, in total, $n$ tasks, and the scheme is resilient to any $s$ tasks. Thus, all the unknowns can be recovered from any $n - s = k_A k_B$ submatrix products, hence $Q = k_A k_B$. Thus, we have $Q/\Delta = 1$, which concludes the proof of part (b).
\end{proof}

% \begin{proof}
% We modify the heterogeneous system of $\bar{n}$ worker nodes to a homogeneous system of $n = k_{AB} + s$ worker nodes all of which are of the ``weakest'' type. Now, we assign the jobs to the homogeneous system according to Alg. \ref{Alg:New_matmat} and each node is assigned $1/k_{AB}$-th fraction of the overall computational load. Thus, according to Theorem \ref{thm:matmat}, we will have resilience to {\it any} $s$ nodes, which concludes the proof of part (a).

% Next, according to Alg. \ref{Alg:New_matmat}, we partition matrix $\bfA = [ \bfA_0 \;\;\; \bfA_1 \;\;\; \dots \;\;\; \bfA_{k_A - 1}]$, and $\bfB = [ \bfB_0 \;\;\; \bfB_1 \;\;\; \dots \;\;\; \bfB_{k_B - 1} ]$, thus, we have total $\Delta = k_A k_B$ submatrix products to be decoded. Now according to part (a) of this proof, in the modified homogeneous system of $n$ worker nodes, we have assigned, in total, $n$ tasks, and the scheme is resilient to any $s$ tasks. Thus, all the unknowns can be recovered from any $n - s = k_A k_B$ submatrix products, hence $Q = k_A k_B$. Thus, we have $Q/\Delta = 1$, which concludes the proof of part (b).
% \end{proof}

\section{Properties of Our Proposed Schemes}
\label{sec:property}
\subsection{Computational Complexity for a Worker Node} 
\label{sec:compcomplexity}
Consider random sparse matrices $\bfA \in \mathbb{R}^{t \times r}$ and $\bfB \in \mathbb{R}^{t \times w}$ where the probability that any entry is non-zero is $\eta$. In our proposed approach in Alg. \ref{Alg:New_matmat}, the respective weights of the assigned submatrices of $\bfA$ and $\bfB$ are $\omega_A$ and $\omega_B$. Thus, when $\eta$ is small, the probability of any entry in an encoded submatrix of $\bfA$ to be non zero is
\begin{align}
\label{eq:prob}
    1 - \prod_{i = 1}^{\omega_A} (1 - \eta) = 1 - (1 - \eta)^{\omega_A} \approx 1 - \left(1 -\omega_A \eta  \right) = \omega_A \eta ; %\textrm{when $\eta$ is small}.
\end{align}
%Besides, in our approach, we suggest  to use the minimum $\omega_A$ and $\omega_B$ which satisfy the inequality $\omega_A \omega_B > s$. Thus, even in a large system with $n = 100$ nodes with $s = 10$ stragglers (where $k_A = 10$ and $k_B = 9$), we just set $\omega_A = 4$ and $\omega_B = 3$. Thus, with our approach, $\omega_A \eta$ (or $\omega_B \eta$) is reasonably small.
For example, (i) if $\eta = 0.01$ and $\omega_A = 6$, we have $1 - (1 - \eta)^{\omega_A} = 1 - (1 - 0.01)^6 = 0.0585 \approx \omega_A \eta$, or (ii) if $\eta = 0.02$ and $\omega_B = 4$, we have $1 - (1 - \eta)^{\omega_B} = 1 - (1 - 0.02)^4 = 0.077 \approx \omega_B \eta$. Thus, in any encoded submatrix of $\bfA$ or $\bfB$, the probability of any entry being non-zero is can be approximated by $\omega_A \eta$ or $\omega_B \eta$, respectively. 

In this work, we consider the scenarios where $\bfA$ and $\bfB$ are sparse, hence $\eta$ is small. Therefore, the computational complexity for any worker node is $\calO \left( \omega_A \eta \times \omega_B \eta \times t \times \frac{r w}{k_A k_B} \right) = \calO \left(\omega_A \omega_B \eta^2 \times \frac{r w t}{k_A k_B} \right)$. On the other hand, in a dense coded approach  \cite{yu2017polynomial, 8849468} which assign linear combination of $k_A$ and $k_B$ submatrices, the computational complexity is approximately $\calO \left( k_A \eta \times k_B \eta t \times \frac{r w}{k_A k_B} \right) = \calO \left( \eta^2 \times r w t \right)$, which is $\frac{k_A k_B}{\omega_A \omega_B}$ times larger than our proposed method, since $\omega_A < k_A$ and $\omega_B < k_B$.

The recent approach proposed in \cite{dasunifiedtreatment} can deal with sparse matrices with lesser per worker node computational complexity in comparison to the approaches in \cite{yu2017polynomial, das2020coded, 8849468, 8919859}. The coding scheme in \cite{dasunifiedtreatment} sets the weight of matrix $\bfB$ as $\zeta$ which is given by $\zeta \geq 1 + k_B - \Bigl\lceil{ \frac{k_B}{c} \Bigr\rceil}$ where $c = 1 + \Bigl\lceil{\frac{s}{k_B}\Bigr\rceil}$. In that case, the per worker node computational complexity is
$O\left( \eta^2 \times rwt \times \left( \frac{\zeta}{n} + \frac{\zeta s}{n k_B}\right)\right)
$. Thus in order to compare our approach against \cite{dasunifiedtreatment}, we consider the ratio
\begin{align}
\label{eq:ratio}
    \frac{\left( \frac{\zeta}{n} + \frac{\zeta s}{n k_B}\right)}{\frac{\omega_A \omega_B}{k_A k_B}} = \frac{\zeta k_A (k_B + s)}{n \omega_A \omega_B} = \frac{k_A (k_B + s)}{n} \times \frac{\zeta}{\omega_A \omega_B}.
\end{align} Thus, our proposed approach involves less computational complexity whenever $\frac{k_A (k_B + s)}{n} \times \frac{\zeta}{\omega_A \omega_B} > 1$. In the following, we discuss such examples.

\vspace{-0.1 cm}
\begin{example}
Consider a scenario, where $k_A =8, k_B = 6$ and $s = 3$. In our approach, we set $\omega_A = \omega_B = 2$. On the other hand, the approach in \cite{dasunifiedtreatment} sets $c = 2$ and $\zeta \geq 4$. Thus, according to \eqref{eq:ratio}, the ratio of the per worker node computational complexity for the approach in \cite{dasunifiedtreatment} and that in this work is $72/51$, which indicates a $30\%$ reduction in our case compared to the method in \cite{dasunifiedtreatment}.   
\end{example}

\vspace{-0.2 cm}
% \begin{example}
% Consider another scenario, where $k_A =10, k_B = 8$ and $s = 5$. In that case, $\omega_A = 3$, $\omega_B = 2$, $c = 2$ and $\zeta \geq 5$. Thus, the ratio in \eqref{eq:ratio} is $65/51$, which indicates a $21\%$ reduction in the per worker node computational complexity for our approach compared to the method in \cite{dasunifiedtreatment}.
% \end{example}

\begin{remark}
The weights in our proposed approach depend on the number of stragglers ($s$) since we just need to satisfy the inequality, $\omega_A \omega_B > s$. On the other hand, in the approach in \cite{dasunifiedtreatment}, the corresponding weights depend on $k_A$ and $k_B$. Thus, for fixed $s$, in our approach, the weights remain fixed, whereas for the method in \cite{dasunifiedtreatment}, the weights increase with the increase of $k_A$ and $k_B$.
\end{remark}

\vspace{-0.1 cm}

\begin{example}
Consider the cases where $k_A$ and $k_B$ are even and $k_A = k_B \geq 6$,  $s = 5$. In those cases, in our proposed approach, we set $\omega_A  = 3$ and $\omega_B = 2$. Thus, the ratio in \eqref{eq:ratio} becomes $\frac{k_A (k_A + s)}{k_A^2 + s} \times \frac{\zeta}{6}$. Now assume that $k_A = k_B = 8$. In that case, this ratio is $1.25$, which indicates a $20\%$ reduction in the per worker node computational complexity for our approach compared to the method in \cite{dasunifiedtreatment}. Now, with the increase of $k_A = k_B$ (even), this ratio in \eqref{eq:ratio} will increase (since, $\frac{k_A (k_A + s)}{k_A^2 + s}$ will saturate to $1$ and $\zeta$ will increase), and thus, the gain of our method will be more significant for large $n$. For instance, when $k_A = k_B = 10$, the ratio is $1.43$, which indicates a $30\%$ gain over the approach in \cite{dasunifiedtreatment}. When $k_A = k_B = 12$, the ratio is $1.60$, and so on. 
\end{example}

% \begin{remark}
% \label{rem:mv}
% Our proposed Alg. \ref{Alg:New_matmat} leads to Alg. \ref{Alg:New_matvec} if we set $k_B = 1$ for matrix-vector multiplication. In that case, the computational complexity for any worker node is $\calO \left( (s+1) \eta t \times \frac{r}{k_A} \right) = \calO \left((s + 1) \eta \times \frac{rt}{k_A} \right)$. While it is significantly smaller than that of dense coded approaches, the approach in \cite{dasunifiedtreatment} has a per worker node computational complexity $\calO \left((s + 1) \eta \times \frac{rt}{n} \right)$, slightly smaller than our proposed method (since $n = k_A + s$). 
% %However, it assigns a number of uncoded submatrices to the worker nodes, hence the worker nodes can have exact information about the data matrix $\bfA$. Thus it does not address the security issue in distributed matrix computations. Our method can control the amount of information leakage as we discuss in Sec. \ref{sec:security}. 
% \end{remark}

\vspace{-0.6 cm}
\subsection{Numerical Stability and Coefficient Determination Time}
\label{sec:trialtime}
The condition number is often considered as an important metric  for the numerical stability of a linear system \cite{das2019random, 8849468, 8919859}. In distributed computation, for a system of $n$ workers and $s$ stragglers, we define the worst case condition number ($\kappa_{worst}$) in the homogeneous system  as the maximum of the condition numbers of the decoding matrices over all different choices of $s$ stragglers. In the approaches where random coding is involved \cite{das2020coded, das2019random, 8919859}, the idea is to generate random coefficients several times (say, $20$ trials), and keep the set of coefficients which provides the minimum $\kappa_{worst}$. 

In this work, for the proofs of Theorems \ref{thm:matvec} and \ref{thm:matmat}, we need the coefficients to be chosen i.i.d. at random from a continuous distribution. Now, we briefly explore a few continuous distributions with different parameters to observe the effect on the worst case condition number ($\kappa_{worst}$). We consider (i) Gaussian distribution (denoted as ``rand$(c,d)$'' if the mean is $c$ and the standard deviation is $d$) and (ii) uniform distribution (denoted as ``unifrand($lb,ub$)'' if the lower and the upper endpoints are $lb$ and $ub$, respectively), for a distributed system of $n = 30$ nodes and $s = 2$ stragglers for both distributed matrix-vector and matrix-matrix multiplication. 
Table \ref{table:dist} gives the results, where we observe that $\kappa_{worst}$ remains within an order of magnitude for the different distributions on each type of multiplication. 
In our numerical simulations in Sec. \ref{sec:numexp}, we draw the linear coefficients at i.i.d. from the standard normal distribution, since it gives the best performance out of the distributions considered here.\footnote{Table II implies that the impact of the distribution parameters on the worst case condition number can be non-monotonic. Since our focus in this work is to provide resilience to the maximum number of stragglers irrespective of the coefficient distribution, we leave a more comprehensive investigation of how the choice of distribution impacts numerical stability to future work.}

\begin{table}[t]
\caption{{\small Comparison of the worst case condition number, $\kappa_{worst}$ for different distributions for distributed computation over $n = 30$ nodes with $s = 2$ stragglers}}
\vspace{-0.2 cm}
\label{table:dist}
\begin{center}
\begin{small}
\begin{sc}
\begin{tabular}{c c c}
\hline
\toprule
% \multirow{2}{*}{Methods} & \multicolumn{5}{c}{Worst case condition number, $\kappa_{worst}$}   \\ \cline{2-6} 
% &  rand(0,1) &  rand(0,5) & unifrand(0,1) & unifrand(-1,1) & unifrand(-5,5)
% \\
%  \midrule
% Matrix-vector & $2.43 \times 10^4$ & $3.61 \times 10^4$  & $2.87 \times 10^4$ & $2.67 \times 10^4$ & $2.84 \times 10^4$ \\
% Matrix-matrix & $1.10 \times 10^4$ & $1.63 \times 10^4$  & $1.45 \times 10^4$ & $1.21 \times 10^4$ & $1.38 \times 10^4$ \\

\multirow{2}{*}{Distributions} & \multicolumn{2}{c}{Worst case condition number, $\kappa_{worst}$}   \\ \cline{2-3} 
&  Matrix-vector &  Matrix-matrix \\ 
\midrule
rand(0,0.5) &  $5.53 \times 10^4$ & $1.60 \times 10^4$ \\
rand(0,1) &  $2.43 \times 10^4$ & $1.10 \times 10^4$ \\
rand(0,5) &  $3.61 \times 10^4$ & $1.63 \times 10^4$ \\
unifrand(0,1) &  $2.87  \times 10^4$ & $1.45 \times 10^4$ \\
unifrand(-1,1) &  $2.67 \times 10^4$ & $1.21 \times 10^4$ \\
unifrand(-5,5) &  $2.84 \times 10^4$ & $1.38 \times 10^4$ \\
 
% &  rand(0,1) &  rand(0,5) & unifrand(0,1) & unifrand(-1,1) & unifrand(-5,5)
% \\
% Matrix-vector & $2.43 \times 10^4$ & $3.61 \times 10^4$  & $2.87 \times 10^4$ & $2.67 \times 10^4$ & $2.84 \times 10^4$ \\
% Matrix-matrix & $1.10 \times 10^4$ & $1.63 \times 10^4$  & $1.45 \times 10^4$ & $1.21 \times 10^4$ & $1.38 \times 10^4$ \\
\bottomrule
\end{tabular}
\end{sc}
\end{small}
\end{center}
\vspace{-0.5 cm}
\end{table}%

% \begin{table}[t]
% \caption{{\small Comparison of the worst case condition number, $\kappa_{worst}$ for different distributions for distributed computation over $n = 30$ nodes with $s = 2$ stragglers}}
% \vspace{-0.4 cm}
% \label{table:dist}
% \begin{center}
% \begin{small}
% \begin{sc}
% \begin{tabular}{c c c c c c}
% \hline
% \toprule
% % \multirow{2}{*}{Methods} & \multicolumn{5}{c}{Worst case condition number, $\kappa_{worst}$}   \\ \cline{2-6} 
% % &  rand(0,1) &  rand(0,5) & unifrand(0,1) & unifrand(-1,1) & unifrand(-5,5)
% % \\
% %  \midrule
% % Matrix-vector & $2.43 \times 10^4$ & $3.61 \times 10^4$  & $2.87 \times 10^4$ & $2.67 \times 10^4$ & $2.84 \times 10^4$ \\
% % Matrix-matrix & $1.10 \times 10^4$ & $1.63 \times 10^4$  & $1.45 \times 10^4$ & $1.21 \times 10^4$ & $1.38 \times 10^4$ \\

% \multirow{2}{*}{Case}  &  rand(0,1) & rand(0,5) & urand(0,1) & urand(-1,1) & urand(-5,5) \\
%  &  $(\times 10^4)$ & $(\times 10^4)$ & $(\times 10^4)$& $(\times 10^4)$ & $(\times 10^4)$\\ 
% \midrule
% MV &  $2.43$ & $3.61 $ & $2.87 $ & $2.67 $ & $2.84$ \\
% MM &  $1.10 $ & $1.63 $ & $1.45$ & $1.21$& $1.38$\\
% \bottomrule
% \end{tabular}
% \end{sc}
% \end{small}
% \end{center}
% \vspace{-0.8 cm}
% \end{table}%

In our proposed matrix-matrix multiplication approach in Alg. \ref{Alg:New_matmat}, we partition $\bfA$ and $\bfB$ into $k_A$ and $k_B$ block-columns, respectively, and we have the recovery threshold $\tau = k_A k_B$. Thus, in every trial we need to determine ${n \choose \tau}$ condition numbers of $\tau \times \tau$ sized (decoding) matrices, which has a total complexity of $\calO\left( {n \choose \tau} \tau^3\right)$. On the other hand, the recent sparse matrix computation approaches in \cite{das2020coded, dasunifiedtreatment} partition matrix $\bfA$ into $\Delta_A = \textrm{LCM}(n, k_A)$ block-columns. Thus, in every trial, they need to determine ${n \choose \tau}$ condition numbers of $\Delta_A k_B \times \Delta_A k_B$ sized matrices which has a total complexity of $\calO\left( {n \choose \tau} \Delta_A^3 k_B^3\right)$. Since, $\Delta_A$ can be significantly larger than $k_A$, every trial involves considerably more complexity in comparison to ours. For example, consider a case where $n$ and $k_A$ are coprime. In that case $\Delta_A = n k_A$, and thus the complexity corresponding to the approaches \cite{das2020coded, dasunifiedtreatment} are around $\calO\left( n^3 \right)$ times higher than ours. Similar result holds for the matrix-vector multiplication case in Alg. \ref{Alg:New_matvec}. 
Note that the approach in \cite{8919859} involves similar computational complexity per trial as ours, and the approaches in \cite{yu2017polynomial, 8849468} do not require such coefficient search; however, these methods have very high computational complexity per worker node for sparse matrices, as discussed in Sec. \ref{sec:compcomplexity}.

\begin{remark}
\label{rem:mv}
In case of distributed matrix-vector multiplication, our proposed Alg. \ref{Alg:New_matvec} has a per worker node computational complexity $\calO \left( (s+1) \eta t \times \frac{r}{k_A} \right) = \calO \left((s + 1) \eta \times \frac{rt}{k_A} \right)$. While it is significantly smaller than that of dense coded approaches, the approach in \cite{dasunifiedtreatment} has a per worker node computational complexity $\calO \left((s + 1) \eta \times \frac{rt}{n} \right)$, slightly smaller than our proposed method (since $n = k_A + s$). However, it involves very high coefficient determination time which is confirmed by numerical experiments in Sec. \ref{sec:numexp}. In addition, unlike the approaches in \cite{das2020coded, dasunifiedtreatment}, our approach is extended to heterogeneous systems.
\end{remark}

\begin{remark}
It should be noted that our proposed approach requires the central node to invert a $k_A k_B\times k_A k_B$ sized decoding matrix. Thus, for $\bfA \in \mathbb{R}^{t \times r}$ and $\bfB \in \mathbb{R}^{t \times w}$, the corresponding decoding complexity is $\mathcal{O} \left( k_A^3 k_B^3 + r w k_A k_B\right)$. On the other hand, the SCS optimal approach in \cite{das2020coded} and the class-based scheme in \cite{dasunifiedtreatment} involves a decoding complexity  $\mathcal{O} \left( \Delta_A^3 k_B^3 + r w \Delta_A k_B\right)$, where $\Delta_A = \textrm{LCM}(n, k_A)$ could be significantly higher than $k_A$.
\end{remark}

\section{Numerical Experiments}
\label{sec:numexp}

In this section, we compare the performance of our proposed approaches with different competing methods \cite{yu2017polynomial,8849468,8919859, das2020coded, dasunifiedtreatment} via numerical experiments. We conduct our experiments on an AWS (Amazon Web Services) cluster with {\tt t2.small} machines as the worker nodes. It should be noted that the work in \cite{wang2018coded} is also suited for sparse matrix computations, however, it does not follow the storage constraints as mentioned in \cite{yu2017polynomial,8849468,8919859, das2020coded, dasunifiedtreatment} and also does not meet the exact optimal recovery threshold. Therefore, we do not include \cite{wang2018coded} in our comparison.

\subsection{Matrix-vector Multiplication}
% \label{sec:numexp_matvec}
We consider matrix-vector multiplication in a system with $n = 30$ workers, each of which can store $\gamma_A = \frac{1}{28}$ fraction of matrix $\bfA$. We consider a sparse input matrix $\bfA$ of size $40,000 \times 31,500$ and a dense vector $\bfx$ of length $40,000$. We assume three different cases where the sparsity of $\bfA$ is $95\%$, $98\%$ and $99\%$, respectively, which indicates that randomly chosen $95\%$, $98\%$ and $99\%$ entries of matrix $\bfA$ are zero. There are many practical examples where the structure of data matrices exhibit this level of sparsity (see \cite{sparsematrices} for such examples). 

{\bf Worker computation time:} 
We compare different methods in terms of worker computation time (the required time for a worker to complete its respective job) in Table \ref{worker_comp}. In this example, the approaches in \cite{yu2017polynomial, 8849468, 8919859} assign linear combinations of $k_A = 28$ submatrices to the worker nodes. Hence, the inherent sparsity of $\bfA$ is destroyed in the encoded submatrices. Thus the worker nodes take a significantly longer amount of time to finish their respective tasks in comparison to our proposed approach or the approaches in \cite{das2020coded} or \cite{dasunifiedtreatment} which are specifically suited for sparse matrices and combine relatively fewer number of submatrices to obtain the encoded submatrices. Note that according to Remark \ref{rem:mv}, the approach in \cite{dasunifiedtreatment} involves slightly lower worker node computation complexity than our proposed approach which is further verified by Table \ref{worker_comp}.

\begin{table*}[t]
\caption{{\small Comparison of worker computation time and communication delay (matrix transmission time) for matrix-vector multiplication for $n = 30, \gamma_A = \frac{1}{28}$ when randomly chosen $95\%$, $98\%$ and $99\%$ entries of $\bfA$ are zero.}}
\label{worker_comp}
\begin{center}
\begin{small}
\begin{sc}
\begin{tabular}{c c c c c c c c c}
\toprule
\multirow{2}{*}{Methods} & \multicolumn{3}{c}{Worker Comp. Time (in ms)} & & \multicolumn{3}{c}{Communication Delay (in s)}  \\ \cline{2-4} \cline{6-8}
&  $\mu = 99\%$ &  $\mu= 98\%$ & $\mu = 95\%$ && $\mu = 99\%$ &  $\mu= 98\%$ & $\mu = 95\%$    \\
 \midrule
Poly. Code  \cite{yu2017polynomial} & $54.7$ & $55.2$  & $53.7$ && $0.48$ &  $0.92$ & $1.31$  \\
Ortho Poly \cite{8849468}    & $54.3$ & $54.8$  & $55.2$ && $0.49$ &  $0.88$ & $1.34$ \\
RKRP Code \cite{8919859}    & $55.1$ & $53.4$  & $53.7$ && $0.52$ &  $0.95$ & $1.27$ \\
SCS Opt. Scheme \cite{das2020coded}  & $15.1$ & $21.3$  & $29.4$ && $0.15$ &  $0.23$ & $0.32$ \\
Class-based Scheme \cite{dasunifiedtreatment} & $14.2$ & $20.8$  & $29.2$ && $0.13$ &  $0.21$ & $0.31$ \\
{\textbf{Proposed Scheme}}  & ${ 14.9}$ & ${ 21.1}$  & ${ 29.6}$ && ${ 0.14}$ &  ${ 0.23}$ & ${ 0.33}$ \\
\bottomrule
\end{tabular}
\end{sc}
\end{small}
\end{center}
\end{table*}%

{\bf Communication delay:} Now we demonstrate the comparison among different approaches in terms of communication delay. Here we define the communication delay as the required time for the central node to transmit the coded submatrices to all the worker nodes. Thus, depending on the coding procedure, the central node may involve different communication delay for different approaches. Since the approaches in \cite{yu2017polynomial}, \cite{8849468}, \cite{8919859} assign dense linear combinations of the submatrices, the number of non-zero entries in the encoded submatrices increase significantly. Thus, in order to transmit these large number of non-zero entries, the system suffers from a considerable communication delay as demonstrated in Table \ref{worker_comp}. On the other hand, our proposed scheme combines limited number of submatrices, which limits the number of non-zero entries in the encoded submatrices; hence the communication delay is reduced significantly.

{\bf Numerical stability:} Next we evaluate the numerical stability of the system for different distributed computation techniques. For any system of $n$ workers and $s$ stragglers, we find the condition numbers of the decoding matrices over all different choices of $s$ stragglers and find the worst case condition number ($\kappa_{worst}$). We consider two different systems with different number of workers and stragglers. In system 1, we set $n = 30$ and $s = 2$, and in system 2, $n = 30$ and $s = 3$. We compare the $\kappa_{worst}$ values for different approaches in Table \ref{table:kappa}. As expected, since the polynomial code approach \cite{yu2017polynomial} involves the ill-conditioned Vandermonde matrices, it has a very high $\kappa_{worst}$ which indicates its numerical instability. Among the numerically stable systems, our proposed approach provides smaller $\kappa_{worst}$ values in comparison to the methods in \cite{das2020coded} and \cite{dasunifiedtreatment}. Note that the works in \cite{8849468, 8919859} provide smaller $\kappa_{worst}$ than our approach; however, as we discussed in Table \ref{worker_comp}, these approaches are not specifically suited for sparse matrices.

\begin{table*}[t]
\caption{\small Comparison among different approaches in terms of worst case condition number $\left(\kappa_{worst} \right)$ among all different choices of $s$ stragglers and the corresponding required time for $20$ trials to find a good set of coefficients}
\label{table:kappa}
\begin{center}
\begin{small}
\begin{sc}
\begin{tabular}{c c c c c}
\toprule
\multirow{2}{*}{Methods}  & \; $\kappa_{worst}$ for  \; & \;Req. time for \;& \;$\kappa_{worst}$ for \; & \; Req. time for\\
  &  \;$n = 30$, $s = 2$\;   & \; $20$ trials\, (in s) \;&  \;$n = 30$, $s = 3$\; &\; $20$ trials \, (in s)\\

 \midrule

 Poly. Code  \cite{yu2017polynomial}   & $2.29 \times 10^{13}$ & $0$ & $1.86 \times 10^{13}$ & $0$\\
%\; \; Convolutional Code \cite{8849395} \; \; & $5.124 \times 10^4$ & \\
 Ortho-Poly\cite{8849468}    & $7.90 \times 10^3$ & $0$ & $3.65 \times 10^5$ & $0$\\
 RKRP Code\cite{8919859}     & $3.64 \times 10^3$ & $1.05$ &$1.34 \times 10^5$  & $9.91$\\
SCS Opt. Scheme \cite{das2020coded}  & $1.68 \times 10^5$ & $172$ &$1.86 \times 10^6$ & $397$\\
 Class-based Scheme \cite{dasunifiedtreatment} & $9.25 \times 10^4$ & $311$ &$3.46 \times 10^6$ & $1020$ \\
{\textbf{Proposed Scheme}}   & ${ 2.43 \times 10^4}$ & ${ 1.09}$ & ${8.33 \times 10^5}$ & ${ 9.84}$\\
\bottomrule
\end{tabular}
\end{sc}
\end{small}
\end{center}
\end{table*}%

{\bf Coefficient determination time:}
Finally, we compare different methods in terms of the required time for $20$ trials to find a ``good'' set of random coefficients that make the system numerically stable. As discussed in Sec. \ref{sec:trialtime}, the approaches in \cite{das2020coded} and \cite{dasunifiedtreatment} partition matrix $\bfA$ into $\Delta_A = \textrm{LCM}(n, k_A)$ block-columns. 
Thus, when $n = 30, s= 2$, $\Delta_A$ is $420$ and when $n = 30, s= 3$, $\Delta_A$ is $270$. In both cases, they are significantly higher than $k_A$ ($28$ and $27$, respectively), i.e. the partition level in our approach. Thus, in order to find the condition number of higher sized matrices, the approaches in \cite{das2020coded} and \cite{dasunifiedtreatment} take much more time than ours. Table \ref{table:kappa} demonstrates more than $100\times$ speed gain from our proposed approach in comparison to the methods in \cite{das2020coded, dasunifiedtreatment} to find a ``good'' set of coefficients. Note that the approaches in \cite{yu2017polynomial, 8849468} do not require random coefficients, thus, do not involve any such delay due to coefficient determination. However, as discussed before, the approaches in \cite{yu2017polynomial, 8849468, 8919859} are not suited to sparse matrices. 

Thus, in summary, while the approaches in \cite{das2020coded, dasunifiedtreatment} involve similar worker computation time as our proposed approach, they require significantly higher encoding time than ours.

\subsection{Matrix-matrix Multiplication}
\label{sec:numexp_matmat}

We consider the case of matrix-matrix multiplication in a system with $n = 39$ workers, each of which can store $\gamma_A = \gamma_B = \frac{1}{6}$ fraction of matrices $\bfA$ and $\bfB$. We consider sparse input matrices $\bfA$ of size $20,000 \times 15000$ and $\bfB$ of size $20,000 \times 12000$. 
%We assume three different cases where the sparsity of $\bfA$ and $\bfB$ are $95\%$, $98\%$ and $99\%$, respectively.
We assume three different cases where the sparsity of $\bfA$ and $\bfB$ are $95\%$, $98\%$ and $99\%$, respectively, which indicate that randomly chosen $95\%$, $98\%$ and $99\%$ entries of matrix $\bfA$ are zero. There are many practical examples where the structure of data matrices exhibit this level of sparsity (see \cite{sparsematrices} for such examples). %Note that we also carry out numerical simulations on matrix-vector multiplication. The results follow a similar trend of the matrix-matrix case and is discussed in supplementary document.

%, which indicates that randomly chosen $95\%$, $98\%$ and $99\%$ entries of matrices $\bfA$ and $\bfB$ are zero. 

{\bf Worker computation time:} 
First we compare different methods in terms of worker computation time (the required time for a worker to complete its respective job) for our system of $n = 39$ workers and the results are shown in Table \ref{worker_comp_matmat}. In this example, the approaches in \cite{yu2017polynomial, 8849468, 8919859} assign linear combinations of $k_A = k_B = 6$ submatrices to the worker nodes. Hence, the inherent sparsity of both $\bfA$ and $\bfB$ can be destroyed in the encoded submatrices. On the other hand, our proposed approach or the approaches in \cite{das2020coded} or \cite{dasunifiedtreatment} assign linear combinations of less number of submatrices, and hence, are specifically suited for sparse matrices. Table \ref{worker_comp_matmat} demonstrates that the worker node computations in these approaches are significantly faster than the dense coded approaches. In addition, if we compare our proposed approach against the approach in \cite{dasunifiedtreatment}, we can see that the ratio in \eqref{eq:ratio} is $\frac{6 \times (6 + 3)}{39} \times \frac{4}{4} \approx 1.38$. It indicates around a $30\%$ reduction of worker computational complexity in our proposed approach than the method in \cite{dasunifiedtreatment} which can be roughly verified from the results in Table \ref{worker_comp_matmat}.

\begin{table*}[t]
\caption{{\small Comparison of worker computation time and communication delay for matrix-matrix multiplication for $n = 39, \gamma_A = \gamma_B = \frac{1}{6}$ when randomly chosen $95\%$, $98\%$ and $99\%$ entries of matrices $\bfA$ and $\bfB$ are zero.}}
\vspace{-0.3 cm}
\label{worker_comp_matmat}
\begin{center}
\begin{small}
\begin{sc}
\begin{tabular}{c c c c c c c c c}
\hline
\toprule
\multirow{2}{*}{Methods} & \multicolumn{3}{c}{Worker Comp. Time (in s)} & & \multicolumn{3}{c}{Communication Delay (in s)}  \\ \cline{2-4} \cline{6-8}
&  $\mu = 99\%$ &  $\mu= 98\%$ & $\mu = 95\%$ && $\mu = 99\%$ &  $\mu= 98\%$ & $\mu = 95\%$    \\
 \midrule
Poly. Code  \cite{yu2017polynomial} & $1.61$ & $5.13$  & $8.91$ && $0.76$ &  $1.41$ & $2.39$  \\
Ortho Poly Code \cite{8849468}    & $1.56$ & $5.18$  & $9.04$ && $0.81$ &  $1.43$ & $2.37$ \\
RKRP Code \cite{8919859}    & $1.58$ & $5.09$  & $8.95$ && $0.78$ &  $1.38$ & $2.35$  \\
SCS Optimal Scheme \cite{das2020coded}  & $0.97$ & $1.38$  & $4.31$ && $0.28$ &  $0.42$ & $0.61$ \\
Class-based Scheme \cite{dasunifiedtreatment} &$0.52$ & $0.85$  & $3.42$ && $0.23$ &  $0.34$ & $0.55$  \\
{\textbf{Proposed Scheme}}  & ${0.34}$ & ${0.53}$  & ${ 2.24}$ && ${0.16}$ & ${0.25}$ & ${0.42}$ \\
\bottomrule
\end{tabular}
\end{sc}
\end{small}
\end{center}
\vspace{-0.8 cm}
\end{table*}%

{\bf Communication delay:}
Now, the comparison among different approaches in terms of communication delay is also demonstrated in Table \ref{worker_comp_matmat}. Here we define the communication delay as the required time for the central node to transmit the coded submatrices to all the worker nodes. Thus, depending on the coding procedure, the central node may involve different communication delay for different approaches. 
%Table \ref{worker_comp_matmat} also demonstrates the delay due to the transmission of the encoded matrices from the central node to the worker node which is defined as communication delay (in Sec. \ref{sec:numexp_matvec}). 
Since the approaches in \cite{yu2017polynomial}, \cite{8849468}, \cite{8919859} assign dense linear combinations of the submatrices, in order to transmit these large number of non-zero entries, the system suffers from a significant communication delay. On the other hand, the algorithm for our proposed scheme and the methods in \cite{das2020coded} and \cite{dasunifiedtreatment} limit the number of non-zero entries; hence the communication delay is reduced significantly. 

In addition, here we compare our proposed approach against the method in \cite{dasunifiedtreatment} in terms of the approximate number of non-zero entries. Consider the case when the matrices $\bfA$ and $\bfB$ have approximately $99\%$ entries to be zero. Now, in our proposed approach, the number of non-zero entries to be sent to each worker node from the central node is approximately 
\begin{align*}
\frac{20k \times 15k}{k_A}  \times 0.01 \times \omega_A + \frac{20k \times 12k}{k_B}  \times 0.01 \times \omega_B = 1.8 \times 10^6.
\end{align*}On the other hand, the number of the corresponding non-zero entries to be sent to each worker node in the approach \cite{dasunifiedtreatment} is approximately
\begin{align*}
\frac{20k \times 15k}{k_A}  \times \frac{1}{13} \times (12 \times 0.01 + k_A \times 0.01)  + \frac{20k \times 12k}{k_B}  \times 0.01 \times \zeta  \approx 2.3 \times 10^6.
\end{align*} Thus our proposed approach requires the central node to transmit approximately $20\%$ less non-zero entries than the method in \cite{dasunifiedtreatment}, which confirms the gain of our approach in Table \ref{worker_comp_matmat}.

{\bf Numerical stability:} Next we evaluate the numerical stability of the system for different distributed computation techniques. For any system of $n$ workers and $s$ stragglers, we find the condition numbers of the decoding matrices over all different choices of $s$ stragglers and find the worst case condition number ($\kappa_{worst}$). We consider two different systems with different number of workers and stragglers. In system 1, we set $n = 33$ and $s = 3$, and in system 2, we set $n = 39$ and $s = 3$ and demonstrate the $\kappa_{worst}$ values of different approaches in Table \ref{table:kappamatmat}. As expected, the approach in \cite{yu2017polynomial} has a very high $\kappa_{worst}$ which indicates its numerical instability. Among the numerically stable systems, our proposed approach provides smaller $\kappa_{worst}$ values in comparison to the methods in \cite{das2020coded} and \cite{dasunifiedtreatment} and also comparable with \cite{8849468} and \cite{8919859}.

\begin{table*}[t]
\caption{\small Comparison among different approaches in terms of worst case condition number $\left(\kappa_{worst} \right)$ among all different choices of $s$ stragglers and the corresponding required time for $10$ trials to find a good set of coefficients}
\vspace{-0.3 cm}
\label{table:kappamatmat}
\begin{center}
\begin{small}
\begin{sc}
\begin{tabular}{c c c c c}
\hline
\toprule
\multirow{2}{*}{Methods}  & \; $\kappa_{worst}$ for  \;& \;Req. time for \;&  \;$\kappa_{worst}$ for  \;& \; Req. time for\\
  & \; $n = 33$, $s = 3$ \;  & \; $10$ trials\, (in s) \;&  \;$n = 39$, $s = 3$\; &\;$10$ trials \, (in s)\\

 \midrule

 Poly. Code  \cite{yu2017polynomial}  & $5.14 \times 10^{14}$ & $0$ & $4.39 \times 10^{19}$ & $0$\\
%\; \; Convolutional Code \cite{8849395} \; \; & $5.124 \times 10^4$ & \\
 Ortho-Poly\cite{8849468}    & $7.23 \times 10^5$ & $0$ & $1.81 \times 10^6$ & $0$\\
 RKRP Code\cite{8919859}   & $2.38 \times 10^5$ & $5.45$ &$3.43 \times 10^5$  & $10.31$\\
 SCS Opt. Scheme \cite{das2020coded}   & $5.39 \times 10^7$ & $738$ &$9.15\times 10^7$ & $3191$\\
 Class-based Scheme \cite{dasunifiedtreatment} & $4.95 \times 10^7$ & $1327$ &$6.34 \times 10^7$ & $5772$ \\
 {\textbf{Proposed Scheme}}  & ${ 4.40 \times 10^5}$ & $5.87$ & ${2.21 \times 10^6}$ & ${ 11.37}$\\
\bottomrule
\end{tabular}
\end{sc}
\end{small}
\end{center}
\vspace{-0.8 cm}
\end{table*}%

{\bf Coefficient determination time:}
Finally, we compare different approaches in terms of the required time for running $10$ trials to find a ``good'' set of random coefficients that make the system numerically stable. Our proposed approach and the approach in \cite{8919859} partition matrices $\bfA$ and $\bfB$ into $k_A$ and $k_B$ block-columns, which leads to $k_A k_B$ unknowns. 
On the other hand, the approaches in \cite{das2020coded} and \cite{dasunifiedtreatment} partition matrices $\bfA$ and $\bfB$ into $\Delta_A = \textrm{LCM}(n, k_A)$ and $k_B$ block-columns, which leads to $\Delta_A k_B$ unknowns. 
%While the approaches in \cite{das2020coded} and \cite{dasunifiedtreatment} partition matrix $\bfB$ into $k_B$ block-columns, they partition matrix $\bfA$ into $\Delta_A = \textrm{LCM}(n, k_A)$ block-columns, which leads to $\Delta_A k_B \geq k_A k_B$ unknowns. 
Now as discussed in Sec. \ref{sec:trialtime}, $\Delta_A$ can be significantly higher than $k_A$ and to find the condition numbers of these larger-sized matrices, approaches in \cite{das2020coded} and \cite{dasunifiedtreatment} take much more time than ours. Table \ref{table:kappamatmat} confirms more than $100\times$ speed gain for our proposed scheme over the methods in \cite{das2020coded, dasunifiedtreatment} to find a ``good'' set of coefficients. 
%Note that the polynomial coded based-approaches \cite{yu2017polynomial, 8849468} require very little time for determine the coefficients, since their coefficients have closed form equations, however, they are not suited for sparse input matrices. 

Thus, in summary, while the approaches in \cite{das2020coded, dasunifiedtreatment} involve similar worker computation time as our proposed approach, they require significantly higher encoding time than ours.

{\bf Heterogeneous system:} Next we consider a matrix-matrix multiplication over a heterogeneous system of $\bar{n} = 19$ worker nodes of $\lambda = 3$ different types of nodes. We assume that there are $\bar{n}_0 = 11$, $\bar{n}_1 = 5$ and $\bar{n}_2 = 3$ nodes of types $0, 1$ and $2$, respectively, which are assigned $1$, $2$ and $3$ block-columns each, respectively, hence $n = 11 \times 1 + 5 \times 2 + 3 \times 3 = 30$. We design the scheme according to Alg. \ref{Alg:New_matmat}, such that it is resilient to {\it any} $s = 6$ block-columns processing. 

Now, in this heterogeneous setting, since different nodes are assigned different amounts of jobs depending on their corresponding types, the central node will wait until it receives the results of fastest $Q =  24$ block-column processing. In this regard, we define the worst case condition number in the heterogeneous system ($\bar{\kappa}_{worst}$) as the maximum of the condition numbers of the decoding matrices over all different choices of $s$ block-columns. 
%In our example, it corresponds to choose {\it any six}  block-columns assignments out of $n = 30$ of them. 
Table \ref{table:hetero} shows the comparison among different approaches in terms of ($\bar{\kappa}_{worst}$), where we notice that our proposed approach provides significantly smaller ($\bar{\kappa}_{worst}$) values than the approaches in \cite{yu2017polynomial, 8849468, das2020coded, dasunifiedtreatment} and also provides competitive $\bar{\kappa}_{worst}$ value in compared to the approach in \cite{8919859}.

\begin{table}[t]
\caption{\small Comparison among different approaches in terms of $\kappa_{worst}$ for a heterogeneous system with $\bar{n} = 19$ where $n = 30$ and $s = 6$.}
\vspace{-0.2 in}
\label{table:hetero}
\begin{center}
\begin{small}
\begin{sc}
\begin{tabular}{c c c c c c c}
\hline
\toprule
\multirow{2}{*}{Methods}  & Polynomial  & Ortho- & RKRP  & SCS Optimal & Class-based & {\textbf{Proposed}} \\
&  Code  \cite{yu2017polynomial}  & Poly\cite{8849468} & Code\cite{8919859} & Scheme \cite{das2020coded} & Scheme \cite{dasunifiedtreatment} & {\textbf{Scheme}} \\
 \midrule
$\bar{\kappa}_{worst}$ & $1.49\times 10^{13}$ & $1.74\times 10^{9}$  & $7.11\times 10^{6}$ & $3.16\times 10^{8}$ &  $5.29\times 10^{8}$ & $7.78\times 10^{7}$  \\
\bottomrule
\end{tabular}
\end{sc}
\end{small}
\end{center}
\vspace{-0.8 cm}
\end{table}%

\section{Conclusion}
\label{sec:conclusion}

In this work, we have developed distributed matrix computation schemes which preserve sparsity properties of the inputs while remaining resilient to the maximum number of stragglers for given storage constraints. We saw how existing dense coded approaches \cite{yu2017polynomial, 8849468, 8919859, das2019random} suffer from a huge communication and computation delay in case of sparse matrices. Since our proposed approach allows very limited amounts of coding within the submatrices, it preserves the inherent sparse structure of the input matrix $\bfA$ (and $\bfB$) up to certain level. Thus, the worker computation delay and the communication delay were seen to be significantly reduced in comparison to those dense coded approaches. There are some sparsely coded approaches in \cite{das2020coded, dasunifiedtreatment} which have been developed specifically to deal with sparse matrices; however, our proposed approach was seen to provide three-fold gains over them. Overall, we showed analytically and experimentally that our proposed approach (i) provides significant gain in worker computation and communication delay, (ii) saves considerable amount of time to find a ``good'' set of coefficients to make the system numerically stable, and (iii) is applicable to the systems where the worker nodes are heterogeneous in nature. %Numerical experiments conducted on an AWS cluster support our claims.  

There are a number of directions for the future work of this paper. While there are several secure distributed matrix computation schemes \cite{tandon2018secure, aliasgari2020private, hollanti2022secure, xhemrishi2022distributed} which protect the system against privacy leakage, most of them add dense random matrices to the coded submatrices which destroy the sparsity of the assigned submatrices. Thus, a straggler resilient secure coded scheme needs to be developed which is particularly suitable for sparse input matrices. Another direction can be developing a scheme for a server-less architecture, where there is no such central node to encode the matrices and the worker nodes may communicate among them to establish straggler resilience. This can be particularly helpful for distributed learning or federated learning methods \cite{9252954}. We could also aim for improving the performance in the heterogeneous setting by assigning multiple jobs with varying weights as guided in \cite{ozfatura2021coded,ozfatura2020age}. Moreover, we need to develop schemes where true knowledge about the worker nodes in the heterogeneous setting may not be available prior to the assignment of the jobs. 

\appendix

%\subsection{Claims on the Structure of Our Proposed Scheme}

% \begin{claim}
% \label{cl:omegaB}
% According to our rearrangement, either of the following two cases has to be true.
% \begin{align*}
% \text{{\it Case 1:}}& \; \sum\limits_{i = k_A - \omega_A + 1}^{k_A - 1} \tilde{\delta}_i \leq (k_B - \omega_B)(\omega_A - 1), \, \text{or}\\
% \text{{\it Case 2:}}& \;\; \text{If case 1 is not true, then}\; \, \tilde{\delta}_{k_A - \omega_A} \geq k_B - \omega_B + 1.
% \end{align*}
% \end{claim}
% \begin{proof}
% Since $\sum_{i = 0}^{k_A - 1} \tilde{\delta}_i = m$ and $1 \leq m \leq k_A k_B$, we can make $m$ as small as possible so that Case 1 is true. Now assume that Case 1 is not true, and to prove by contradiction, assume that $\tilde{\delta}_{k_A - \omega_A} \leq k_B - \omega_B$. But in that case,
% \begin{align*}
%     \sum\limits_{i = k_A - \omega_A+1}^{k_A - 1} \tilde{\delta}_i \leq (\omega_A - 1) (k_B - \omega_B);
% \end{align*} hence Case 1 is true, which is contradictory.
% \end{proof}

\subsection{Notation Table}
\label{app:not}
In this section, for the ease of the readers, we provide an overview of the notations used in this work in Table \ref{tab:not}. This table also includes very brief definitions of the corresponding notations.

\begin{table*}[t]
\caption{{\small Notation Table}} 
\vspace{-0.6 cm}
\label{tab:not}
\begin{center}
\begin{small}
\begin{sc}
\begin{tabular}{c c c}
\hline
\toprule
notation & definition & description\\
 \midrule
$\bfA, \bfB$ & $\;\;\;$ sparse large-sized matrices & $\bfA \in \mathbb{R}^{t \times r}, \bfB \in \mathbb{R}^{t \times w}$ \\  \hline
$\gamma_A, \gamma_B$ & $\;\;\;$ storage fraction for $\bfA$ and $\bfB$, respectively & $\gamma_A = \frac{1}{k_A}, \gamma_B = \frac{1}{k_B}$ \\  \hline
 $n$ & number of total worker nodes & $n \geq k_A k_B$\\ \hline
 $s$ & number of maximum possible stragglers & $s = n - k_A k_B$ \\ \hline
 $W_i$ & Worker node with index $i$ & $0 \leq i \leq n - 1$ \\ \hline
\multirow{2}{1 cm}{$\Delta_A, \Delta_B$} & number of block-columns that & $\Delta_A = k_A$ \\ 
& $\bfA$ and $\bfB$, respectively, are partitioned into &  and $\Delta_B = k_B$\\  \hline
$\Delta$ & total number of unknowns that need to be recovered & $\Delta = \Delta_A \Delta_B$ \\ \hline
%$s$ & number of stragglers & $s = n - \tau$ \\ \hline
%$x$ & relaxation in number of stragglers & $x = s_m - s$ \\ \hline
%$y$ & reduction of weights in coded submatrices of $\bfA$ & $y = \floor{\frac{k_A x}{s_m}}$ \\ \hline
$\kappa_{worst}$ & Worst case condition number over all $n \choose s$ stragglers & $-$ \\ \hline
$\tau$ & recovery threshold of the scheme & $\tau = k_A k_B$ \\ \hline
\multirow{2}{1 cm}{$\;\;\; Q$} & number of submatrix products that have to be computed & \multirow{2}{2 cm}{$\;\;\;  \;\;\; Q \geq \Delta$} \\ 
 & in the worst case to recover the intended result &  \\ \hline
$\omega_A, \omega_B$ & weights for the encoding of $\bfA$ and $\bfB$ & $\omega_A \omega_B > s$ \\ 
%\multirow{2}{2 cm}{$\;\;\;  \;\;\; \ell_u, \ell_c$} & number of uncoded and coded submatrices of $\bfA$ & $\ell_u = \frac{\Delta}{n}$ and \\
%& assigned to every worker node & $ \; \; \ell_ u + \ell_c = \Delta_A /k_A$\\
\bottomrule
\end{tabular}
\end{sc}
\end{small}
\end{center}
\end{table*}%

\vspace{-0.3 cm}
\subsection{Proof of Lemma \ref{lem:no_of_unknowns}}
\label{App:proof_no_of_unknowns}
\begin{proof}
As discussed in Sec. \ref{sec:homogeneous_matmat}, we denote the minimum number of participating unknowns in $\tilde{\calM}_i$ by $\rho_i$ when $0 \leq i \leq k_A - 1$. The trivial lower bound for $\rho_i$ is {\it zero} when $k_A - \omega_A + 1 \leq i \leq k_A - 1$, hence, $  \sum\limits_{i=0}^{k_A - 1} \rho_i \geq \sum\limits_{i=0}^{k_A - \omega_A} \rho_i$. Thus, in order to prove the lemma, we need to show that
\vspace{-0.4 cm}
\begin{align}
\label{eq:provehallfull}
    \sum\limits_{i=0}^{k_A - \omega_A} \rho_i \; \geq \; \sum\limits_{i = 0}^{k_A - 1} \tilde{\delta}_i.
\end{align} Now we provide the following definition for the next part of the proof.

\begin{definition}
We say that $\tilde{\calM}_i$ covers itself, if $\rho_i \geq \tilde{\delta_i}$. Next, we say that a set of $\tilde{\calM}_i$'s, denoted by $\calU$, covers itself and another set of $\tilde{\calM}_j$'s, denoted by $\calV$, (where $\calU$ and $\calV$ are disjoint) if
\vspace{-0.2 cm}
\begin{align*}
    \sum\limits_{i: \tilde{\calM}_i \in \calU} \rho_i =  \sum\limits_{i: \tilde{\calM}_i \in \calU} \tilde{\delta}_i + \sum\limits_{j: \tilde{\calM}_j \in \calV} \tilde{\delta}_j .
\end{align*}
\end{definition}

\begin{claim}
\label{cl:cover}
Every $\tilde{\calM}_j$ covers itself, for $0 \leq j \leq k_A - \omega_A$.
\end{claim}
\begin{proof}
From \eqref{eq:rhozero}, we have $\rho_0 \geq \omega_A \times (\omega_B + \tilde{\delta}_0 - 2)  \geq \tilde{\delta_0}$, since $\omega_A, \omega_B \geq 2$. So, $\tilde{\calM}_0$ covers itself. Next from \eqref{eq:rhobigger}, for $1 \leq j \leq k_A - \omega_A$, we have $\rho_j \geq  \tilde{\delta}_j + \omega_B - 2 \geq \tilde{\delta}_j$. So, $\tilde{\calM}_j$ covers itself.
\end{proof} 
Now since \eqref{eq:provehallfull} leads to $\sum\limits_{i=0}^{k_A - \omega_A} \rho_i \; \geq \; \sum\limits_{i = 0}^{k_A - \omega_A} \tilde{\delta}_i + \sum\limits_{i = k_A - \omega_A + 1}^{k_A - 1} \tilde{\delta}_i$, 
 we first denote a set $\calV = \{ \tilde{\calM}_{k_A - \omega_A + 1}, \tilde{\calM}_{k_A - \omega_A + 2}, \dots, \tilde{\calM}_{k_A - 1}\}$ (thus, $\abs*{\calV} = \omega_A - 1$) and in order to satisfy \eqref{eq:provehallfull}, we always need to find an appropriate $\calU$ which can also cover  $\calV$ along with itself. To do so, we need to find those $\tilde{\calM}_i$'s where $\rho_i - \tilde{\delta}_i > 0$. In this proof, we define $\calU_{\lambda} = \{ \tilde{\calM}_0, \tilde{\calM}_1, \tilde{\calM}_2, \dots, \tilde{\calM}_{\lambda}\}$. 
%and we state the following claim on the difference of $\rho_i$ and $\tilde{\delta}_i$ which could help us to find an appropriate $\calU$. 
 Now we consider the following two cases for each of which we show that \eqref{eq:provehallfull} is true.

{\it Case 1}: $\mathbf{1 \leq \tilde{\delta}_0 \leq k_B}$. If, $\tilde{\delta}_0 = 1$, using \eqref{eq:rhozero}, we have $\rho_0 = \omega_A \omega_B > \omega_A \tilde{\delta}_0$, since $\omega_B \geq 2$. Moreover, if $1 < \tilde{\delta}_0 \leq k_B$, using \eqref{eq:rhozero}, we have  $\rho_0 \geq \omega_A \tilde{\delta}_0$. Thus, when $1 \leq \tilde{\delta}_0 \leq k_B$, we have $\rho_0 \geq \omega_A \tilde{\delta}_0 \; = \; \tilde{\delta}_0 + (\omega_A - 1) \tilde{\delta}_0 \; \geq \; \tilde{\delta}_0 + \sum\limits_{i= k_A - \omega_A + 1}^{k_A - 1} \tilde{\delta}_i$,  since $\tilde{\delta}_i$'s are arranged in a non-increasing order. In this case, we can set $\calU = \calU_0 = \{ \tilde{\calM}_0 \}$ which covers $\calV$ along with itself. Now, since each of the other $\tilde{\calM}_i$'s (for $i \leq 1 \leq k_A - \omega_A$) covers itself (according to Claim \ref{cl:cover}), \eqref{eq:provehallfull} is true.

{\it Case 2}:  $\mathbf{\tilde{\delta}_0 = k_B + 1}$. Assume that $\tilde{\delta}_0 = \tilde{\delta}_1 = \dots = \tilde{\delta}_{\alpha - 1} = k_B + 1$. Now, there can be at most $s \leq \min (\omega_A \omega_B - 1, k_A)$ such $\tilde{\calM}_i$'s which have cardinality $k_B + 1$.
%If $\omega_A \omega_B - 1 \leq k_A$, there can be at most $s \leq \omega_A \omega_B - 1$ such $\tilde{\calM}_i$'s which have cardinality $k_B + 1$. Otherwise, since $n \leq (k_A+1)k_B$, there can be at most $k_A$ such $\tilde{\calM}_i$'s which have cardinality $k_B + 1$ 
Thus, $\alpha$ is upper bounded by $\min(\omega_A \omega_B - 1, k_A)$. Now, before moving into details in this case, we state the following claims.

\begin{claim}
\label{cl:excess}
For $\alpha \geq 1$, if $\sum\limits_{i = k_A - \omega_A + 1}^{k_A - 1} \tilde{\delta}_i \leq (\omega_A - 1) k_B - \alpha$, then \eqref{eq:provehallfull} is true by setting $\calU = \calU_{\alpha-1}$. 
\end{claim}
\begin{proof}
First we can say, 
\begin{align*}
\sum_{i = 0}^{\alpha - 1} \rho_i - \sum_{i = 0}^{\alpha - 1} \tilde{\delta}_i = \omega_A k_B + (\alpha - 1) k_B - \alpha (k_B + 1) = (\omega_A - 1 ) k_B - \alpha.
 \end{align*} Here, $\omega_B < k_B$, we assume $k_B = \omega_B + \nu$ where $\nu \geq 1$. Now $\alpha$ is upper bounded by $\omega_A \omega_B - 1$, so $(\omega_A - 1 ) k_B - \alpha  \geq (\omega_A - 1 ) k_B - (\omega_A \omega_B - 1) = \omega_A \nu - (\omega_B + \nu - 1)= (\omega_A - \omega_B) + (\nu - 1)(\omega_A - 1)$ 
is non-negative since according to Alg. \ref{Alg:New_matmat}, we have $\omega_A \geq \omega_B \geq 2$. Thus, if $\sum\limits_{i = k_A - \omega_A + 1}^{k_A - 1}\tilde{\delta}_i \leq (\omega_A - 1 ) k_B - \alpha$, then we set $\calU = \calU_{\alpha-1}$ which covers $\calV$ along with itself, and we are done.
\end{proof}

\begin{remark}
While we cannot assume $\omega_A \geq \omega_B$ without loss of generality, note that $\omega_A$ and $\omega_B$ are design parameters. The constraint that we have on the number of stragglers in terms of $\omega_A$ and $\omega_B$ is that $s \leq \omega_A \omega_b - 1$. Since the upper bound is symmetric in terms of $\omega_A$ and $\omega_B$, we can always set $\omega_A \geq \omega_B$ in our design to be resilient to the same number of stragglers.
\end{remark}

% \begin{remark}
% In this algorithm, since we have already assumed $k_A \geq k_B$, we cannot assume $\omega_A \geq \omega_B$ without loss of generality. While that is true, $\omega_A$ and $\omega_B$ are design parameters, and the only constraint that we have on the number of stragglers is that $s \leq \omega_A \omega_b - 1$, where the upper bound is symmetric in terms of $\omega_A$ and $\omega_B$. Thus, we can always set $\omega_A \geq \omega_B$ in our design to be resilient to the stragglers.
% \end{remark}

\begin{claim}
\label{cl:difference}
Assume that $\alpha \geq 1$ and let us define $\kappa$ as the minimum $i$ such that $\tilde{\delta}_i \leq k_B - 1$. (a) If $\kappa > k_A - \omega_A$, then $\calU = \calU_{\alpha- 1}$ will cover $\calV$ along with itself. (b) If $\kappa \leq k_A - \omega_A$ and $\tilde{\delta}_i \geq k_B - \omega_B + 2 \;$ for all $i \leq k_A - \omega_A$,  then $\calU = \calU_{\alpha- 1} \cup \tilde{\calU}_{\kappa}$ will cover $\calV$ along with itself, where $ \, \tilde{\calU}_{\kappa} \coloneqq \left\lbrace \tilde{\calM}_{\kappa},  \tilde{\calM}_{\kappa+1}, \dots,  \tilde{\calM}_{k_A - \omega_A} \right\rbrace$. 
\end{claim}

\begin{proof}
(a) If $\kappa > k_A - \omega_A$, then we have $\tilde{\delta}_i = k_B + 1$, when $0 \leq i \leq \alpha - 1$ and $\tilde{\delta}_i = k_B$ when $\alpha \leq i \leq k_A - \omega_A$. Thus
\begin{align*}
        \sum\limits_{i = 0}^{k_A - 1} \tilde{\delta}_i &= \sum\limits_{i = 0}^{\alpha - 1} \tilde{\delta}_i + \sum\limits_{i = \alpha}^{k_A - \omega_A} \tilde{\delta}_i + \sum\limits_{i = k_A - \omega_A + 1}^{k_A - 1} \tilde{\delta}_i  = \alpha (k_B + 1) + (k_A - \omega_A - \alpha + 1) k_B + \sum\limits_{i = k_A - \omega_A + 1}^{k_A - 1} \tilde{\delta}_i  \\
        & =  \sum\limits_{i = k_A - \omega_A + 1}^{k_A - 1} \tilde{\delta}_i + k_A k_B + \alpha  - (\omega_A - 1) k_B \leq k_A k_B , 
\end{align*}
where the upper bound holds since $m = \sum\limits_{i = 0}^{k_A - 1} \tilde{\delta}_i \leq k_A k_B$. Hence, $\sum\limits_{i = k_A - \omega_A + 1}^{k_A - 1} \tilde{\delta}_i \leq (\omega_A - 1) k_B - \alpha$, and then according to Claim \ref{cl:excess}, we are done by setting $\calU = \calU_{\alpha- 1}$.

(b) If $\omega_B = 2$, then $ \tilde{\delta}_i \geq k_B - \omega_B + 2 = k_B$ for all $i \leq k_A - \omega_A$. Thus, $\kappa > k_A - \omega_A$ (i.e., $\tilde{\calU}_{\kappa} = \varnothing$), hence according to part (a) of this claim, $\calU = \calU_{\alpha- 1}$ will cover $\calV$ along with itself. 

Now we consider the scenario when $\omega_B \geq 3$. In this scenario, if $\kappa > k_A - \omega_A$, then according to part (a) of this claim, $\calU = \calU_{\alpha- 1}$ will cover $\calV$ along with itself. Now, we assume that $\kappa \leq k_A - \omega_A$. We assume that $\tilde{\delta}_i = k_B - \omega_B + 2 + \sigma_i$ where $0 \leq \sigma_i \leq \omega_B - 3$ for all $\tilde{\calM}_i \in \tilde{\calU}_{\kappa}$, and thus, according to \eqref{eq:rhobigger}, we have $\rho_i \geq  \textrm{min} \left( \omega_B + \tilde{\delta}_i - 2, k_B\right)  = k_B$. Assume that $\sum\limits_{i = k_A - \omega_A + 1}^{k_A - 1} \tilde{\delta}_i = k_B(\omega_A - 1) - \lambda_1$. Now if $\lambda_1 \geq \alpha$, then $\sum\limits_{i = k_A - \omega_A + 1}^{k_A - 1} \tilde{\delta}_i \leq k_B(\omega_A - 1) - \alpha$ and we are done by setting $\calU = \calU_{\alpha-1} \subset \left(\calU_{\alpha- 1} \cup \tilde{\calU}_{\kappa}\right)$ according to Claim \ref{cl:excess}. Now, if $\lambda_1 < \alpha$, we have
\begin{align*}
   & \sum\limits_{i = 0}^{k_A - 1} \tilde{\delta}_i = \sum\limits_{i = 0}^{\alpha - 1} \tilde{\delta}_i + \sum\limits_{i = \alpha}^{\kappa - 1} \tilde{\delta}_i + \sum\limits_{i = \kappa}^{k_A - \omega_A} \tilde{\delta}_i + \sum\limits_{i = k_A - \omega_A + 1}^{k_A - 1} \tilde{\delta}_i = \alpha (k_B + 1) + (\kappa - \alpha) k_B  \\ + & \sum\limits_{i = \kappa}^{k_A - \omega_A}  \left( k_B - \omega_B + 2 + \sigma_i \right)+ k_B(\omega_A - 1) - \lambda_1 = k_A k_B + \alpha - \left\lbrace\sum_{i = \kappa}^{k_A - \omega_A} \left( \omega_B - 2 - \sigma_i \right) + \lambda_1 \right\rbrace ,
\end{align*} which is upper bounded by $k_A k_B$. It indicates that $\alpha \leq \left\lbrace\sum_{i = \kappa}^{k_A - \omega_A} \left( \omega_B - 2 - \sigma_i \right) + \lambda_1 \right\rbrace$. Now, 
\begin{align*}
    \sum\limits_{i = 0}^{\alpha - 1} \rho_i +  \sum\limits_{i = \kappa}^{k_A - \omega_A} \rho_i & = \omega_A   k_B + (\alpha-1) k_B + (k_A - \omega_A - \kappa + 1) k_B = (\alpha + k_A - \kappa) k_B\; \; \; \text{and} \\
     \sum\limits_{i = 0}^{\alpha - 1} \tilde{\delta}_i +  \sum\limits_{i = \kappa}^{k_A - 1 } \tilde{\delta_i} & \leq \alpha (k_B + 1) + \sum\limits_{i = \kappa}^{k_A - \omega_A} \left( k_B - \omega_B + 2 + \sigma_i\right)+ k_B(\omega_A - 1) - \lambda_1 \\
     % & = (\omega_A  - 1) k_B + \alpha k_B + (k_A - \omega_A - \kappa + 1) k_B + \alpha - \left\lbrace \lambda_1 + \sum_{i = \kappa}^{k_A - \omega_A} \left(\omega_B - 2 - \sigma_i\right)\right\rbrace \, .
     & = (\alpha + k_A - \kappa) k_B + \alpha - \left\lbrace \lambda_1 + \sum_{i = \kappa}^{k_A - \omega_A} \left(\omega_B - 2 - \sigma_i\right)\right\rbrace \leq (\alpha + k_A - \kappa) k_B.
\end{align*} Thus we are done with part (b) by setting $\calU = \calU_{\alpha- 1} \cup \tilde{\calU}_{\kappa}$.
\end{proof}

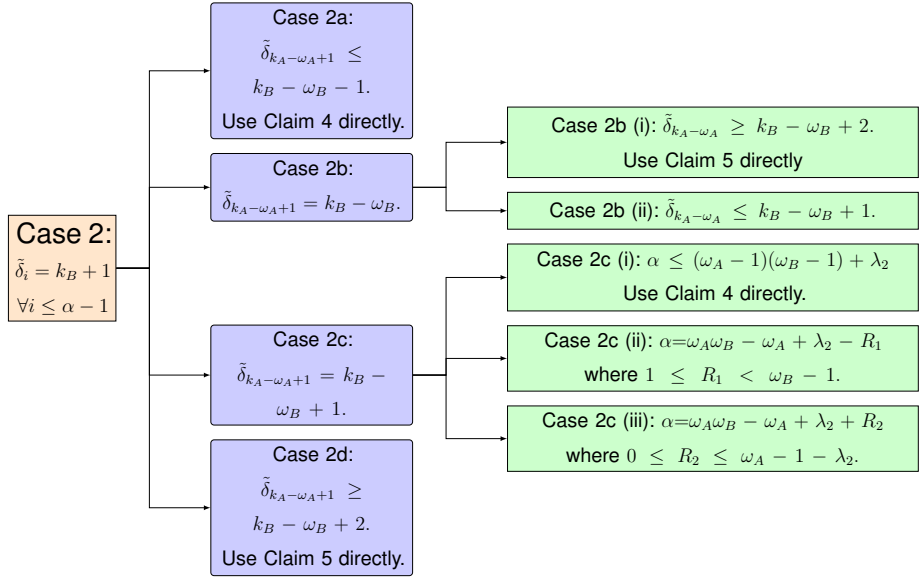
\begin{figure*}
    \centering
\resizebox{0.75\linewidth}{!}{    
\tikzset{
    basic/.style  = {draw, text width=2cm, fill=orange!20, align=center, font=\sffamily, rectangle},
    root/.style   = {basic, rounded corners=2pt, thin, align=center, fill=green!30},
    onode/.style = {basic, thin, rounded corners=2pt, align=center, fill=green!60,text width=3cm,},
    tnode/.style = {basic, thin, align=left, fill=green!20, text width=20em, align=center},
    xnode/.style = {basic, thin, rounded corners=2pt, align=center, fill=blue!20,text width=4cm,},
    wnode/.style = {basic, thin, align=left, fill=pink!10!blue!80!red!10, text width=6.5em},
    edge from parent/.style={draw=black, edge from parent fork right}

}

\begin{forest} for tree={
    %node options={align=center,},
    calign=center,
    s sep= 0.3 cm,
    grow=east,
    growth parent anchor=west,
    parent anchor=east,
    child anchor=west,
    edge path={\noexpand\path[\forestoption{edge},->, >={latex}] 
         (!u.parent anchor) -- +(20pt,0pt) |-  (.child anchor) 
         \forestoption{edge label};}
}
% lsep is used for arrow distance
[{{\Large Case 2:} \\ $\tilde{\delta}_i = k_B + 1$\\ $\forall  i \leq \alpha-1$} \\ , basic,  l sep=20mm,
    [Case 2d: \\ $\tilde{\delta}_{k_A - \omega_A + 1} \geq k_B - \omega_B + 2.$\\ Use Claim \ref{cl:difference} directly., xnode,  l sep=20mm]
    [Case 2c: \\ $\tilde{\delta}_{k_A - \omega_A + 1} \; {=} \; k_B - \omega_B + 1.$, xnode,  l sep=20mm,
        [Case 2c (iii): $\alpha {=} \omega_A \omega_B - \omega_A + \lambda_2 + R_2$ \\ where $0 \leq R_2 \leq \omega_A - 1 - \lambda_2.$, tnode]
        [Case 2c (ii): $\alpha {=} \omega_A \omega_B - \omega_A + \lambda_2 - R_1$ \\ where $1 \leq R_1 < \omega_B - 1.$, tnode]
        [Case 2c (i): $\alpha \leq (\omega_A - 1) (\omega_B - 1) + \lambda_2$ \\ Use Claim \ref{cl:excess} directly., tnode] ]
    [Case 2b: \\ $\tilde{\delta}_{k_A - \omega_A + 1} \; {=} \; k_B - \omega_B.$, xnode,  l sep=20mm,
       [Case 2b (ii): $\tilde{\delta}_{k_A - \omega_A} \leq k_B - \omega_B + 1.$, tnode] 
        [Case 2b (i): $\tilde{\delta}_{k_A - \omega_A} \geq k_B - \omega_B + 2.$ \\ Use Claim \ref{cl:difference} directly, tnode]] 
    [Case 2a: \\ $\tilde{\delta}_{k_A - \omega_A + 1} \leq k_B - \omega_B - 1.$\\ Use Claim \ref{cl:excess} directly., xnode,  l sep=20mm ]
]
\end{forest}
}
    \caption{An overview of Case 2}
    \label{fig:case2}
    \vspace{-0.8 cm}
\end{figure*}

%Based on the above discussion in Claims \ref{cl:excess} and \ref{cl:difference}(a), we only need to consider the remaining situation where $\sum\limits_{i = k_A - \omega_A + 1}^{k_A - 1}\tilde{\delta}_i > (\omega_A - 1 ) k_B - \alpha$ and $\kappa \leq k_A - \omega_A$. 
Our main idea is to perform an exhaustive case analysis on the value of $\tilde{\delta}_{k_A - \omega_A + 1}$, an overview of which is depicted in Fig. \ref{fig:case2}. 
In all cases, we find the appropriate $\calU$ such that \eqref{eq:provehallfull} holds.

{\it Case 2a}: $\tilde{\delta}_{k_A - \omega_A + 1} \leq k_B - \omega_B - 1$. Now, $\tilde{\delta}_i$'s are in non-increasing order, thus, in this case
\begin{align*}
\sum\limits_{i = k_A - \omega_A + 1}^{k_A - 1} \tilde{\delta}_i \leq (k_B - \omega_B - 1) (\omega_A - 1) & = k_B(\omega_A - 1) - ( \omega_B + 1)(\omega_A -1) \leq k_B(\omega_A - 1) - \alpha ,
\end{align*} 
% \begin{align*}
% \sum\limits_{i = k_A - \omega_A + 1}^{k_A - 1} \tilde{\delta}_i \leq (k_B - \omega_B - 1) (\omega_A - 1) & = k_B(\omega_A - 1) - (\omega_A \omega_B + \omega_A - \omega_B -1) \\ & \leq k_B(\omega_A - 1) - \alpha ,
% \end{align*} 
since $\omega_A \geq \omega_B$ and $\alpha \leq \omega_A \omega_B - 1$. Thus, according to Claim \ref{cl:excess}, we are done by setting $\calU = \calU_{\alpha - 1}$.

{\it Case 2b}: $\tilde{\delta}_{k_A - \omega_A + 1} = k_B - \omega_B $. To prove \eqref{eq:provehallfull}, here we consider the following two subcases.

{\it Case 2b (i)}: $\tilde{\delta}_{k_A - \omega_A} \geq k_B - \omega_B + 2$. In this scenario, if $\kappa > k_A - \omega_A$, then we are done using  Claim \ref{cl:difference}(a), and, if $\kappa \leq k_A - \omega_A$, then we are done using  Claim \ref{cl:difference}(b) since $\tilde{\delta}_i \geq k_B - \omega_B + 2$ for all $i \leq k_A - \omega_A$. We can find $\kappa$ in Claim \ref{cl:difference} as the minimum value of $i$ when $\tilde{\delta}_i < k_B$.

{\it Case 2b (ii)}: $\tilde{\delta}_{k_A - \omega_A} \leq k_B - \omega_B + 1$. Since $\tilde{\delta}_{k_A - \omega_A + 1} = k_B - \omega_B $, then $\sum\limits_{i = k_A - \omega_A + 1}^{k_A - 1} \tilde{\delta}_i \leq (k_B - \omega_B) (\omega_A - 1)$. Now in this scenario, $\sum\limits_{i = 0}^{\alpha - 1} \rho_i +  \rho_{k_A - \omega_A} = \omega_A k_B + (\alpha-1) k_B + \rho_{k_A - \omega_A}$, and, 
\begin{align*}
\sum\limits_{i = 0}^{\alpha - 1} \tilde{\delta}_i +  \sum\limits_{i = k_A - \omega_A}^{k_A - 1 } \tilde{\delta_i} & \leq \alpha (k_B + 1) + \tilde{\delta}_{k_A - \omega_A} +  (k_B - \omega_B) (\omega_A - 1)  \\ & =  \omega_A k_B + (\alpha-1) k_B + \tilde{\delta}_{k_A - \omega_A} +  \alpha - \omega_A \omega_B + \omega_B.
\end{align*} If $\abs*{\tilde{\calM}_{k_A - \omega_A}} = k_B$, then according to \eqref{eq:rhobigger}, $\rho_{k_A - \omega_A} - \tilde{\delta}_{k_A - \omega_A} \geq \omega_B - 1$. Since $\alpha \leq \omega_A \omega_B - 1$, we are done by setting $\calU = \calU_{\alpha- 1} \cup \{ \tilde{\calM}_{k_A - \omega_A} \}$ to cover $\calV$ along with itself. Otherwise, if $\abs*{\tilde{\calM}_{k_A - \omega_A}} = k_B + 1$, we have $\rho_{k_A - \omega_A} - \tilde{\delta}_{k_A - \omega_A} \geq \omega_B - 2$. However, in that case $\alpha \leq \omega_A \omega_B - 2$ since $\abs*{\tilde{\calM}_{k_A - \omega_A}}$ is already set as $k_B + 1$. Thus, we can again set $\calU = \calU_{\alpha- 1} \cup \{ \tilde{\calM}_{k_A - \omega_A} \}$.  

{\it Case 2c}: $\tilde{\delta}_{k_A - \omega_A + 1} = k_B - \omega_B + 1$. In this case, if $\tilde{\delta}_{k_A - \omega_A} \geq k_B - \omega_B + 2$, according to Claim \ref{cl:difference}(b), we are done by setting $\calU = \calU_{\alpha- 1} \cup \tilde{\calU}_{\kappa}$. We can find $\kappa$ in as the minimum value of $i$ when $\tilde{\delta}_i < k_B$. Now, we consider the only remaining scenario where $\tilde{\delta}_{k_A - \omega_A} = k_B - \omega_B + 1$. Since, $\tilde{\delta}_{k_A - \omega_A + 1} = k_B - \omega_B + 1$, we assume $\sum_{i = k_A - \omega_A + 1}^{k_A - 1} \tilde{\delta}_i = (\omega_A - 1) (k_B  - \omega_B + 1) - \lambda_2 = (\omega_A - 1) k_B - (\omega_A - 1)(\omega_B - 1) - \lambda_2$. Now, we consider the following three subcases.

{\it Case 2c (i)}: $\alpha \leq (\omega_A - 1) (\omega_B - 1) + \lambda_2$. Here, we are done by setting $\calU = \calU_{\alpha - 1}$ (Claim \ref{cl:excess}).  

{\it Case 2c (ii)}: $\alpha = \omega_A \omega_B - \omega_A + \lambda_2 - R_1$, where $1 \leq R_1 < \omega_B - 1$. Here, $\sum\limits_{i = 0}^{\alpha - 1} \rho_i +  \rho_{k_A - \omega_A}  = \omega_A k_B + (\alpha-1) k_B + \rho_{k_A - \omega_A}$, and, 
\begin{align*}
 \sum\limits_{i = 0}^{\alpha - 1} \tilde{\delta}_i +  \sum\limits_{i = k_A - \omega_A}^{k_A - 1 } \tilde{\delta_i} & = \alpha (k_B + 1) + \tilde{\delta}_{k_A - \omega_A} +  (k_B - \omega_B + 1) (\omega_A - 1) - \lambda_2 \\ & =  \omega_A k_B + (\alpha-1) k_B + \tilde{\delta}_{k_A - \omega_A} +  \alpha - \left\lbrace (\omega_A - 1) (\omega_B - 1) + \lambda_2 \right\rbrace .  
\end{align*} Since, $\tilde{\delta}_{k_A - \omega_A} = k_B - \omega_B + 1$, according to \eqref{eq:rhobigger}, $\rho_{k_A - \omega_A} - \tilde{\delta}_{k_A - \omega_A} \geq \omega_B - 2$. In addition, $\alpha - \left[ (\omega_A - 1) (\omega_B - 1) + \lambda_2\right] = \omega_B - 1 - R_1$. Thus, we are done by setting $\calU = \calU_{\alpha - 1} \cup \{ \tilde{\calM}_{k_A - \omega_A} \}$. 

{\it Case 2c (iii)}: $\alpha = \omega_A \omega_B - \omega_A + \lambda_2 + R_2$, where $0 \leq R_2 \leq \omega_A - 1 - \lambda_2$. Since we know, $\tilde{\delta}_{k_A - \omega_A} = k_B - \omega_B + 1$, we assume that $\tilde{\delta}_{k_A - \omega_A} = \tilde{\delta}_{k_A - \omega_A - 1} = \dots = \tilde{\delta}_{k_A - \omega_A - \beta + 1} = k_B - \omega_B + 1$, where $\beta \geq 1$. We define the set ${\calH} = \left\lbrace \tilde{\calM}_{k_A - \omega_A}, \tilde{\calM}_{k_A - \omega_A - 1}, \dots, \tilde{\calM}_{k_A - \omega_A - \beta + 1}  \right\rbrace$. 

Now, we consider two scenarios depending on the value of $\omega_B$. If $\omega_B \geq 3$, we assume that $\tilde{\delta}_i = k_B - \omega_B + 2 + \bar{\sigma}_i$ with $0 \leq \bar{\sigma}_i \leq \omega_B - 3$ where $\kappa \leq i \leq k_A -\omega_A - \beta$. Thus, in this scenario,
\begin{align}
\label{eq:2crho_3}
 \sum\limits_{i = 0}^{k_A - \omega_A} & \rho_i  =  \sum\limits_{i = 0}^{\alpha - 1} \rho_i +  \sum\limits_{i = \alpha}^{\kappa - 1} \rho_i +\sum\limits_{i = \kappa}^{k_A - \omega_A - \beta} \rho_i +  \sum\limits_{i = k_A - \omega_A - \beta + 1}^{k_A - \omega_A} \rho_i \nonumber \\  & = \alpha k_B + k_B (\omega_A - 1) + (\kappa - \alpha) k_B + (k_A - \omega_A - \beta - \kappa + 1) k_B + \sum\limits_{i = k_A - \omega_A - \beta + 1}^{k_A - \omega_A} \rho_i  
  %& = (k_A -\beta) k_B + \sum\limits_{i = k_A - \omega_A - \beta + 1}^{k_A - \omega_A} \rho_i .
  \end{align} In the remaining scenario where $\omega_B = 2$, we have $\tilde{\delta}_{k_A - \omega_A - \beta + 1} = k_B - 1$, thus $\tilde{\delta}_i = k_B$, when $\alpha \leq i \leq k_A - \omega_A - \beta$. Now, $ \sum\limits_{i = 0}^{k_A - \omega_A} \rho_i =  \sum\limits_{i = 0}^{\alpha - 1} \rho_i +  \sum\limits_{i = \alpha}^{k_A - \omega_A - \beta} \rho_i +  \sum\limits_{i = k_A - \omega_A - \beta + 1}^{k_A - \omega_A} \rho_i$, so,
\begin{align}
\label{eq:2crho_2}
  \sum\limits_{i = 0}^{k_A - \omega_A} \rho_i = \alpha k_B + k_B (\omega_A - 1) + (k_A - \omega_A - \beta -  \alpha + 1) k_B + \sum\limits_{i = k_A - \omega_A - \beta + 1}^{k_A - \omega_A} \rho_i  
 % & = (k_A -\beta) k_B + \sum\limits_{i = k_A - \omega_A - \beta + 1}^{k_A - \omega_A} \rho_i .
  \end{align} 
  \vspace{-0.5 cm}
\begin{align}
\label{eq:2crho}
\hspace*{-2cm} \textrm{Thus, from \eqref{eq:2crho_3} and \eqref{eq:2crho_2}, for any $\omega_B \geq 2$, } \; \; \sum\limits_{i = 0}^{k_A - \omega_A} \rho_i = (k_A -\beta) k_B + \sum\limits_{i = k_A - \omega_A - \beta + 1}^{k_A - \omega_A} \rho_i .
  \end{align}

Now, since $\lambda_2 + R_2 \leq \omega_A - 1$  and $\sum_{i = k_A - \omega_A + 1}^{k_A - 1} \tilde{\delta}_i = (\omega_A - 1) (k_B  - \omega_B + 1) - \lambda_2$, the number of $\tilde{\calM}_i$'s in $\calV$ having $\tilde{\delta}_i$ to be less than $k_B - \omega_B + 1$ is upper bounded by $\lambda_2$. 
Thus, the number of $\tilde{\calM}_i$'s left in $\calV$ with $\tilde{\delta}_i = k_B - \omega_B + 1$ is lower bounder by $\omega_A - 1 - \lambda_2$, we denote the set of such $\tilde{\calM}_i$'s as $\tilde{\calV}$. 
Besides, we have $\beta \geq 1$ more $\tilde{\delta}_i$'s with the same value, $k_B - \omega_B + 1$. Thus, 
\begin{align*}
    \tilde{\delta}_i = k_B - \omega_B + 1 \; \; \textrm{when} \; \; i = \undermat{\tilde{\calM}_i \in {\calH}}{k_A - \omega_A - \beta + 1, \dots, k_A - \omega_A}, \undermat{\tilde{\calM}_i \in \tilde{\calV}}{k_A - \omega_A + 1, \dots k_A - 1 - \lambda_2} .
\end{align*} 
\\

\vspace{-0.4 in} \noindent But, the total number of $\tilde{\calM}_i$'s with cardinality $k_B + 1$ is upper bounded by $\omega_A \omega_B - 1$. Since we have already taken account $\alpha = \omega_A \omega_B - \omega_A + \lambda_2 + R_2$ of such $\tilde{\calM}_i$'s, we have at most $\omega_A - 1 - \lambda_2 - R_2$ of $\tilde{\calM}_i$'s left which belong to $\tilde{\calU}_{\kappa} \cup \calV$. Thus,
we can have at most $\omega_A - 1 - \lambda_2 - R_2$ of such $\tilde{\calM}_i$'s in ${\calH} \cup \tilde{\calV}$, since ${\calH} \subseteq \tilde{\calU}_{\kappa}$ and $\tilde{\calV} \subseteq \calV$. 

Now, since $\tilde{\delta}_i = k_B -\omega_B + 1$ when $\tilde{\calM}_i \in {\calH} \cup \tilde{\calV}$, and $\abs*{\tilde{\calV}} \geq \omega_A - 1 - \lambda_2 - R_2$, according to our arrangement of $\tilde{\calM}_i$'s (in Sec. \ref{subsec:rearrange}(ii)), all such $\tilde{\calM}_i$'s (i.e., $|\tilde{\calM}_i| = k_B + 1$ and $\tilde{\delta}_i = k_B - \omega_B + 1$) belong to $\tilde{\calV}$. Therefore, $\abs*{\tilde{\calM}_{i}} = k_B$ when $\tilde{\calM}_i \in \calH$ and $\rho_{i} = \left( \omega_B + \tilde{\delta}_i - 1, k_B\right) = k_B$, since $\tilde{\delta}_i = k_B - \omega_B + 1$. Thus we are done using \eqref{eq:2crho}, since $\sum\limits_{i = 0}^{k_A - \omega_A} \rho_i  = k_A k_B \geq \sum\limits_{i = 0}^{k_A - 1} \tilde{\delta}_i $.

{\it Case 2d}: $\tilde{\delta}_{k_A - \omega_A + 1} \geq k_B - \omega_B + 2$. 
In this case, if $\kappa > k_A - \omega_A$, then we are done using  Claim \ref{cl:difference}(a), and, if $\kappa \leq k_A - \omega_A$, then we are done using  Claim \ref{cl:difference}(b) since $\tilde{\delta}_i \geq k_B - \omega_B + 2$ for all $i \leq k_A - \omega_A$. We can find $\kappa$ (from Claim \ref{cl:difference}) as the minimum value of $i$ when $\tilde{\delta}_i < k_B$.
\end{proof}
%In this case, we are done using  Claim \ref{cl:difference}, since all the $\tilde{\delta}_i$'s will be at least $k_B - \omega_B + 2$ when $i \leq k_A - \omega_A$. We can find $\kappa$ in Claim \ref{cl:difference} as the minimum value of $i$ when $\tilde{\delta}_i < k_B$.

\ifCLASSOPTIONcaptionsoff
  \newpage
\fi

\bibliographystyle{IEEEtran}
\bibliography{citations}
\end{document}